\documentclass[a4paper,12pt]{article}

\usepackage[scale=0.7275, centering, heightrounded]{geometry}
\usepackage{enumitem}
\usepackage{geometry}
\usepackage{setspace}
\usepackage{amssymb}
\usepackage{amsfonts}
\usepackage{amsmath}
\usepackage{amsthm}
\usepackage{amsbsy}

\usepackage{natbib}
\usepackage{hyperref}
\usepackage{bm}
\usepackage{tikz}
\usepackage{epsfig}
\usepackage{amsfonts}
\usepackage{color}
\newtheorem{lemma}{Lemma}
\newtheorem{theorem}{Theorem}

\newtheorem{proposition}{Proposition}

\newtheorem{remark}{Remark}
\newtheorem{example}{Example}

\renewcommand{\arraystretch}{1.5}

\title{Stable partitions for proportional generalized claims problems\thanks{We gratefully acknowledge financial support from the Swiss National Science Foundation (SNSF) through Project 100018$\_$192583. Oihane Gallo gratefully acknowledges additional financial support from the Spanish Government through projects PID2019-107539GB-I00 and PID2021-127119NBI00. We also thank Elena Inarra, Mehmet Karakaya, William Thomson, the associate editor, and two referees for helpful comments and suggestions.}}

\author{Oihane Gallo\thanks{Faculty of Business and Economics, University of Lausanne, 1015 Lausanne, Switzerland; \textit{e-mail}: \href{mailto:oihane.gallo@unil.ch}{\tt oihane.gallo@unil.ch}} \and Bettina Klaus\thanks{\textit{Corresponding author}. Faculty of Business and Economics, University of Lausanne, 1015 Lausanne, Switzerland; \textit{e-mail}: \href{mailto:Bettina.Klaus@unil.ch}{\tt Bettina.Klaus@unil.ch}.}}

\begin{document}
\sloppy
\date{\today}

\maketitle

\begin{abstract}
We consider a set of agents who have claims on an endowment that is not large enough to cover all claims. Agents can form coalitions but a minimal coalition size $\theta$ is required to have positive coalitional funding that is proportional to the sum of the claims of its members. We analyze the structure of stable partitions when coalition members use well-behaved rules to allocate coalitional endowments, e.g., the well-known constrained equal awards rule (CEA) or the constrained equal losses rule (CEL).
For continuous, (strictly) resource monotonic, and consistent rules, stable partitions with (mostly) $\theta$-size coalitions emerge. For CEA and CEL we provide algorithms to construct such a stable partition formed by (mostly) $\theta$-size coalitions.\smallskip

\noindent \textbf{Keywords:} claims problems, coalition formation, stable partitions.\smallskip

\noindent \textbf{JEL codes:} C71, C78, D63, D71, D74.
\end{abstract}

\pagebreak

\section{Introduction}

The formation of coalitions is a widespread aspect of social, economic, or political environments. Agents form coalitions in very different situations in order to achieve some joint benefits. Cooperation between agents is sometimes hampered by the existence of two opposing fundamental forces: on the one hand, the increasing returns to scale, which incentivizes agents to cooperate and, therefore, to form large coalitions and, on the other hand, the heterogeneity of agents, which causes instability and pushes towards the formation of only small coalitions. A partition of agents into coalitions is \textit{stable} if there is no set of agents that would want to form another coalition at which they all would be better off.\smallskip

\cite{gallo2018rationing} introduce \emph{generalized claims problems}\footnote{Note that these authors use the term \textit{coalition formation problems with claims} instead of \textit{generalized claims problems}.} to deal with coalition formation in a bankruptcy framework. A generalized claims problem consists of a group of agents, each of them with a claim, and a set of \textit{coalitional endowments}, one for each possible coalition, which are not sufficient to meet the claims of their members. Coalitional endowments are divided among their members according to a pre-specified rule, which thus is a decisive element of the coalition formation process. Their main result \citep[][Theorem~2]{gallo2018rationing} states that, given a generalized claims problem, there is a stable partition for each coalition formation problem that is induced by a continuous, resource monotonic, and consistent rule.\footnote{In fact, \citet[][Theorem~2]{gallo2018rationing} state that, given a generalized claims problem, there is a stable partition for each coalition formation problem that is induced by a continuous rule if and only if it also satisfies resource monotonicity and consistency. However, the only-if-part of the theorem is not derived from their proofs and, as a consequence, is not proven. In a recent paper, \cite{alcalde2023solidarity} deal with this issue and provide a correction. \label{footnote:GalloInarra}}
\citeauthor{pycia2012stability}, who considers a unified coalition formation model
that includes many-to-one matching problems with externalities, was the first in establishing
a link between consistency and stability \citep[][Footnote~5]{pycia2012stability}.
In this paper, we study the structure of stable partitions under different continuous, resource monotonic, and consistent rules to answer two types of questions: What coalition sizes can emerge? And, who are the coalition partners?\smallskip

The model proposed by \cite{gallo2018rationing} does not impose any restriction on coalitional endowments and, consequently, answering the above questions is not really possible in their general model. In contrast, we consider
$\mathit{\theta}$\emph{-minimal proportional generalized claims problems} where coalitions of size smaller than $\theta\in\mathbb{N}$ have zero endowments and all remaining coalitional endowments are a fixed proportion of the sum of their members' claims. Proportionality is justified in many situations such as the funding of research projects where the budgets are often divided proportionally to funding needs or according to other funding criteria such as project quality.\footnote{Other examples can be found in a bankruptcy situation, where assets have to be allocated proportionally among creditors according to their claims or, in a legislature, where seats are distributed proportionally among the parties according to voting shares.} Moreover, in many situations, institutions are interested in sparking cooperation among a critical mass of agents and hence, in discouraging coalitions that are ``too small'' \citep[see][for another model with a minimal size for coalitions to be productive]{barbera2015meritocracy}.\smallskip

$\theta$-minimal proportional generalized claims problems are a subclass of the class of generalized claims problems studied by \cite{gallo2018rationing} and hence, their results hold. Then, given a $\theta$-minimal proportional generalized claims problem, we first characterize the structure of any possible stable partition when the rule applied satisfies continuity, strict resource monotonicity, and consistency. We show that there are at most $\theta-1$ agents in coalitions of size smaller than $\theta$. Furthermore, for each coalition in the stable partition of size larger than $\theta$, each agent of the coalition receives a proportional payoff (Theorem~\ref{theorem:thetasize-strictRM}). If the rule satisfies resource monotonicity instead of its strict version, we show that a stable partition formed by the maximal possible number of coalitions of size $\theta$ and one coalition (of size smaller than $\theta$) formed by the remaining agents exists (Theorem~\ref{theorem:thetasize}).\smallskip

With the result of Theorem~\ref{theorem:thetasize} as the departure point, we analyze how agents sort themselves into $\theta$-size coalitions under some parametric rules \citep[see][]{young1987dividing,stovall2014asymmetric}. Parametric rules are well-studied in the literature because the payoff of each agent is given by a function that depends only on the claim of the agent and a parameter that is common to all agents. We focus on two well-known parametric rules that represent two egalitarian principles: the constrained equal awards rule  (CEA) and the constrained equal losses rule (CEL). On the one hand, CEA divides the endowment as equally as possible subject to no agent receiving more than her claim (e.g., rationing toilet paper when shortage occurs). On the other hand, CEL divides the losses as equally as possible subject to no agent receiving a negative amount (e.g., equal sacrifice taxation when utility is measured linearly\footnote{The idea of the equal sacrifice principle in taxation is that all tax payers end up sacrificing equally, according to some cardinal utility function. \cite{young1988distributive} provides a characterization of the family of equal-sacrifice rules based on a
few compelling principles and, more recently, \cite{chambers2017taxation} generalize the previous family.}).\smallskip

We propose two algorithms, one for each rule, to find a stable partition formed by (mostly) $\theta$-size coalitions. Under CEA, for didactic reasons, we first consider the case of $\theta=2$ and construct the CEA algorithm that sequentially pairs off either two highest-claims agents or a highest-claim with a lowest-claim agent (Theorem~\ref{theorem:CEA}). Examples of the first type of cooperation are found in social environments where agents tend to join other agents with similar characteristics (positively assortative coalition). In contrast, the second type of cooperation may be interpreted as a transfer of knowledge between agents as happens, for instance, between apprentices and advisors (negatively assortative coalition). For any $\theta>2$, the $\theta$-CEA algorithm generates a stable partition (mostly) formed by coalitions of size $\theta$ by sequentially adding a lowest-claim agent or a highest-claim agent (Theorem~\ref{theorem:CEAtheta}). The $\theta$-CEA algorithm leads to a stable partition that can contain positively assortative coalitions (in which only highest-claims agents are together) or ``mixed coalitions'' (in which lowest-claims as well as highest-claims agents are together). In contrast to the $\theta$-CEA algorithm, the $\theta$-CEL algorithm always produces a positively assortative stable partition, i.e., a partition formed by only positively assortative coalitions by sequentially pairing off $\theta$ lowest-claims agents (Theorem~\ref{theorem:CELtheta}). We would like to point out that, to compute a stable partition, our algorithms require a minimal amount of information: the agents' claims, the endowment of coalition $N$, and possibly, the payoffs of a few coalitions. In particular, it is not necessary to know agents' payoffs in all possible coalitions.\smallskip

Algorithms are often used in coalition formation to prove the existence of core stable partitions and to construct them. \cite{farrell1988partnerships} show the existence of core stable partitions by an algorithmic argument. In a similar way, \cite{BanerjeeKonishiSonmezSCW2001} offer a constructive proof to find the unique core partition whenever the top-coalition property is satisfied. Based on fixed-point methods, \cite{inal2015core} provides an algorithm to find all core stable partitions in a coalition formation setting where not all coalitions are permissible.\footnote{This algorithm is an adaptation of the one introduced in \cite{echenique2007solution} in a matching setting.} More recently, \cite{inal2019existence} constructs an algorithm that finds unique core partitions in those coalition formation problems that satisfy the layered top-coalitions condition.\footnote{This algorithm is an extension of the one in \cite{BanerjeeKonishiSonmezSCW2001}.} \cite{gallo2018rationing} also applies an algorithmic procedure in the proof of their Theorem~1. In the above papers, general hedonic coalition formation problems without an underlying claims structure are considered. On the one hand, this leads to results for more general classes of hedonic coalition formation problems and, on the other hand, agents' complete preferences over coalitions have to be available. In contrast, while restricting attention to a class of hedonic coalition formation problems that are based on claims problems and rules, our algorithms require very little input to produce a stable partition (essentially, only the underlying claims problem and a claims rule is needed); in particular, we show that it is not necessary to know agents' preferences over all coalitions.\smallskip

Next, there is a large number of papers that pay attention to the structure of the coalitions that form. \cite{Becker} and \cite{greenberg1986strong} introduce the notion of assortative coalitions\footnote{Assortativeness is based on an ordering of agents according to a specific variable such as claims, productivity, or location. Alternative terminology includes that of consecutive coalitions.} and, in a matching setting, \cite{shimer2000assortative} introduce notions of positively and negatively assortative matchings. Observe that in both our algorithms assortative coalitions (in terms of claims) may form. The following papers obtain similar results concerning assortative stable coalitions.\smallskip

\cite{barbera2015meritocracy} consider a model in which each agent is endowed with a productivity level and agents can cooperate to perform certain tasks. Each coalition generates an output equal to the sum of its members' productivities. The authors then analyze the formation of coalitions when all agents in a society vote between meritocracy and egalitarianism. They find societies where assortative and non-assortative partitions (in terms of productivity) arise.\smallskip

In a bargaining framework, \cite{pycia2012stability} presents a model in which each agent has a utility function and, for each possible coalition of agents, there is an output to be distributed among its members. He analyzes coalition formation problems induced by different bargaining rules and shows that when agents are endowed with productivity levels and ``when shares are divided by a stability-inducing sharing rule, agents sort themselves into coalitions in a predictably assortative way.'' \cite{pycia2012stability} deals with many-to-one problems and his notion of assortativeness implies that the most productive agents join the most productive firms.\smallskip

Finally, \cite{bogomolnaia2008stability} study societies where agents are located in an interval and form jurisdictions to consume public projects, which are located in the same interval. Agents share their costs equally and they divide transportation costs to the location of the public project based on its distance to each agent. They analyze both core and Nash stable partitions with a focus on assortative and non-assortative (in terms of location) stable jurisdiction structures.\smallskip

Our paper is organized as follows. In Section~\ref{model} we introduce our model. First, we illustrate the class of coalition formation problems we consider with an example (Subsection~\ref{ex1}). Then, in Subsections~\ref{secModel} to \ref{model:final}, we introduce all ingredients needed to define the class of generalized claims problems of \cite{gallo2018rationing}. In particular, we introduce the proportional, the CEA, and the CEL rules and their key properties (continuity, (strict) resource monotonicity, and consistency).
The class of proportional generalized claims problems is introduced in Section~\ref{secResults}. This section also introduces the subclass of $\theta$-minimal proportional generalized claims problems and obtains first results concerning the existence and structure of stable partitions if the underlying rule is continuous, (strictly) resource monotonic, and consistent (Theorems~\ref{theorem:thetasize-strictRM} and \ref{theorem:thetasize}). Sections~\ref{secCEA} and \ref{secCEL} establish the main results of our paper (Theorems~\ref{theorem:CEA}, \ref{theorem:CEAtheta}, and \ref{theorem:CELtheta}): algorithms to construct stable partitions for $\theta$-minimal proportional generalized claims problems under the CEA rule and the CEL rule. We conclude in Section~\ref{secConclusion}, where we also discuss how our results extend to coalition formation problems induced by resource allocation situations that allow for excess supply and where agents have single-peaked preferences.

\section{The Model}\label{model}

\subsection{An Example to illustrate the model and results}\label{ex1}

There are three researchers $1$, $2$, and $3$. They can form research groups to conduct joint projects that can be submitted for funding. Each researcher has a claim that, for instance, could express her research expertise or aspiration.  Consider the following (exogenously given) claims $c_1=2$, $c_2=6$, and $c_3=22$. We further assume that each researcher can be member of only one funding coalition and that the (expected) funding that each coalition can receive, the endowments, are given by Table~\ref{tableEx1endowments}.

\begin{table}[th]\label{table1}
\centering
\begin{tabular}{c||c|c|c|c|c}

Coalition & $\{1\},\ \{2\},\ \{3\}$ & $\{1,2\}$ & $\{1,3\}$ & $\{2,3\}$ & $\{1,2,3\}$
\\ \hline\hline
Endowment & $0$  & $4$ & $12$ & $14$ & $15$ \\
\end{tabular}%
\caption{Coalitions and their endowments.}
\label{tableEx1endowments}
\end{table}
\pagebreak

Thus, for instance, coalition $\{1,2,3\}$, which claims $30$, has an endowment of $15$ and single researchers do not receive any funding. More generally, we consider that (i) the funding is not sufficient to meet the claims of all researchers in their respective coalitions and (ii) the expected coalitional funding, when positive, is a fixed proportion of the sum of the claims of the researchers in each coalition. In our example, the expected funding proportion is $\frac{1}{2}$, which might come from the fact that it is known that 50\% of total claims for research projects can be awarded. These two characteristics are part of our model.\medskip

The endowment of each coalition is divided according to a sharing rule. Let us consider two different rules:
\begin{itemize}
	 \item[R1:]The funding is divided equally among the members of the coalition, assuming that no researcher receives more than her claim.
  \item[R2:]The loss (i.e., the difference between the sum of the claims and the funding) is divided equally among the members of the coalition, assuming that no researcher receives a negative amount.
\end{itemize}

We provide the formal definition of these two rules in Section~\ref{model:generalizedclaimsproblems}; here, we explain how to compute the payoffs of coalition $\{1,3\}$ (the others are similar).
\begin{itemize}
\item[R1:]The endowment of coalition $\{1,3\}$ is $12$; so the per capita funding is $\frac{12}{2}=6$ and equally sharing the funding yields $6>2=c_1$ for agent~$1$ and $6<22=c_3$ for agent~$3$. Since agents cannot receive amounts larger than their claims, agent~$1$ receives $2$. By feasibility, agent $3$ now gets $12-2=10$.
\item[R2:]Coalition $\{1,3\}$ claims $c_1+c_3=24$ and their endowment is $12$; so the loss is $24-12=12$. Hence, the per capita loss is $\frac{12}{2}=6$ and equally sharing the loss yields $2-6=-4<0$ for agent~$1$ and $22-12=10>0$ for agent~$3$. Since agents cannot receive negative amounts, agent~$1$ receives $0$. By feasibility, agent $3$ now gets $12-0=12$.
\end{itemize}

Payoffs for each researcher in each of their possible coalitions are listed in Table~\ref{tableEx1payoffs}.\medskip

\begin{table}[htb]
\centering
\begin{tabular}{c||c|c|c|c|c}

Coalition & $\{1\}, \{2\},\ \{3\}$ & $\{1,2\}$ & $\{1,3\}$ & $\{2,3\}$ & $\{1,2,3\}$
\\ \hline\hline
 R1 & $(0)$ & $(2,2)$ & $(2,10)$ & $(6,8)$ & $(2,6,7)$ \\
  R2  & $(0)$ & $(0,4)$ & $(0,12)$ & $(0,14)$ & $(0,0,15)$\\
\end{tabular}%
\caption{Researchers' coalitional payoffs.}
\label{tableEx1payoffs}
\end{table}
\pagebreak
Assuming that higher payoffs are preferred, researchers can now rank all the coalitions they belong to and the following preferences over coalitions emerge.\medskip

Under rule R1, researchers have the following preferences over coalitions.
\begin{align*}
\succsim_1^{R1}:&\quad\{1,2\}\sim \{1,3\}\sim\{1,2,3\}\succ \{1\},\\
\succsim_2^{R1}:&\quad\{2,3\}\sim \{1,2,3\}\succ\{1,2\}\succ\{2\},\\
\succsim_3^{R1}:&\quad\{1,3\}\succ\{2,3\}\succ\{1,2,3\}\succ\{3\}.
\end{align*}

Under rule R2, researchers have the following preferences over coalitions.
\begin{align*}
\succsim_1^{R2}:&\quad\{1,2\}\sim \{1,3\}\sim\{1,2,3\}\sim \{1\},\\
\succsim_2^{R2}:&\quad\{1,2\}\succ \{2,3\}\sim\{1,2,3\}\sim\{2\},\\
\succsim_3^{R2}:&\quad\{1,2,3\}\succ\{2,3\}\succ\{1,3\}\succ\{3\}.
\end{align*}

Based on the researchers' preferences, which coalitions will form? If researchers take into account that they will not want to deviate to strictly better coalitions ex post, \textit{stability} emerges as a central solution concept. A partition of agents into coalitions is \textit{stable} if there is no set of agents that would want to form another coalition at which they all would be strictly better off.\medskip

Then,  under rule R1, partitions $\{\{1,3\},\{2\}\}$ and $\{\{1,2,3\}\}$ are stable, while under rule R2 any partition is stable.\medskip

There are several questions concerning stability we may be interested in.

\begin{itemize}
\item Given any number of researchers, any vector of claims, and any values of the endowments, do the defined rules always induce preferences with stable partitions?
\item Are there other rules that induce preferences with stable partitions?
\end{itemize}

The answer to both questions is ``Yes'', as has been established in \citet{gallo2018rationing}. In this paper, we address the following questions.
\begin{itemize}

\item Is there a particular structure that all stable partitions have?

\item How do agents sort themselves into coalitions in a stable manner?

\item Is there a constructive procedure (algorithm) for finding stable partitions?

\end{itemize}

Our example already indicates that a general structure for all stable partitions may be hard to find, at least for the two rules we consider. However, we show that, when considering rules that additionally satisfy a strict monotonicity property (strict resource monotonicity as defined in the next section), a general structure for all stable partitions emerges (Theorem~\ref{theorem:thetasize-strictRM}).
However, even without requiring strict monotonicity, we can establish the existence of a stable partition that has a certain structure with respect to coalition sizes (Theorem~\ref{theorem:thetasize}).

Our main results emerge when answering the last two questions: we propose step-by-step methods for rules R1 and R2, respectively, and show that these algorithms always produce stable partitions (Theorems~\ref{theorem:CEAtheta} and \ref{theorem:CELtheta}).

\subsection{Preliminaries}\label{secModel}

Consider a \textbf{coalition of agents}, e.g., a group of researchers, who have claims on an \textbf{endowment}, e.g., a research budget  from a national science foundation (NSF). The researchers' \textbf{claims} could be related to the past performance / productivity of researchers or be an estimate of the research costs. The research budget is not large enough to cover all claims. Now assume that the NSF, apart from subsidizing individual researchers, can also subsidize research groups and that then researchers will need to allocate the research funding within the research teams. Furthermore, anticipating this method of allocating funding, researchers might prefer to be members of certain research teams over others. This situation was analyzed by \citet{gallo2018rationing} under the name of \textit{coalition formation problem with claims}. Before fully specifying this class of problems, we present the preliminaries of the two classical type of problems it is based on: \textbf{claims problems} and \textbf{coalition formation problems}.\medskip

First, we introduce some notation. Let $\mathbb{N}$ be the set of potential agents and $\mathcal{N}$ the set of all non-empty finite subsets or \textbf{coalitions} of $\mathbb{N}$. Given $N\in \mathcal{N}$ and $x,y\in\mathbb{R}^{N}$, if for each $i\in N$, $x_{i}> y_{i}$, then $x\gg y$ and if for each $i\in N$, $x_{i}\geq y_{i}$, then $x\geq y$. Furthermore, for each $x\in\mathbb{R}^{N}$ and each $S\subseteq N$, $x_S:=(x_j)_{j\in S}$ denotes the \textbf{restriction of $\bm{x}$ to coalition $\bm{S}$}.

\subsection{Generalized claims problems}\label{model:generalizedclaimsproblems}

Consider a coalition of agents who have claims on a certain endowment, this endowment being insufficient to satisfy all the claims. A primary example is bankruptcy, where agents are the creditors of a firm and the endowment is its liquidation value; however, we have a more general interpretation of the data in mind.\medskip

Formally, let $N\in \mathcal{N}$. For $i\in N$, let $c_{i}\in \mathbb{R}_{++}$ be agent $i$'s claim and $c=(c_{j})_{j\in N}$ the \textbf{claims vector}. Let $E\in \mathbb{R}_{+}$ be the \textbf{endowment}. A \textbf{claims problem with coalition} $\bm{N}$ is a pair $(c,E)\in \mathbb{R}^{N}_{++}\times \mathbb{R}_{+}$ such that $\sum_{j\in N}c_{j}\geq E$. Let $\bm{\mathcal{C}^{N}}$ denote the class of such problems and  $\bm{\mathcal{C}}\equiv\bigcup_{N\in\mathcal{N}} \mathcal{C}^{N}$.\medskip

An \textbf{allocation for} $\bm{(c,E)}\in \mathcal{C}^{N}$ is a (payoff) vector $x=(x_{j})_{j\in N}\in \mathbb{R}^{N}_{+}$ that satisfies the non-negativity and claims boundedness conditions $0\leq x\leq c$, and the efficiency condition $\sum_{j\in N}x_{j}=E$. A \textbf{rule} is a function $F$ defined on $\mathcal{C}$ that associates with each $N\in\mathcal{N}$ and each $(c,E)\in \mathcal{C}^{N}$ an allocation for $(c,E)$. Let $\bm{\mathcal{F}}$ denote the set of rules.\medskip

A rule is continuous if small changes in the data of the problem do not lead to large changes in the chosen allocation.\medskip

\noindent\textbf{Continuity.} For each $N\in \mathcal{N}$, each $(c,E)\in \mathcal{C}^{N}$, and each sequence $\{(c^{k},E^{k})\}$ of elements of $\mathcal{C}^{N}$, if $(c^{k},E^{k})$ converges to $(c,E)$ then $F(c^{k},E^{k})$ converges to $F(c,E)$.\medskip

Consider a claims problem and the allocation given by the rule for it. We require that if the endowment increases, then each agent should receive at least as much as (more than, respectively) initially.\medskip

\noindent\textbf{Resource monotonicity.} For each $N\in \mathcal{N}$, each $(c,E)\in \mathcal{C}^{N}$, and each $E^{\prime}>E$, if $\sum_{j\in N} c_{j}\geq E^{\prime}$, then $F(c,E^{\prime})\geq F(c,E)$.\bigskip
	
\noindent\textbf{Strict resource monotonicity.} For each $N\in \mathcal{N}$, each $(c,E)\in \mathcal{C}^{N}$, and each $E^{\prime}>E$, if $\sum_{j\in N} c_{j}\geq E^{\prime}$, then $F(c,E^{\prime})\gg F(c,E)$.\medskip

Consider a claims problem and the allocation given by the rule for it. Imagine that some agents leave with their payoffs. At that point, reassess the situation of the remaining agents, that is, consider the problem of dividing what remains of the endowment among them. Then, each agent should receive the same payoffs as initially.\medskip

\noindent\textbf{Consistency.} For each $N\in \mathcal{N}$, each $(c,E)\in \mathcal{C}^{N}$, and each $S\subsetneq N$, let $x\equiv F(c,E)$ and recall that  $c_{S}=(c_i)_{i\in S}$. Then, $x_{S}=F(c_{S}, E-\sum_{j\in N\backslash S}x_{j})$ or, equivalently, $x_{S}=F(c_{S},\sum_{j\in S}x_{j})$.\medskip

For claims problems many rules are continuous, resource monotonic, and consistent: the most important ones are the so-called ``parametric'' rules (we define them in Appendix~\ref{appendix:parametricrules}). \citet{thomson2003axiomatic,thomson2015axiomatic,thomson2019divide} defines and discusses various symmetric parametric rules (the proportional, constrained equal awards, constrained equal losses, Talmud, reverse Talmud, and Piniles rules) as well as the asymmetric sequential priority rule associated with a strict priority $\succ$ on agents. In particular, all these rules are continuous, resource monotonic, and consistent.\medskip

We now define three well-known parametric rules that represent various fairness notions.\medskip

The most commonly used rule in practice makes awards proportional to claims.\medskip

\noindent \textbf{Proportional rule, $\bm{P}$.} For each $N\in \mathcal{N}$, each $(c,E)\in\mathcal{C}^N$, and each $i\in N$, $P_{i}(c,E)=\lambda c_{i}$, where $\lambda$ is chosen so that $\sum_{j\in N}\lambda c_j = E$.\medskip

Our next rule assigns the endowment as equally as possible among agents subject to no one receiving more than her claim.\medskip

\noindent \textbf{Constrained equal awards rule, $\bm{CEA}$.} For each $N\in \mathcal{N}$, each $(c,E)\in\mathcal{C}^N$, and each $i\in N$, $CEA_{i}(c,E)=\min\{c_i,\lambda\}$, where $\lambda$ is chosen so that $\sum_{j\in N}\min\{c_j,\lambda\}= E$.\medskip

An alternative to the constrained equal awards rule is obtained by focusing on the losses agents incur (what they do not receive), as opposed to what they receive, and to assign losses as equally as possible among agents subject to no one receiving a negative amount.\medskip

\noindent \textbf{Constrained equal losses rule, $\bm{CEL}$.} For each $N\in \mathcal{N}$, each $(c,E)\in\mathcal{C}^N$, and each $i\in N$, $CEL_{i}(c,E)=\max\{0,{c_i-\lambda}\}$, where $\lambda$ is chosen so that $\sum_{j\in N}{\max\{0,c_j-\lambda\}}= E$.\medskip

Next, we generalize the notion of a claims problem. Consider $N\in \mathcal{N}$ and $(c,E)\in\mathcal{C}^{N}$. Then, each coalition of agents $S \subseteq N$ has the \textbf{reduced claims vector} $\bm{c_{S}}=(c_i)_{i\in S}$. Next, assume that for each coalition $S \subseteq N$, there is a \textbf{coalitional endowment} $\bm{E_{S}}$ such that $(c_{S},E_{S})\in\mathcal{C}^{S}$ and $E_{N}=E$. Formally, given $N\in \mathcal{N}$, a \textbf{generalized claims problem with coalition} $\bm{N}$ is a pair $(c,(E_{S})_{S\subseteq N})\in \mathbb{R}_{++}^{N}\times \mathbb{R}_{+}^{2^{|N|}-1}$, such that for each coalition $S\subseteq N$, $(c_{S},E_{S})\in\mathcal{C}^{S}$. Let $\bm{\mathcal{G}^{N}}$ denote the class of such problems and $\bm{\mathcal{G}}\equiv\bigcup_{N\in\mathcal{N}} \mathcal{G}^{N}$.\medskip

We study the following subclass of generalized claims problems. First, for each coalition $S\subseteq N$, we define a \textbf{coalitional claim}, $\bm{c^{S}}$, that is equal to the sum of the claims of the members of the coalition, i.e., $c^{S}:=\sum_{j\in S}c_j$. Then, given $(c,E)\in\mathcal{C}^{N}$ and $\alpha:=\frac{E}{c^{N}}\in(0,1)$, a \textbf{proportional generalized claims problem} is a tuple $(c,(E_{S})_{S\subseteq N})$ such that, for each coalition $S\subseteq N$, the coalitional endowment $E_{S}=\alpha c^{S}$. Let $\bm{\mathcal{P}^{N}}$ denote the class of such problems and $\bm{\mathcal{P}}\equiv\bigcup_{N\in\mathcal{N}} \mathcal{P}^{N}$. Since  coalitional endowments $E_{S}$, $S\subseteq N$, are completely determined by $c$ and $E$, we will simplify notation and denote a proportional generalized claims problem $(c,(E_{S})_{S\subseteq N})\in \mathcal{P}^N$ by $(c,E)\in \mathcal{P}^N$.\medskip

An \textbf{allocation configuration for} $\bm{(c,E)}\in\mathcal{P}^{N}$ is a list $(x_{S})_{S\subseteq N}$ such that for each $S\subseteq N$, $x_{S}$ is an allocation for the claims problem derived from $(c,E)$ for coalition $S$, $(c_{S},E_{S})$.\footnote{The notion of allocation (payoff) configurations was introduced by \cite{hart1985axiomatization} in his characterization of the Harsanyi non-transferable utility solution.} Any rule $F\in \mathcal{F}$ can be extended to a \textbf{generalized rule} defined on $\mathcal{P}$  by associating with each $N\in\mathcal{N}$ and each $(c,E)\in \mathcal{P}^{N}$ an allocation configuration $F(c,E)=(F(c_{S},E_{S}))_{S\subseteq N}$. Since it should not lead to any confusion we use $\bm{\mathcal{F}}$ to also denote the set of generalized rules.

\subsection{Coalition formation problems}\label{model:coalitionformationproblems}

Consider a society where each agent can rank the coalitions that she may belong to. Some well-known examples of such problems are matching problems, in particular, marriage and roommate problems.\medskip

Formally, let $N \in \mathcal{N}$. For each agent $i\in N$, $\succsim_{i}$ is a complete and transitive preference relation over the coalitions of $N$ containing $i$. Given $S,S'\subseteq N$ such that $i\in S\cap S'$, $S\succsim_{i}S'$ means that agent $i$ finds coalition $S$ at least as desirable as coalition $S'$. Let $\mathcal{R}_{i}$ be the set of such preference relations for agent $i$ and $\mathcal{R}^{N}\equiv \Pi_{i\in N} \mathcal{R}_{i}$. A \textbf{coalition formation problem with agent set} $\bm{N}$ consists of a list of preference relations, one for each $i\in N$, $\bm{\succsim}=(\succsim_{i})_{i\in N}\in \mathcal{R}^{N}$. Let $\bm{\mathcal{D}^{N}}$ be the class of such problems and $\bm{\mathcal{D}}\equiv\bigcup_{N\in\mathcal{N}}\mathcal{D}^{N}$.\medskip

A partition of a set of agents $N\in\mathcal{N}$ is a set of disjoint coalitions of $N$ whose union is $N$ and whose pairwise intersections are empty. Formally, a \textbf{partition} of $N\in\mathcal{N}$ is a list $\pi=\{S_{1},\dots ,S_{m}\}$, ($m\leq |N|$ is a positive integer) such that (i) for each $l=1,\dots ,m$, $S_{l}\neq \emptyset ,$ (ii) $\bigcup_{l=1}^{m}S_{l}=N$, and (iii) for each pair $l,l'\in \{1,\dots ,m\}$ with $l\neq l'$, $S_{l}\cap S_{l'}=\emptyset $. Let $\bm{\Pi(N)}$ denote the set of all partitions of $N$. For each $\pi \in \Pi(N)$ and each $i\in N$, let $\bm{\pi(i)}$ denote the unique coalition in $\pi$ that contains agent $i$. We refer to the partition $\bm{\pi^0}$ at which each agent forms a singleton coalition, i.e., for each $i\in N$, $\pi^0(i)=\{i\}$, as the \textbf{singleton partition}.\medskip

An important question for coalition formation problems is the existence of partitions from which no agent wants to deviate. Let $\succsim \in \mathcal{D}^{N}$ and consider a partition $\pi \in \Pi(N)$. Then, coalition $T\subseteq N$ \textbf{blocks} $\pi$ if for each agent $i\in T$, $T \succ_{i}\pi(i)$. A partition $\pi \in \Pi(N)$ is \textbf{stable} for $\succsim$ if it is not blocked by any coalition $T\subseteq N$. Let $\bm{St(\succsim)}$ denote the set of all stable partitions of $\succsim$. We refer to the coalitions that are part of a stable partition as \textbf{stable coalitions}.

\subsection{Coalition formation problems induced by a generalized claims problem and a rule}\label{model:final}

Given a generalized claims problem and a rule, each agent, by calculating her payoff in each coalition, can form preferences over coalitions, giving rise to a coalition formation problem. Examples of this situation can be the formation of jurisdictions and research groups \citep[see][for further examples]{gallo2018rationing}.\medskip

Formally, given $N\in \mathcal{N}$, $(c,(E_{S})_{S\subseteq N})\in \mathcal{G}^{N}$, and $F\in\mathcal{F}$, the \textbf{coalition formation problem induced by} $\bm{((c,(E_{S})_{S\subseteq N}),F)}$\footnote{Note that \cite{gallo2018rationing} refer to this as \textit{coalition formation problems with claims}.} consists of the list of preferences $\bm{\succsim^{((c,E),F)}}=\left(\succsim_{j}\right)_{j\in N}^{((c,E),F)}$ defined as follows: for each $i\in N$, and each pair $S,S'\subseteq N$ such that $i\in S\cap S'$, $S\succsim_{i}^{((c,E),F)}S^{\prime }$ if and only if $F_{i}(c_{S},E_{S})\geq F_{i}(c_{S'},E_{S'})$. Let $\bm{St(\succsim^{((c,E),F)})}$ denote the set of all stable partitions of $\succsim^{((c,E),F)}$.\medskip

\citet[][Proposition~1]{gallo2018rationing} show that, given any generalized claims problem, parametric rules always induce coalition formation problems that have stable partitions. More generally, they show that, given a generalized claims problem, there is a stable partition for each coalition formation problem that is induced by a continuous, resource monotonic, and consistent rule (see Footnote~\ref{footnote:GalloInarra} in the Introduction).

\section{Preliminary results}\label{secResults}

\subsection{Stable partitions for proportional generalized claims problems}\label{secproportionalproblems}

We first restrict attention to the class of proportional generalized claims problems $\mathcal{P}$. Thus, consider $N\in \mathcal{N}$, $(c,E)\in \mathcal{P}^{N}$, and $F\in\mathcal{F}$. Even without any further assumptions on rule $F$, the coalition formation problem induced by $((c,E),F)$ has a stable partition. Due to the assumption that the generalized claims problems we consider are proportional, the singleton partition is always stable and all stable partitions are payoff equivalent to the singleton partition.

\begin{proposition}\label{proposition:singletons}Let $N\in \mathcal{N}$, $(c,E)\in \mathcal{P}^{N}$, and $F\in\mathcal{F}$. Then, for the coalition formation problem with agent set $N$ induced by $((c,E),F)$, the singleton partition $\pi^0$ is stable and each stable partition $\pi$ induces the proportional allocation configuration where each agent $i\in N$ receives $\alpha c_i$.
\end{proposition}

\begin{proof}[\textbf{Proof}]Let $N\in \mathcal{N}$, $(c,E)\in \mathcal{P}^{N}$, and $F\in\mathcal{F}$. Consider the coalition formation problem with agent set $N$ induced by $((c,E),F)$.

First, assume, by contradiction, that there exists a stable partition $\pi$ that does not induce the proportional allocation configuration. Hence, there exists some agent $i\in N$ who receives an under-proportional payoff $F_i(c_{\pi(i)},E_{\pi(i)})< \alpha c_i$. However, since $(c,E)\in \mathcal{P}^{N}$, agent $i$ by forming the singleton coalition $\{i\}$ obtains $\alpha c_i>F_i(c_{\pi(i)},E_{\pi(i)})$ and blocks $\pi$; contradicting the stability of $\pi$. Thus, each stable partition $\pi$ induces the proportional allocation configuration.

Second, consider the singleton partition $\pi^0$ at which each agent $i\in N$ obtains $\alpha c_i$. Since $(c,E)\in \mathcal{P}^{N}$, for each $S\subseteq N$, $E_{S}=\alpha c^{S}$. Thus, no coalition $S\subseteq N$ can achieve payoffs larger than $\alpha c_i$ for its members $i\in S$ and block $\pi^0$. Hence, $\pi^0$ is stable.\end{proof}

\subsection[theta--minimal proportional generalized claims problems]{$\bm{\theta}$-minimal proportional generalized claims problems}

Proposition~\ref{proposition:singletons} illustrates that for proportional generalized claims problems, essentially only the ``trivial'' singleton partition is stable. For our motivating example, the allocation of research funding to research teams, Proposition~\ref{proposition:singletons} predicts the following consequences. If the main principle of research funding allocation from the funding agency to teams is proportionality, then, since essentially only individual proportional funding is stable, the formation of larger research collaborations is unlikely. However, many scientific funding schemes are aimed at the promotion of cooperation of researchers from different countries or disciplines and require research teams of at least size two; see, for instance the international programs of the Swiss National Sciences Foundation (\href{http://www.snf.ch}{http://www.snf.ch}), the synergy grants of the European Research Council (\href{https://erc.europa.eu/funding/synergy-grants}{https://erc.europa.eu/funding/synergy-grants}), or the European Union funded COST (European Cooperation in Science and Technology) actions (\href{https://www.cost.eu/}{https://www.cost.eu/}) that requires at least seven participating countries.\medskip

We take this as the departure point to modify the class of proportional generalized claims problems such that only coalitions of a specified minimal size have an incentive to form. Hence, we consider the following adjustment of the class of proportional generalized claims problems where coalitions of size smaller than $\theta\in \mathbb{N}$ are disincentivized by assuming that all coalitions of size smaller than $\theta$ receive a zero coalitional endowment. Hence, given $N\in \mathcal{N}$, $|N|\geq \theta$, and $(c,(E_{S})_{S\subseteq N})\in \mathcal{G}^{N}$, for each $S\subseteq N$ such that $|S|\geq \theta$, $E_{S}=\alpha c^{S}$ (where $\alpha:=\frac{E}{c^{N}}\in(0,1)$) and for each $S\subseteq N$, $S\neq\emptyset$, such that $|S|<\theta$, $E_{S}=0$. Let $\bm{\mathcal{P}_{\theta}^{N}}$ be the class of such problems. We refer to this subclass of generalized claims problems as $\bm{\theta}$\textbf{-minimal proportional generalized claims problems} and denote it by $\bm{\mathcal{P}_{\theta}}\equiv\bigcup_{N\in\mathcal{N}} \mathcal{P}_{\theta}^{N}$. Since coalitional endowments $E_{S}$, $S\subseteq N$, are completely determined by $\theta$, $c$ and $E$, we simplify notation and denote a $\theta$-minimal proportional generalized claims problem $(c,(E_{S})_{S\subseteq N})\in \mathcal{P}_{\theta}^N$ by $(c,E)\in \mathcal{P}_{\theta}^N$.\medskip

We first analyze the structure of stable partitions for coalition formation problems that are induced by $\theta$-minimal proportional generalized claims problems if the underlying rule is continuous, strictly resource monotonic, and consistent. By \citet{gallo2018rationing} the set of stable partitions is non-empty. Furthermore, for each stable partition, there are at most $\theta-1$ agents in coalitions of size smaller than $\theta$,\footnote{Thus, there are at most $\theta-1$ coalitions of size smaller than $\theta$.} and proportional payoffs are allocated in any coalition of size larger than $\theta$. Hence, in comparison to the benchmark result of proportional stable sharing that is equivalent to the singleton partition, non-proportional cooperation can be sustained in a stable way in $\theta$-size coalitions, i.e., research teams formed by $\theta$ researchers.

\begin{theorem}\label{theorem:thetasize-strictRM}
Let $N\in \mathcal{N}$ and $(c,E)\in \mathcal{P}_{\theta}^{N}$. Consider $F\in{\cal F}$ satisfying continuity, strict resource monotonicity, and consistency. Then, for the coalition formation problem with agent set $N$ induced by $((c,E),F)$, the set of stable partitions is non-empty and each stable partition $\pi$ is such that
\begin{itemize}[nosep]
  \item[\emph{(i)}] there are at most $\theta-1$ agents in coalitions of size smaller than $\theta$, and
  \item[\emph{(ii)}]if $S\in\pi$ is such that $|S|>\theta$, then for all $i\in S$, $F_i(c_{S},E_{S})=\alpha c_{i}$.
\end{itemize}
\end{theorem}

Theorem~\ref{theorem:thetasize-strictRM} imposes restrictions on the number of coalitions of size smaller that $\theta$ and on the payoffs in coalitions of size larger than $\theta$; however, it does not impose any restriction on either the total number of coalitions or the payoffs in coalitions of size~$\theta$.

\begin{proof}[\textbf{Proof}] Let $N$, $(c,E)$, and $F$ be as specified in the theorem. First, by \citet{gallo2018rationing}, each coalition formation problem with agent set $N$ induced by $((c,E),F)$ has at least one stable partition. \smallskip

\noindent (i) Let $\pi\in St(\succsim^{((c,E),F)})$ and $T\equiv\{i\in N:|\pi(i)|<\theta\}$ be the set of agents in coalitions of size smaller than $\theta$. Assume, by contradiction, that $|T|\geq\theta$.
For the trivial claims problem $(c_{T},0)\in\mathcal{C}^{{T}}$, we have that for each $i\in T$, $F_i(c_{T},0)=0$. Then, since $E_{T}>0$, by strict resource monotonicity, for each $i\in T$, $F_i(c_{T},E_{T})>0$. Given that $(c,E)\in \mathcal{P}_{\theta}^{N}$, by the definition of $T$, for each $i\in T$, we have $E_{\pi(i)}=0$. Thus, coalition $T$ blocks $\pi$, which is a contradiction.\smallskip

\noindent (ii) Assume, by contradiction, that  $\pi\in St(\succsim^{((c,E),F)})$ is such that
$S\in\pi$, $|S|>\theta$, and for some $i\in S=\pi(i)$, $F_i(c_{S},E_{S})\neq\alpha c_{i}$. Without loss of generality, assume that agent $i$ receives an over-proportional payoff $F_i(c_{S},E_{S})> \alpha c_i$. By consistency, for each $j\in S\setminus\{i\}$, $$F_j\left(c_{S\setminus\{i\}},\sum\nolimits_{k\in S\setminus\{i\}}F_k(c_{S},E_{S})\right)=F_j(c_{S},E_{S}).$$ Consider problem $(c_{S\setminus\{i\}},E_{S\setminus\{i\}})\in\mathcal{C}^{S\setminus\{i\}}$. Since $E_{S\setminus\{i\}}=\alpha c^{ S\setminus\{i\}}>\sum_{k\in S\setminus\{i\}}F_k(c_{S},E_{S})$, the agents in coalition $S\setminus\{i\}$ have a larger endowment to share among themselves if they get rid of agent $i$. Then, by strict resource monotonicity, for each $j\in S\setminus\{i\}$, $$F_j\left(c_{S\setminus\{i\}},E_{S\setminus\{i\}}\right)>F_j\left(c_{S\setminus\{i\}},\sum\nolimits_{k\in S\setminus\{i\}}F_k(c_{S},E_{S})\right).$$ Hence, for each $j\in S\setminus\{i\}$, $$F_j\left(c_{S\setminus\{i\}},E_{S\setminus\{i\}}\right)> F_j(c_{S},E_{S})$$ and coalition $S\setminus\{i\}$ blocks $\pi$, which is a contradiction.\end{proof}

We now weaken the property of strict resource monotonicity to resource monotonicity and show that among all possible stable partitions that can exist in a coalition formation problem induced by a $\theta$-minimal proportional generalized claims problem, there is always one stable partition that is composed of the maximal possible number of coalitions of size $\theta$ and one coalition (of size smaller than $\theta$) formed by the remaining agents.\footnote{For convenience, we put all ``remaining agents'' together in one coalition. However, note that all coalitions formed among remaining agents have sizes smaller than $\theta$ and thus all payoffs are zero. Hence, we could have payoff-equivalently partitioned this set of remaining agents differently.}

\begin{theorem}\label{theorem:thetasize}
Let $N\in \mathcal{N}$ and $(c,E)\in \mathcal{P}_{\theta}^{N}$. Consider $F\in{\cal F}$ satisfying continuity, resource monotonicity, and consistency. Then, for the coalition formation problem with agent set $N$ induced by $((c,E),F)$, there exists a stable partition $\pi$ such that
\begin{itemize}[nosep]
  \item[\emph{(i)}]if $|N|$ is divisible by $\theta$, then for each $i\in N$, $|\pi(i)|=\theta$ and
  \item[\emph{(ii)}]if $|N|$ is not divisible by $\theta$, then $\pi$ contains $k=\lfloor\frac{|N|}{\theta}\rfloor$ coalitions of size $\theta$ and one coalition of size $|N|-k\,\theta<\theta$.
\end{itemize}
\end{theorem}

Recall that in Theorem~\ref{theorem:thetasize-strictRM}, we use strict resource monotonicity and, for \emph{all} stable partitions, we obtain a lot of structure concerning coalition sizes and payoffs. However, by weakening the monotonicity property in Theorem~\ref{theorem:thetasize-strictRM}, we lose the general structure for all stable partitions and, in Theorem~\ref{theorem:thetasize}, we can only establish the existence of a stable partition with a specific coalitional size configuration (the CEA (R1) and CEL (R2) examples in Section~\ref{ex1} show that stable partitions with coalitions of sizes larger than $\theta$ are possible).\medskip

We prove Theorem~\ref{theorem:thetasize} in Appendix~\ref{appendix:proofTheorem6}.\medskip

Now, taking Theorem~\ref{theorem:thetasize} as our point of departure, we focus on stable coalitions of size $\theta$ that belong to a stable partition. Note that Theorem~\ref{theorem:thetasize} does not give information about how agents sort themselves into those stable $\theta$-size coalitions. In fact, the stable partitions may differ depending on how endowments are divided among agents. In the following sections, we will study the specific structure of stable partitions under two egalitarian parametric rules: the constrained equal awards (CEA) and the constrained equal losses (CEL) rule. Egalitarianism is a natural principle applied in many economic environments and hence, studying CEA and CEL in our context seems a natural first step to understand stable partitions.

\section{Stability under the constrained equal awards rule}\label{secCEA}

\subsection[Analysis of the case theta equals 2]{Analysis of the case $\bm{\theta=2}$}

For didactic reasons, we first analyze how agents organize themselves into stable pairs, i.e., $\theta=2$, when endowments are distributed under the constrained equal awards rule, CEA. Recall that this rule divides endowments as equally as possible, subject to no agent receiving more than her claim. Hence, under CEA, agents who get their claims receive over-proportional payoffs (more so the smaller the claims are) while some agents with higher claims receive under-proportional payoffs (more so the larger the claims are). Furthermore, the higher an agent's claim, the higher her contribution towards the endowment of any coalition she is part of. So intuitively, in order to form a stable pair, one could suspect that an agent with a very high claim will pair up with another high-claim agent. Indeed, high-claim agents play a special role in our construction of stable pairs.\medskip

Let $N=\{1,\dots, n\}$ and assume that $c_1\leq c_2\leq \cdots\leq c_n$. First, consider a highest-claim agent, e.g., agent $n$. As explained above, pair $\{n-1,n\}$ would be a contender to be part of a stable partition with the following possible justification: agent $n$ needs to team up with some agent to obtain a positive payoff and agent $n-1$ provides the highest possible contribution to pair $\{n-1,n\}$ while requiring a smaller transfer from the proportional payoff compared to other agents. This reasoning is correct, unless a lowest-claim agent, e.g., agent $1$, has such a small claim that conceding this small claim to agent $1$ implies a smaller loss from the proportional payoff for agent $n$ than the transfer from the proportional payoff of $n$ to agent $n-1$. We capture this intuition in an algorithm that determines a stable partition by sequentially pairing off either two highest-claims agents or a highest-claim with a lowest-claim agent. Note that if $|N|= 2$, then the grand coalition $N$ forms a stable partition. We therefore define the algorithm for agent sets with at least three agents.\bigskip

\noindent\textbf{CEA algorithm}\medskip

\noindent\textbf{Input:} $N\in{\cal N}$ such that $|N|>2$ and $(c,E)\in{\cal P}_{2}^{N}$.\medskip

\noindent\textbf{Step~$\bm{1}$.} Let $N_1:=N$ and $$\delta_1:=\frac{2c_1}{2c_1 -c_{n-1} +c_n}.$$

We distinguish two cases:

\begin{itemize}[nosep]
	\item[(i)]If $\alpha\leq\delta_1$, then $\{n-1,n\}\succsim^{((c,E),CEA)}_n \{1,n\}$ and set $S_1:=\{n-1,n\}$.
	\item[(ii)] If $\alpha>\delta_1$, then $\{1,n\}\succ^{((c,E),CEA)}_n \{n-1,n\}$ and set $S_1:=\{1,n\}$.
\end{itemize}
Set $N_2:=N\setminus S_1$. If $|N_2|\leq 2$, then set $S_2:=N_2$, define $\pi:=\{S_1,S_2\}$, and stop. Otherwise, go to Step~2.\medskip

\noindent\textbf{Step~$\bm{k}$ ($\bm{k>1}$).} Recall from Step~$k-1$ that $N_k:=N\setminus\left(\cup_{j=1}^{k-1}S_{j}\right)$ and $|N_k|> 2$.  We relabel the agents in $N_k$ such that $N_k=\{1',\ldots,n'\}$ and $c_{1'}\leq\ldots\leq c_{n'}$. Let $$\delta_k:=\frac{2c_{1'}}{2c_{1'} -c_{(n-1)'} +c_{n'}}.$$ We distinguish two cases:
\begin{itemize}[nosep]
	\item[(i)]If $\alpha\leq\delta_k$, then $\{(n-1)',n'\}\succsim^{((c,E),CEA)}_{n'} \{1',n'\}$ and set $S_k:=\{(n-1)',n'\}$.
	\item[(ii)] If $\alpha>\delta_k$, then $\{1',n'\}\succ^{((c,E),CEA)}_{n'} \{(n-1)',n'\}$ and set $S_k:=\{1',n'\}$.
\end{itemize}
Set $N_{k+1}:=N\setminus \cup_{j=1}^{k}S_j$. If $|N_{k+1}|\leq 2$, then set $S_{k+1}:=N_{k+1}$, define $\pi:=\{S_1,\ldots,S_{k+1}\}$, and stop. Otherwise, go to Step~$k+1$.\medskip

\noindent\textbf{Output:} A partition $\pi=\{S_1,\ldots,S_l\}$ for the coalition formation problem with agent set $N$ induced by $((c,E),CEA)$ such that for each $k\in\{1,\dots,l-1\}$, $|S_{k}|=2$, and $|S_{l}|\leq 2$. If $|N|$ is even, then partition $\pi$ is constructed in $l-1=\frac{n-2}{2}$ steps. If $|N|$ is odd, then partition $\pi$ is constructed in $l-1=\frac{n-1}{2}$ steps.\medskip

Note that the CEA algorithm only needs the agents' claims and the endowment of coalition $N$ as input. In particular, it is not necessary to first derive agents' coalitional payoffs and preferences over coalitions.
The following result states that the partition obtained by the CEA algorithm is stable.

\begin{theorem}\label{theorem:CEA}
Let $N\in{\cal N}$ such that $|N|>2$ and $(c,E)\in{\cal P}_{2}^{N}$. Consider the coalition formation problem with agent set $N$ induced by $((c,E),CEA)$. Then, the partition obtained by the CEA algorithm is stable.
\end{theorem}

We prove Theorem~\ref{theorem:CEA} in Appendix~\ref{appendix:proofCEA} but we explain the key intuition of the proof here. Observe that, in each step of the CEA algorithm, either two highest-claims agents or a lowest-claim agent and a highest-claim agent are paired off. We first show (Appendix~\ref{appendix:proofCEA}, Lemma~\ref{lemma:top}) that each agent weakly prefers to form a pair with a highest-claim agent, e.g., with agent~$n$, instead of with any other agent. Hence, by matching agent~$n$ with her most desirable partner, a stable pair is formed.\medskip

Second, we show (Appendix~\ref{appendix:proofCEA}, Lemmas~\ref{corollary:alpha-very-small} and \ref{corollary:alpha-very-large}) that agents $1$ and $n-1$ are potential ``stable partners'' for agent~$n$. Furthermore,  we show that when $\alpha\leq \frac{2c_{1}}{c_{1}+c_{n}}$, then $\{n-1,n\}$ is the candidate for a stable pair, and when $\alpha\geq \frac{2c_{n-1}}{c_{n-1}+c_{n}}$, then $\{1,n\}$ is the candidate for a stable pair. Note that $\frac{2c_{1}}{c_{1}+c_{n}}\leq \frac{2c_{n-1}}{c_{n-1}+c_{n}}$. Thus, we next need to determine a threshold value for parameter $\alpha \in [\frac{2c_{1}}{c_{1}+c_{n}},\frac{2c_{n-1}}{c_{n-1}+c_{n}}]$ to see when agent~$n$'s partner of choice is $n-1$ and when it is $1$. We then show that the threshold we are looking for is exactly $\frac{2c_1}{2c_1 -c_{n-1} +c_n}$ (Appendix~\ref{appendix:proofCEA}, Lemma~\ref{lemma:threshold}), the value specified at Step~1 of the CEA algorithm to trigger either Case~(i) with stable pair candidate $\{n-1,n\}$ or Case~(ii) with stable pair candidate $\{1,n\}$. Using the steps of the CEA algorithm, we then show that the resulting CEA partition is stable.\medskip

Finally, note that the results in Appendix~\ref{appendix:proofCEA} (Lemmas~\ref{corollary:alpha-very-small}, \ref{corollary:alpha-very-large}, and \ref{lemma:threshold}) can be used to show that for values of $\alpha$ low enough to trigger Case~(i) in each step of the CEA algorithm, starting with highest-claim agent $n$, a stable partition is formed by positively assortative pairs, i.e., pairs formed by contiguous agents. This implies that if $n$ is odd, then the singleton coalition will be formed by agent~1. Similarly, for values of $\alpha$ high enough to trigger Case (ii) in each step of the CEA algorithm, a stable partition is formed by negatively assortative pairs, i.e., pairs formed by ``opposed'' agents. This implies that if $n$ is odd, then the singleton coalition will be formed by agent~$\frac{n+1}{2}$. So, depending on $\alpha$, the CEA algorithm constructs a stable partition that is either positively assortative, negatively assortative, or formed by both types of pairs.\medskip

Positively assortative matching of high types has been observed in other contexts, for instance, the neoclassical marriage model by \citet{Becker}. In our example of research team formations, positively assortative matching of high types can indeed be observed in practice but it can also be observed in other situations such as the formation of pairs of students for class projects or other social environments. However, we also observe negatively assortative research team formations, for example in mentor-mentee relationships such as between a PhD student and her advisor.\medskip

Next, we illustrate how the CEA algorithm works.

\begin{example}[\textbf{CEA algorithm}]\normalfont
Let $N=\{1,2,3,4,5\}$, $c=(2,6,22,30,34)$, $E=47$, and $(c,E)\in{\cal P}_{2}^{N}$ (hence, $\alpha=\frac{1}{2}$). With Theorem~\ref{theorem:thetasize} as our point of departure, we focus on coalitions of size $2$.\medskip

\noindent\textbf{Step~$\bm{1}$.} Let $N_1=N$ and $\delta_1=\frac{2c_1}{2c_1 -c_{4} +c_5}=\frac{2\times 2}{2\times 2 - 30+34}=\frac{1}{2}.$\medskip

Since $\alpha=\frac{1}{2}=\delta_1$,  by Case (i), we select pair $\{4,5\}$.\medskip

\noindent\textbf{Step~$\bm{2}$.} Let $N_2=N_1\setminus\{4,5\}=\{1,2,3\}$ and $\delta_2=\frac{2c_1}{2c_1 -c_{2} +c_3}=\frac{2\times 2}{2\times 2 - 6+22}=\frac{1}{5}.$\medskip

Since $\alpha=\frac{1}{2}>\frac{1}{5}=\delta_2$, by Case (ii), we select pair $\{1,3\}$.\medskip

 Now, we define $N_3=N\setminus\{1,3,4,5\}=\{2\}$. Since $|N_{3}|< 2$, we select $\{2\}$ and we obtain $$\pi=\{\{4,5\},\{1,3\},\{2\}\}.$$

Note that this partition is formed by a positively assortative pair, $\{4,5\}$, a negatively assortative pair, $\{1,3\}$, and the remaining singleton, $\{2\}$. Note that once coalition $\{4,5\}$ is formed at Step~1, agents $1$ and $3$ are ``opposed'' agents in the new subset of agents $N_2.$\medskip

Consider now the same set of agents and claims with $E=9.4$ (hence, $\alpha=\frac{1}{10}$). We apply the algorithm again.\medskip

\noindent\textbf{Step~$\bm{1}$.} Let $N_1=N$ and $\delta_1=\frac{2c_1}{2c_1 -c_{4} +c_5}=\frac{2\times 2}{2\times 2 - 30+34}=\frac{1}{2}.$\medskip

Since $\alpha=\frac{1}{10}<\frac{1}{2}=\delta_1$, by Case (i), we select pair $\{4,5\}$.\medskip

\noindent\textbf{Step~$\bm{2}$.} Let $N_2=N_1\setminus\{4,5\}=\{1,2,3\}$ and $\delta_2=\frac{2c_1}{2c_1 -c_{2} +c_3}=\frac{2\times 2}{2\times 2 - 6+22}=\frac{1}{5}.$\medskip

Since $\alpha=\frac{1}{10}<\frac{1}{5}=\delta_2$, by Case (i), we select pair $\{2,3\}$.\medskip

Now, we define $N_3=N\setminus\{2,3,4,5\}=\{1\}$. Since $|N_{3}|< 2$, we select $\{2\}$ and we obtain $$\pi=\{\{4,5\},\{2,3\},\{1\}\}.$$

In this case, both pairs of the partition are positively assortative.
More generally,  for ``small'' values of $\alpha$ ($\alpha\leq\frac{2c_1}{2c_1-c_2+c_3}=\frac{1}{5}$), a stable partition formed by only positively assortative pairs is formed. Similarly, for ``high'' values of $\alpha$ ($\alpha\geq\frac{2c_2}{2c_2-c_3+c_4}=\frac{3}{5}$) a stable partition formed by only negatively assortative pairs is formed. For values of $\alpha$ in $[\frac{1}{5},\frac{3}{5}]$, stable partitions formed by both types of pairs (positively and negatively assortative) arise. Theorem~\ref{theorem:CEA} proves that partition $\pi$ is stable. We would like to point out that the algorithm finds a stable partition with very little information: only the vector of claims and the endowment of coalition $N$ (which gives us the value of $\alpha$) are required.\hfill $\diamond$\label{example:CEAalgorithm}
\end{example}

Finally, the stable partition obtained by the CEA algorithm is not necessarily unique (even beyond tie-breaking between Cases (i) and (ii) in the algorithm); the CEA (R1) example in Section~\ref{ex1} shows that stable partitions with larger coalition sizes are possible.

\subsection{Analysis of the general case}

For $\theta=2$, the CEA algorithm finds a stable partition by stepwise joining a highest-claim agent with either another highest-claim agent or a lowest-claim agent. Therefore, one could suspect that a similar algorithm might work for any $\theta\in \mathbb{N}$.\medskip

Let $N=\{1,\dots, n\}$ and assume that $c_1\leq c_2\leq \cdots\leq c_n$. For each $S\subseteq N$, we denote the CEA parameter associated with $(c_S,E_S)$ by $\lambda_{E_S}$, i.e., for each $i\in S$, $CEA_{i}(c_S,E_S)=\min\{c_i,\lambda_{E_S}\}$, where $\lambda_{E_S}$ is chosen so that $\sum_{j\in S}\min\{c_j,\lambda_{E_S}\}= E_S$.\medskip

First, assume that $\theta=2$ and consider a highest-claim agent, e.g., agent $n$. By the CEA algorithm, we know that agents $1$ and $n-1$ are the possible ``stable partners'' for agent $n$. Hence, either coalition $\{1,n\}$ or coalition $\{n-1,n\}$ would form. Consider now $\theta=3$ and the coalition formed for $\theta - 1 = 2$ (either coalition $\{1,n\}$ or coalition $\{n-1,n\}$), which now plays the role of agent $n$. Then, the next agent to join will again be either a lowest-claim agent (agents $2$ or $1$, respectively) or a highest-claim agent (agents $n-1$ or $n-2$, respectively). By sequentially adding lowest-claim or highest-claim agents, step by step, we construct a first coalition of size $\theta$. We now more formally present the algorithm to determine the first coalition of a stable partition by sequentially adding either a lowest-claim agent or a highest-claim agent. Note that if $|N|= \theta$, then the grand coalition $N$ forms a stable partition. We therefore define the algorithm for agent sets with at least $\theta+1$ agents. \bigskip\pagebreak

\noindent$\bm{\theta}$-\textbf{CEA set algorithm with agent set $\bm{N}$}\medskip

\noindent\textbf{Input:} $N\in{\cal N}$ such that $|N|>\theta$ and $(c,E)\in{\cal P}_{\theta}^{N}$.\medskip

\noindent\textbf{Step~$\bm{1}$.} Let  $\theta'=2$, $N'_1:=N$, and consider coalition $\{n-1,n\}$ and $\lambda_{E_{\{n-1,n\}}}$. We distinguish two cases:
\begin{itemize}[nosep]
	\item[(i)]If $\lambda_{E_{\{n-1,n\}}}\leq (1-\alpha)c_1 +\alpha c_{n-1}$, then $\{n-1,n\}\succsim^{((c,E),CEA)}_{n}\{1,n\}$ and set $S'_{1}:=\{n-1,n\}$.
\item[(ii)]If $\lambda_{E_{\{n-1,n\}}}> (1-\alpha)c_1 +\alpha c_{n-1}$, then $\{1,n\}\succ^{((c,E),CEA)}_{n}\{n-1,n\}$ and set $S'_{1}:=\{1,n\}$.
\end{itemize}
Note that these cases are equivalent to the cases specified at Step~$1$ of the CEA algorithm.\footnote{By the definition of the CEA rule, $\lambda_{E_{\{n-1,n\}}}\geq\frac{\alpha (c_{n-1}+c_{n})}{2}$. In addition, $\frac{\alpha (c_{n-1}+c_{n})}{2}\lessgtr (1-\alpha)c_1 +\alpha c_{n-1}$ is equivalent to $\alpha \lessgtr \frac{2c_1}{2c_1 -c_{n-1} +c_n}=\delta_{1}$. Hence, if $\lambda_{E_{\{n-1,n\}}}=\frac{\alpha (c_{n-1}+c_{n})}{2}$, we obtain the equivalence between Cases (i) and (ii) here and Cases (i) and (ii) in the CEA algorithm. Finally, if $\lambda_{E_{\{n-1,n\}}}>\frac{\alpha (c_{n-1}+c_{n})}{2}$, then $\lambda_{E_{\{n-1,n\}}}>c_{n-1}=(1-\alpha)c_{n-1}+\alpha c_{n-1}\geq (1-\alpha)c_{1}+\alpha c_{n-1}$ and Case (ii) applies. Therefore, $\lambda_{E_{\{n-1,n\}}}=\alpha (c_{n-1}+c_n)-c_{n-1}> \frac{\alpha (c_{n-1}+c_{n})}{2}$, which implies $\alpha \geq \frac{2c_{n-1}}{c_{n-1}+c_n}\geq \frac{2c_{1}}{2c_1-c_{n-1}+c_n}$ (the last inequality is shown as $\gamma_1\geq \delta_1$ at the end of Appendix~\ref{appendix:proofCEA}), which coincides with Case (ii) in the CEA algorithm.}\medskip

\noindent If $\theta=2$, then set $S_1:=S'_1$ and stop. Otherwise, set $N'_{2}:=N\setminus S'_1$, and go to Step~$2$.\medskip

\noindent\textbf{Step~$\bm{k}$ ($\bm{k=2,\ldots,\theta-1}$).} Consider $\theta'=k+1$. Recall from Step~$k-1$ that $N'_k:=N\setminus S'_{k-1}$. Let agent $i$ be the agent with the lowest label in $N_{k}'$, $i=\min N'_{k}$, and let agent $j$ be the agent with the highest label in $N'_{k}$, $j=\max N'_k$. Thus, agent $i$ is a lowest-claim agent and agent $j$ is a highest-claim agent in $N'_{k}$. Consider coalition $S'_{k-1}\cup\{j\}$ and $\lambda_{E_{(S'_{k-1}\cup\{j\})}}$. We distinguish two cases:
\begin{itemize}[nosep]
	\item[(i)]If $\lambda_{E_{(S'_{k-1}\cup\{j\})}}\leq (1-\alpha)c_i +\alpha c_j$, then $S'_{k-1}\cup\{j\}\succsim^{((c,E),CEA)}_{S'_{k-1}}S'_{k-1}\cup\{i\}$ and set $S'_{k}:=S'_{k-1}\cup\{j\}$.
\item[(ii)]If $\lambda_{E_{(S'_{k-1}\cup\{j\})}}> (1-\alpha)c_i +\alpha c_j$, then $S'_{k-1}\cup\{i\}\succ^{((c,E),CEA)}_{S'_{k-1}}S'_{k-1}\cup\{j\}$ and set $S'_{k}:=S_{k-1}\cup\{i\}$.
\end{itemize}
If $\theta=k+1$, then set $S_{1}:=S'_{k}$ and stop. Otherwise, set $N'_{k+1}:=N\setminus S'_k$, and go to Step~$k+1$.\medskip

\noindent\textbf{Output:} A coalition $S_{1}$ of size $\theta$, which is obtained in $\theta-1$ steps.\medskip

Note that in each Step~$k$ of the above algorithm, a set of agents $S'_{k-1}$ considers to add either the lowest-label or the highest-label agent of the remaining set of agents. In Appendix~\ref{appendix:conditionalgorithm}, we prove that the agent who is added according to Cases (i) or (ii) is weakly preferred by the agents in $S'_{k-1}$ to any other agent who could have been added.\medskip

The $\theta$-CEA set algorithm with agent set $N$ generalizes to any agent set with at least $\theta+1$ agents; the thus obtained $\theta$-CEA set algorithm is a subroutine of the general $\theta$-CEA algorithm that we explain next.\medskip

Starting with the set of all agents $N$, the above algorithm constructs the first coalition $S_1$ of a stable partition. Then, if the residual set of agents $N_1 = N\setminus S_1$ is such that $|N_1|\leq \theta$, $(S_1,N_1)$ forms a stable partition. Otherwise, we apply the $\theta$-CEA set algorithm to the set of agents $N_1$. Recall that if $|N|= \theta$, then the grand coalition $N$ forms a stable partition. We therefore define our generalization of the CEA algorithm for agent sets with at least $\theta+1$ agents.\bigskip

\noindent$\bm{\theta}$-\textbf{CEA algorithm}\medskip

\noindent\textbf{Input:} $N\in{\cal N}$  such that $|N|>\theta$ and $(c,E)\in{\cal P}_{\theta}^{N}$.\medskip

\noindent\textbf{Step~$\bm{1}$.} Let $N_1:=N$. By applying the $\theta$-CEA set algorithm with agent set $N_1$ we obtain the first coalition $S_1$.  Set $N_2:=N\setminus S_1$. If $|N_2|\leq \theta$, then set $S_2:=N_2$, define $\pi:=\{S_1,S_2\}$, and stop. Otherwise, go to Step~2.\medskip

\noindent\textbf{Step~$\bm{k}$ ($\bm{k>1}$).} Recall from Step~$k-1$ that $N_k:=N\setminus\left(\cup_{j=1}^{k-1}S_{j}\right)$ and $|N_k|> \theta$. We relabel the agents in $N_k$ such that $N_k=\{1',\ldots,n'\}$ and $c_{1'}\leq\ldots\leq c_{n'}$. By applying the $\theta$-CEA set algorithm with agent set $N_k$, we obtain coalition $S_{k}$. Set $N_{k+1}:=N\setminus \cup_{j=1}^{k}S_j$. If $|N_{k+1}|\leq \theta$, then set $S_{k+1}:=N_{k+1}$, define $\pi:=\{S_1,\ldots,S_{k+1}\}$, and stop. Otherwise, go to Step~$k+1$.\medskip

\noindent\textbf{Output:} A partition $\pi=\{S_1,\ldots,S_l\}$ for the coalition formation problem with agent set $N$ induced by $((c,E),CEA)$ such that for each $k\in\{1,\dots,l-1\}$, $|S_{k}|= \theta$ and $|S_{l}|\leq \theta$. If $|N|$ is divisible by $\theta$, then partition $\pi$ is constructed in $l-1=\frac{n-\theta}{\theta}$ steps. If $|N|$ is not divisible by $\theta$, then partition $\pi$ is constructed in $l-1=\lfloor\frac{n}{\theta}\rfloor$ steps.\medskip

Note that the $\theta$-CEA algorithm only needs the agents' claims, the endowment of coalition $N$, and some few coalitional payoffs (and the associated $\lambda$ parameters) as input. In particular, it is not necessary to first derive agents' coalitional payoffs and preferences over all coalitions. The following result states that the partition obtained by the $\theta$-CEA algorithm is stable.

\begin{theorem}\label{theorem:CEAtheta}
Let $N\in{\cal N}$  such that $|N|>\theta$ and $(c,E)\in{\cal P}_{\theta}^{N}$. Consider the coalition formation problem with agent set $N$ induced by $((c,E),CEA)$. Then, the partition obtained by the $\theta$-CEA algorithm is stable.
\end{theorem}

We prove Theorem~\ref{theorem:CEAtheta} in Appendix~\ref{appendix:proofTheorem7} but we explain the key intuition of the proof here. The $\theta$-CEA set algorithm describes a process of adding agents one by one such that the respective existing coalition weakly prefers the added agent to any other agent that could be added. This process already captures an aspect of coalitional stability but it is ``myopic'' or ``local'' in nature and we still need to show that the thus obtained coalitions are ``farsightedly'' or ``globally'' stable as well. In other words, we have to prove that the greedy one-by-one addition of agents in the $\theta$-CEA set algorithm leads to a stable partition.

To see this, we consider the partition $\pi$ constructed by the $\theta$-CEA algorithm and assume, by contradiction, that a blocking coalition $S$ exists. As part of the proof, we then show that an agent who received an over-proportional payoff in her coalition in partition $\pi$ will receive an even larger over-proportional payoff in $S$. Furthermore, an agent who received an under-proportional payoff in her coalition in partition $\pi$ will receive an over-proportional payoff in $S$ or transfer less to other agents in $S$. This, together with some additional proof steps, leads to an imbalance between the transfers that are being made within coalition $S$ between agents with over-proportional and under-proportional payoffs. Note that when $\theta=2$, the $\theta$-CEA algorithm reduces to the CEA algorithm and an alternative proof can be applied.\medskip

Next, we illustrate how the $\theta$-CEA algorithm works for $\theta = 3$.

\begin{example}[\textbf{3-CEA algorithm}]\normalfont
Let $N=\{1,2,3,4,5,6,7\}$, $c=(2,6,22,30,34,38,46)$, $E=89$, and $\theta = 3$, i.e., $(c,E)\in{\cal P}_{3}^{N}$ (hence, $\alpha=\frac{1}{2}$). With Theorem~\ref{theorem:thetasize} as our point of departure, we focus on coalitions of size $3$.\medskip

\noindent \textbf{3-CEA algorithm, Step~$\bm{1}$.} Let $N_1:=N$.\medskip

\noindent\textbf{3-CEA set algorithm, Step~$\bm{1.1}$.} Let  $\theta'=2$ and consider coalition $\{6,7\}$. By applying the CEA rule to this coalition, we obtain  $CEA((38,46), 42)=(21,21)$ and $\lambda_{ \{6,7\}}=21$. Since
\begin{itemize}[nosep]
	\item[(ii)]$\lambda_{ \{6,7\}}=21>20=\frac{1}{2}\times 2+\frac{1}{2}\times 38=(1-\alpha)c_1 +\alpha c_6,$
\end{itemize}
we select coalition $\{1,7\}$. Given that $\theta\neq 2$, we go to Step~2 of the $3$-CEA set algorithm.\medskip

\noindent\textbf{3-CEA set algorithm, Step~$\bm{1.2}$.}  Let $\theta'=3$ and $N'_2=N'_1\setminus\{1,7\}=\{2,3,4,5,6\}$. Then, $i=2$ and $j=6$ and we consider coalition $\{1,6,7\}.$ By applying the CEA rule to this coalition, we obtain $CEA((2,38,46), 43)=(2, 20.5,20.5)$ and $\lambda_{ \{1,6,7\}}=20.5$. Since
\begin{itemize}[nosep]
	\item[(i)]$\lambda_{ \{1,6,7\}}=20.5<22=\frac{1}{2}\times 6+\frac{1}{2}\times 38=(1-\alpha)c_2 +\alpha c_6,$
\end{itemize}
we select coalition $\{1,6,7\}$. Given that $\theta=3$, we stop and set $S_1=\{1,6,7\}$.
Setting $N_2=N_1\setminus\{1,6,7\}=\{2,3,4,5\}$ completes Step~1 of the 3-CEA algorithm. Since $|N_2|>3$, we continue with \medskip

\noindent \textbf{3-CEA algorithm, Step~$\bm{2}$.} $N_2=\{2,3,4,5\}$.

\noindent\textbf{3-CEA set algorithm, Step~$\bm{2.1}$.} Let  $\theta'=2$ and consider coalition $\{4,5\}$. By applying the CEA rule to this coalition, we obtain  $CEA((30,34), 32)=(16,16)$ and $\lambda_{ \{4,5\}}=16$. Since
\begin{itemize}[nosep]
	\item[(i)]$\lambda_{ \{4,5\}}=16<18=\frac{1}{2}\times 6+\frac{1}{2}\times 30=(1-\alpha)c_2 +\alpha c_4,$
\end{itemize}
we select coalition $\{4,5\}$. Given that $\theta\neq 2$, we go to Step~2 of the $3$-CEA set algorithm.\medskip

\noindent\textbf{3-CEA set algorithm, Step~$\bm{2.2}$.}  Let $\theta'=3$ and $N'_2=N'_1\setminus\{4,5\}=\{2,3\}$.

Then, $i=2$ and $j=3$ and we consider coalition $\{3,4,5\}$. By applying the CEA rule to this coalition, we obtain $CEA((22,30,34), 43)=(\frac{43}{3},\frac{43}{3},\frac{43}{3})$ and $\lambda_{ \{3,4,5\}}=\frac{43}{3}=14.33$. Since
\begin{itemize}[nosep]
	\item[(ii)]$\lambda_{ \{3,4,5\}}=14.33>14=\frac{1}{2}\times 6+\frac{1}{2}\times 22=(1-\alpha)c_2 +\alpha c_3,$
\end{itemize}
we select coalition $\{2,4,5\}$. Given that $\theta=3$, we stop and set $S_2=\{2,4,5\}$. Setting $N_3:=N\setminus\{1,2,4,5,6,7\}=\{3\}$ completes Step~2 of the 3-CEA algorithm. Since $|N_{3}|< 3$, we select $\{3\}$ and we obtain $$\pi=\{\{1,6,7\},\{2,4,5\},\{3\}\}.$$

In this case, none of the coalitions are purely assortative. However, in each step either a lowest-claim agent or a highest-claim agent is added. Theorem~\ref{theorem:CEAtheta} proves that partition $\pi$ is stable. The 3-CEA algorithm finds a stable partition without computing the payoffs of each agent in each coalition. It just needs to know the vector of claims, the endowment of coalition $N$ (which gives us the value of $\alpha$), and the payoffs of a few coalitions, which give us the $\lambda$ parameters needed to decide which agent is added throughout the $3$-CEA set algorithm.\hfill $\diamond$
\end{example}

As mentioned before, the CEA (R1) example in Section~\ref{ex1} shows (for $\theta=2$) that stable partitions with larger coalition sizes than the stable partition obtained by the $\theta$-CEA algorithm are possible.

\section{Stability under the constrained equal losses rule}\label{secCEL}

We analyze how agents organize themselves when endowments are distributed under the constrained equal losses rule, CEL. Recall that this rule allocates losses as equally as possible subject to no agent receiving a negative amount. Hence, under CEL, the agents who get a zero payoff (if they exist) receive under-proportional payoffs (more so the smaller the claims are) while some agents with higher claims receive over-proportional payoffs (more so the larger the claims are). Furthermore, the lower an agent's claim, the lower her contribution towards the loss of any coalition she is part of. For the simplest case of $\theta=2$, intuitively, in order to form a stable pair, one could suspect that an agent with a very low claim will pair up with another low-claim agent. Similarly, one could suspect that a similar result could arise for any $\theta\in\mathbb{N}$. That is, intuitively, in order to form a stable coalition of size $\theta$, one could suspect that an agent with a very low claim will pair up with $\theta-1$ low-claim agents. So in contrast to the CEA rule, when the CEL rule is used, low-claim agents play a special role in our construction of stable coalitions of size $\theta$.\medskip

Let $N=\{1,\dots, n\}$ and assume that $c_1\leq c_2\leq \cdots\leq c_n$. First, consider a lowest-claim agent, e.g., agent$~1$. As explained above, agent~$1$ in a coalition with agents~$2,\dots,\theta$ would be a contender to be part of a stable partition with the following possible justification: agent~$1$ needs to team up with some agents to obtain a positive payoff and agents~$2,\ldots,\theta$ provide the lowest possible loss to coalition $\{1,2,\ldots,\theta\}$ while requiring a smaller transfer from the proportional payoff compared to other agents. We capture this intuition in an algorithm that determines a stable partition by sequentially pairing off $\theta$ lowest-claims agents. Note that if $|N|= \theta$, then the grand coalition $N$ forms a stable partition. We therefore define the algorithm for agent sets with at least $\theta+1$ agents.\bigskip

\noindent$\bm{\theta}$-\textbf{CEL algorithm}\medskip

\noindent\textbf{Input:} $N\in{\cal N}$ such that $|N|>\theta$ and $(c,E)\in{\cal P_{\theta}}^{N}$.\medskip

\noindent\textbf{Step~$\bm{1}$.} Let $N_1:=N$, $|N_1|>\theta$. Set $S_1:=\{1,2,\dots,\theta\}$ and $N_2:=N\setminus S_1$. If $|N_2|\leq \theta$, then set $S_2:=N_2$, define $\pi:=\{S_1, S_2\}$, and stop. Otherwise, go to Step~2.\medskip

\noindent\textbf{Step~$\bm{k}$ ($\bm{k>1}$).} Recall from Step~$k-1$ that $N_k:=N\setminus\left(\cup_{j=1}^{k-1}S_{j}\right)$ and $|N_k|> \theta$.  Set $S_k:=\{\theta k- (\theta-1), \theta k- (\theta-2), \ldots, \theta k\}$ and $N_{k+1}:=N\setminus \cup_{j=1}^{k}S_j$. If $|N_{k+1}|\leq \theta$, then set $S_{k+1}:=N_{k+1}$, define $\pi:=\{S_1, S_2, \ldots, S_{k+1}\}$, and stop. Otherwise, go to Step~$k+1$.\medskip

\noindent\textbf{Output:} A partition $\pi=\{S_1,\ldots,S_l\}$ for the coalition formation problem with agent set $N$ induced by $((c,E),CEL)$ such that for each $k\in\{1,\dots,l-1\}$, $|S_{k}|= \theta$ and $|S_{l}|\leq \theta$. If $|N|$ is divisible by $\theta$, then partition $\pi$ is constructed in $l-1=\frac{n-\theta}{\theta}$ steps. If $|N|$ is not divisible by $\theta$, then partition $\pi$ is constructed in $l-1=\lfloor\frac{n}{\theta}\rfloor$ steps.
\medskip

Note that the $\theta$-CEL algorithm only needs the agents' claims as input. In particular, it is not necessary to derive agents' preferences over coalitions. The following result states that the partition obtained by the $\theta$-CEL algorithm is stable.

\begin{theorem}\label{theorem:CELtheta}
Let $N\in{\cal N}$ such that $|N|>\theta$ and $(c,E)\in{\cal P}_{\theta}^{N}$. Consider the coalition formation problem with agent set $N$ induced by $((c,E),CEL)$. Then, the partition obtained by the $\theta$-CEL algorithm is stable.
\end{theorem}

We prove Theorem~\ref{theorem:CELtheta} in Appendix~\ref{appendix:proofTheorem8} but we explain the key intuition of the proof here. Consider a coalition $S\subsetneq N$ with $E_{S}=\alpha c^S$ and associated loss equals $(1-\alpha) c^S$. Hence, the coalitional loss decreases (increases) if we switch an agent in $S$ with a lower-claim (higher-claim) agent from $N\setminus S$. Since losses are split as equally as possible (taking zero as lower bound), and, by an iterative reasoning, sequentially matching up $\theta$ lowest-claims agents will lead to a stable partition.\medskip

Next, we illustrate how the $\theta$-CEL algorithm works.

\begin{example}[\textbf{2-CEL algorithm}]\normalfont
Let $N=\{1,2,3,4,5\}$, $c=(2,6,22,30,34)$, $E=47$, and $\theta = 2$, i.e., $(c,E)\in{\cal P}_{2}^{N}$ (hence, $\alpha=\frac{1}{2}$). With Theorem~\ref{theorem:thetasize} as our point of departure, we just focus on coalitions of size $2$.\medskip

\noindent\textbf{Step~$\bm{1}$.} Let $N_1=N$. We select pair $\{1,2\}$.\medskip

\noindent\textbf{Step~$\bm{2}$.} Let $N_2=N_1\setminus\{1,2\}=\{3,4,5\}$. We select pair $\{3,4\}$.\medskip

Now, we define $N_3=N\setminus\{1,2,3,4\}=\{5\}$. Since $|N_{3}|< 2$, we select $\{5\}$ and we obtain  $$\pi=\{\{1,2\},\{3,4\},\{5\}\}.$$

Note that the partition given by the algorithm is formed by positively assortative pairs. Theorem~\ref{theorem:CELtheta} proves that partition $\pi$ is stable. The $\theta$-CEL algorithm finds a stable partition without computing the payoffs of each agent in each coalition; it just needs to know the vector of claims.\hfill $\diamond$
\end{example}

Finally, the stable partition obtained by the $\theta$-CEL algorithm is not necessarily unique; the CEL (R2) example in Section~\ref{ex1} shows that stable partitions with larger coalition sizes are possible and that different stable partitions with the same coalition sizes may also exist.

\section{Conclusion}\label{secConclusion}

In this paper, we continue the analysis of \citeauthor{gallo2018rationing}'s (\citeyear{gallo2018rationing}) \textit{coalition formation problems with claims}. They focus on the existence of stable partitions but they do not analyze their exact structure. To make more precise predictions about the possible size and composition of stable partitions, we restrict attention to what we call $\theta$-minimal proportional generalized claims problems where coalitions of size smaller than $\theta$ receive zero endowments and all remaining coalitional endowments are a fixed proportion of the sum of the claims of coalition members. Note that the concept of proportionality implicitly assumes that claims are objectively determined or verifiable. If this was not the case, given the proportionality assumption for coalitional endowments, agents would have strong incentives to overstate their claims in order to receive larger coalitional endowments and payoffs. Therefore, strategic properties, e.g., strategy-proofness,\footnote{Strategy-proofness requires that for each agent, truthful claim reporting is a weakly dominant strategy in the associated direct revelation game. Note that even if coalitional endowments are fixed, most resource monotonic rules are not strategy-proof. The exception for fixed coalitional endowments is the CEA rule when agents' preferences are single-peaked (in this context, the CEA rule is also referred to as the uniform rule, see Appendix~\ref{appendix:surplus}). However, when coalitional endowments are proportional to coalitional claims / peaks, overstating claims / peaks can be a strictly dominant strategy for the CEA rule as well.} are not satisfied by our algorithms.
Let us briefly summarize our results.\medskip

We first characterize the structure of any possible stable partition when the rule applied satisfies continuity, strict resource monotonicity, and consistency. For the weaker notion of resource monotonicity, we demonstrate the existence of a stable partition formed by the maximal possible number of coalitions of size $\theta$ and one coalition of size smaller than $\theta$ formed by the remaining agents. Furthermore, we provide two algorithms to construct stable partitions mostly formed by $\theta$-size coalitions under CEA and CEL, respectively. For the CEA rule, the obtained stable partition is formed by coalitions where a highest-claim agent or a lowest-claim agent is added in each step. In particular, for some values of $\alpha$, the $\theta$-CEA algorithm leads to partitions that are formed either by only positively assortative coalitions or by mixed coalitions. For the CEL rule, a stable partition formed by only positively assortative coalitions is obtained by sequentially matching up $\theta$ lowest-claims agents.\medskip

Observe that our results are based on the assumption of proportional coalitional endowments. Future research could consider another principle of assigning coalitional endowments than proportionality. More generally, a two-step model in which first the total endowment is split among coalitions (by a claims rule) and second, within each coalition the coalitional endowment is split among its members (by another rule, possibly the same as the first one) could be considered \citep[for a related two-step model in a bankruptcy framework see, for instance,][]{izquierdo2016decentralized}.\footnote{Two-step procedures have also been analyzed, among others, by \cite{lorenzo2010two} and \cite{bergantinos2010characterization} for multi-issue allocation problems. \cite{moreno2011coalitional} studies a coalition procedure (two or more steps) for bankruptcy situations.} Therefore, coalition formation will depend on both, the rule that divides the total endowment among the different coalitions and the rule that is used to distribute the coalitional endowments among its members. Observe that our model can be straightforwardly extended to a two-step procedure in which the rule used in the first step is the proportional rule for any coalition of size larger than or equal to $\theta$ and the constant zero rule for coalitions of size smaller than $\theta$, and the rule applied in the second step satisfies continuity, (strict) resource monotonicity, and consistency (or, for some of our results, equals the CEA / CEL rule).\medskip

Another crucial assumption we make in our model is that the class of problems we consider is based on claims problems, i.e., that there is \textit{excess demand} for the endowment that is allocated. Let us assume that \textit{excess supply} could also occur, e.g., in a situation where the resource to be allocated is a certain amount of labor and each agent has an optimal amount of labor, a \textit{peak}, she would like to take on. We essentially assume that the closer the outcome is to the peak, the better it is; i.e., agents have single-peaked preferences. Then, the so-called \textit{uniform rule} represents an extension of the CEA rule to excess supply problems\footnote{For excess demand problems, the CEA / uniform rule allocates the endowment as equally as possible, taking the agents' peaks as upper bounds; correspondingly, for excess supply problems it allocates the endowment as equally as possible, taking agents' peaks as lower bounds.} and the \textit{equal surplus rule} represents an extension of the CEL rule to excess supply problems.\footnote{For excess demand problems, the CEL rule allocates losses as equally as possible, taking zero as lower bound; correspondingly, for excess supply problems, the equal surplus rule allocates the total surplus induced by the endowment  equally among agents.} This model variation is analyzed in  Appendix~\ref{appendix:surplus}. We first show that even though both the uniform and the equal surplus rules satisfy continuity, resource-monotonicity, and consistency, a result similar to \citet[][Theorem~2]{gallo2018rationing} does not hold for the mixed case: Example~\ref{example:uniformrule} (Example~\ref{example:equalsurplusrule}, respectively) shows that for problems with excess demand as well as excess supply that are based on the CEA / uniform rule (the CEL / equal surplus rule, respectively), no stable partition may exist. Second, if we restrict attention to classes of problems with either (i) only excess demand or (ii) only excess supply, then all our algorithmic results extend. The results in Case~(i) are the ones we obtained for the CEA and CEL rules and corresponding results for the uniform and equal surplus rules can be obtained in Case~(ii) (see Appendix~\ref{appendix:surplus}).

\begin{appendix}

\section*{Appendix}

\section{Parametric rules}\label{appendix:parametricrules}

For each parametric rule, there is a function of two variables such that for each problem, each agent's payoff is the value taken by this function when the first argument is her claim and the second one is parameter $\lambda$, which is the same for all agents. This parameter is chosen so that the sum of agents' payoffs is equal to the endowment.\medskip

A \textbf{parametric rule} for $N\in \mathcal{N}$ is defined as follows: Let $f$ be a collection of functions $\{f_{i}\}_{i\in N}$,\footnote{When the parametric rule is symmetric, $f_{i}$ is the same for all agents.} where each $f_{i}$ $:$ $\mathbb{R}_{++}\times \lbrack a,b]\longrightarrow \mathbb{R}_{+}$ is continuous and weakly increasing in the second argument $\lambda$, $\lambda \in\lbrack a,b]$, $-\infty\leq a<b\leq\infty$ and for each $i\in N$ and $c  _{i}\in\mathbb{R}_{++}$, $f_{i}(c_{i},a)=0$ and $f_{i}(c_{i},b)=c_{i}$. Hence, for each $f$, a rule $F$ is defined as follows. For each $(c,E)\in\mathcal{C}^{N}$ and each $i\in N$,
\begin{equation*}
F_{i}(c,E)=f_{i}(c_{i},\lambda )\mbox{ where }\lambda\mbox{ is chosen so that }\sum\nolimits_{j\in N}f_{j}(c_{j},\lambda )=E.
\end{equation*}
Then, $f$ is said to be a \textbf{parametric representation of parametric rule $\bm{F}$}.\medskip

\citet{young1987dividing} characterizes parametric rules on the basis of symmetry,\footnote{Two agents with equal claims should receive equal payoffs.} continuity, and bilateral consistency\footnote{Bilateral consistency requires consistency only when $|S|=2$.}. \citet{stovall2014asymmetric} characterizes the family of possibly asymmetric parametric rules on the basis of continuity, resource monotonicity, bilateral consistency, and two additional axioms, ``$N$-continuity'' and ``intrapersonal consistency.''

\section{Proof of Theorem~\ref{theorem:thetasize}}\label{appendix:proofTheorem6}

We first introduce some lemmas that will be used to prove Theorem~\ref{theorem:thetasize}. We first show that, given a $\theta$-minimal proportional generalized claims problem and a consistent rule, if all agents receive proportional payoffs in a coalition, then all agents receive proportional payoffs in any subcoalition (except subcoalitions of size smaller than $\theta$) as well.

\begin{lemma}\label{lemma:subcoalitionstheta}
Let $N\in \mathcal{N}$ and $(c,E)\in \mathcal{P}_{\theta}^{N}$. Consider $F\in{\cal F}$ satisfying consistency. If $S\subseteq N$, $|S|>\theta$, is such that for each $i\in S$, $F_i(c_S, E_S)=\alpha c_i$, then for each $S'\subsetneq S$ with $|S'|\geq\theta$ and each $j\in S'$, $F_j(c_{S'}, E_{S'})=\alpha c_j$, i.e., $S'\sim_j S$.
\end{lemma}

\begin{proof}[\textbf{Proof}] Let $N$, $(c,E)$, and $F$ be as specified in the lemma. Consider $S\subseteq N$, $|S|>\theta$, such that for each $i\in S$, $F_i(c_S, E_S)=\alpha c_i$ and consider $S'\subsetneq S$ with $|S'|\geq \theta$ and problem $(c_{S'}, E_{S'})\in{C}^{S'}$. By consistency, for each $j\in S'$, $F_j\left(c_{S'},\sum_{k\in S'}F_{k}(c_S,E_S)\right)=F_{j}(c_S,E_S)$ and hence, $F_j\left(c_{S'},\sum_{k\in S'}F_{k}(c_S,E_S)\right)=\alpha c_j$. Since $(c,E)\in \mathcal{P}_{\theta}^{N}$ we have that $E_{S'}=\alpha c^{S'}=\sum_{k\in S'}F_k (c_{S},E_{S})$. Thus, for each $j\in S'$, $F_j\left(c_{S'},E_{S'}\right)=\alpha c_j$, i.e., $S'\sim_j S$.\end{proof}

We next show that, given a $\theta$-minimal proportional generalized claims problem and a resource monotonic and consistent rule, if some agent in a coalition of size larger than $\theta$ does not receive a proportional payoff, then one agent can be excluded from the coalition such that all remaining agents are weakly better off.

\begin{lemma}\label{lemma:noproportionaltheta}
Let $N\in \mathcal{N}$ and $(c,E)\in \mathcal{P}_{\theta}^{N}$. Consider $F\in{\cal F}$ satisfying resource monotonicity and consistency. If $S\subseteq N$, $|S|>\theta$, is such that for some $i\in S$, $F_i(c_S, E_S)\neq\alpha c_i$, then there is $S'\subsetneq S$, with $|S'|=|S|-1$, such that for all agents $j\in S'$, $S'\succsim^{((c,E),F)}_{j} S$.
\end{lemma}

\begin{proof}[\textbf{Proof}] Let $N$, $(c,E)$, and $F$ be as specified in the lemma. Consider $S\subseteq N$, $|S|>\theta$, such that for some $i\in S$, $F_i(c_S, E_S)\neq\alpha c_i$. By consistency, for each $j\in S\setminus\{i\}$, \begin{equation}\label{eq:1}F_j\left(c_{S\setminus\{i\}},\sum_{k\in S\setminus\{i\}}F_k(c_S ,E_{S})\right)= F_j(c_S,E_{S}).\end{equation}
Since $E_{S}=\alpha c^{S}$ and $F_i(c_S, E_S)\neq\alpha c_i$, we can assume, without loss of generality, that agent $i$ receives an over-proportional payoff, i.e.,  $F_i(c_S,E_{S})> \alpha c_i$. Therefore, subcoalition $S\setminus\{i\}$ can achieve a larger joint endowment without agent $i$ compared to what they jointly receive at $(c_S, E_S)$, i.e., $E_{S\setminus\{i\}}=\alpha c^{S\setminus\{i\}}>\sum_{k\in S\setminus\{i\}} F_k(c_S,E_{S})$.
Hence, the endowment at problem $(c_{S\setminus\{i\}}, E_{S\setminus\{i\}})$ is larger than at problem $\left(c_{S\setminus\{i\}}, \sum_{k\in S\setminus\{i\}}F_k(c_S,E_{S})\right)$ and by resource monotonicity, for each agent $j\in S\setminus\{i\}$,
$$F_j(c_{S\setminus\{i\}},E_{S\setminus\{i\}})\geq F_j\left(c_{S\setminus\{i\}},\sum_{k\in S\setminus\{i\}}F_k(c_S,E_{S})\right)\stackrel{(\ref{eq:1})}{=}F_j(c_S, E_{S}).$$
Hence, for each $j\in S':=S\setminus\{i\}$, $S'\succsim^{((c,E),F)}_{j} S$.
\end{proof}

We next show that, given a $\theta$-minimal proportional generalized claims problem and a resource monotonic and consistent rule, there is a coalition of size $\theta$ that is weakly preferred by all its agents to any other coalition of size larger than $\theta$.

\begin{proposition}\label{prop:toptheta}
Let $N\in \mathcal{N}$ and $(c,E)\in \mathcal{P}_{\theta}^{N}$. Consider $F\in{\cal F}$ satisfying resource monotonicity and consistency. If $S\subseteq N$, $|S|> \theta$, then there is $S'\subsetneq S$, with $|S'|=\theta$, such that for each agent $j\in S'$, $S'\succsim^{((c,E),F)}_{j} S$.
\end{proposition}

\begin{proof}[\textbf{Proof}] Let $N$, $(c,E)$, and $F$ be as specified in the lemma. Consider $S\subseteq N$, $|S|>\theta$. We distinguish two cases:\medskip

\noindent \textit{Case~$1$.} $S\subseteq N$ is such that for each $i\in S$, $F_i(c_S, E_S)=\alpha c_i$.\smallskip

\noindent Then, by Lemma~\ref{lemma:subcoalitionstheta}, for each $S'\subsetneq S$ with $|S'|\geq \theta$ and each $j\in S'$, $F_j(c_{S'}, E_{S'})=\alpha c_j$. In particular, this is true for any $S'\subsetneq S$ with $|S'|= \theta$. Thus, for each agent $j\in S'$, $S'\sim^{((c,E),F)}_{j} S$ and we are done.
\medskip

\noindent \textit{Case~$2$.} $S\subseteq N$ is such that for some $i\in S$, $F_i(c_S, E_S)\neq\alpha c_i$.\smallskip

\noindent Then, by Lemma~\ref{lemma:noproportionaltheta}, there exists a subcoalition $\bar{S}\subsetneq S$ with $|\bar{S}|=|S|-1$, such that for each agent $j\in \bar{S}$, $\bar{S}\succsim^{((c,E),F)}_{j} S$. If $|\bar{S}|=\theta$, then we are done. Otherwise, starting with coalition $\bar{S}$, successively apply Case $1$ or the above argument to obtain a coalition $S'\subsetneq S$ with $|S'|=\theta$ such that for each $j\in S'$,
$S'\succsim^{((c,E),F)}_{j} S$.\end{proof}

Now, we introduce some properties for coalition formation problems that play an important role in the proof of Theorem~\ref{theorem:thetasize}. Let $\succsim$ be a coalition formation problem with agent set $N$.\medskip

The first property, \textbf{weak pairwise alignment}, requires that if one agent orders two coalitions in one direction, no other agent in the intersection of both coalitions orders them in the opposite direction. If $S, S'\subseteq N$, and $i,j\in S\cap S'$, then $S\succ_i S'$ implies $S\succsim_j S'$.\medskip

For the second property we first introduce some notation. Let $S,S'\subseteq N$. Then, if for each $i\in S\cap S'$, $S\succsim_{i} S'$, and for at least one $j\in S\cap S'$, $S\succ_{j}S'$, we write $S\trianglerighteq S'$. A \textbf{cycle} \citep[called ring in][]{gallo2018rationing} for $\succsim$ is an ordered set of coalitions $(S_{1},...,S_{k})$, $k>2,$ such that for each $l=1,\ldots,k$, $S_{l+1}\trianglerighteq S_{l}$ (subscripts modulo $k$). This definition of a cycle requires one agent in the intersection of any two
consecutive coalitions to have strict preferences between both coalitions while her ``intersection-mates'' can be indifferent between them. Coalition formation problem $\succsim$ satisfies \textbf{acyclicity} if it has no cycles.\medskip

The last property, the \textbf{top-coalition property} \citep[]{BanerjeeKonishiSonmezSCW2001} is sufficient to guarantee stability. A coalition $S'\subseteq S$ is a \textbf{top-coalition of $\bm{S}$} if for each $i\in S' $ and each $T\subseteq S$ with $i\in T$, we have $S'\succsim _{i}T$. Coalition formation problem $\succsim$ satisfies the \textbf{top-coalition property} if each non-empty set of agents $S\subseteq N$ has a top coalition.\medskip

\citet[][Theorem~1]{gallo2018rationing} show that a coalition formation problem that satisfies weak pairwise alignment and acyclicity satisfies the top-coalition property. Furthermore, \citet[][Lemma~2]{gallo2018rationing} \footnote{\citet[][Lemma~2]{gallo2018rationing} only states that resource monotonicity and consistency are necessary to satisfy weak pairwise alignment. However, the proof of the lemma shows that the two properties are also sufficient. Therefore, when referring to \citet[][Lemma~2]{gallo2018rationing}, we refer to the following more general result: \emph{Let $N\in \mathcal{N}$, $(c,E)\in \mathcal{P}^{N}$, and $F\in\mathcal{F}$. Then, $\succsim^{((c,E),F)}$ satisfies weak pairwise alignment if and only if $F$ is resource monotonic and consistent}.} and \citet[][Lemma~3]{gallo2018rationing} show that given a generalized claims problem, if the rule applied is continuous, resource monotonic, and consistent, then the coalition formation problem induced by $F$ satisfies weak pairwise alignment and acyclicity and, consequently, the top coalition property.\medskip

We next show that given a $\theta$-minimal proportional generalized claims problem, there is a top-coalition of $N$ of size at least $\theta$.\pagebreak

\begin{lemma}\label{lemma:topcoalition}
Let $N\in \mathcal{N}$ and $(c,E)\in \mathcal{P}_{\theta}^{N}$. Consider $F\in{\cal F}$ satisfying continuity, resource monotonicity, and consistency. Then, for the coalition formation problem with agent set $N$ induced by $((c,E), F)$, there exists a top-coalition $S\subseteq N$ of $N$ with $|S|\geq \theta$.
\end{lemma}

\begin{proof}[\textbf{Proof}] Let $N$, $(c,E)$, and $F$ be as specified in the lemma. By \citet[][Lemmas~2, 3, and Theorem~1]{gallo2018rationing}, $\succsim^{((c,E),F)}$ satisfies weak pairwise alignment, acyclicity, and the top-coalition property.  Thus, a top-coalition of $N$ exists.

For each $i\in N$, let $TC_{i}(N)$ be the set of most / top preferred coalitions of $N$ containing agent $i$, i.e., $TC_{i}(N)=\{S\subseteq N:i\in S$ and for each $T\subseteq N$ with $i\in T$, $S\succsim_{i}T\}$. We show that there exists at least one coalition $S\subseteq N$ with $|S|\geq \theta$ such that for each $i\in S$, $S\in TC_{i}(N)$; thus $S$ is a top-coalition of $N$.\medskip

Assume by contradiction that there is no such coalition $S$, i.e., all top-coalitions of $N$ are of size smaller than $\theta$. Given that $(c,E)\in \mathcal{P}_{\theta}^{N}$, all agents get zero payoffs in all coalitions of size smaller than $\theta$. In particular, each agent gets a zero payoff in each top-coalition of $N$. However, each coalition of size $\theta$ and larger has a positive endowment and at least one agent in each such coalition receives a positive payoff (not necessarily the same agent in all these coalitions). In the following iterative procedure, we either contradict weak pairwise alignment or acyclicity.\medskip

\noindent\textbf{Step~$\bm{1}$.} Let $S_1\subseteq N$ be a coalition of size $\theta$ or larger, i.e., $|S_1|\geq\theta$, and $i_1\in S_1$ be an agent with a positive payoff in coalition $S_1$, i.e., $F_{i_1}(S_1, E_{S_1})>0$. Without loss of generality, assume that  $S_1\in TC_{i_1}(N)$.
Since, by assumption, $S_{1}$ is not a top-coalition of $N$, there exists an agent in $S_1$, say agent $i_2\in S_1$, and a coalition in $TC_{i_2}(N)$, say coalition $S_2\in TC_{i_2}(N)$, such that $S_{2}\succ _{i_2}S_{1}$. Note that $|S_{2}|\geq \theta$ (otherwise, agent $i_2$ would get a zero payoff and the strict preference would not be possible). If there is $j\in S_1\cap S_2$ such that $S_{1}\succ _{j}S_{2}$, then weak pairwise alignment is violated. Otherwise, $S_2\trianglerighteq S_{1}$; set $(S_1,S_2)$ as ordered set of coalitions and go to Step~2.\medskip

\noindent\textbf{Step~$\bm{k}$ ($\bm{k > 1}$).} Note that Step~$k$ starts with an ordered set of coalitions $(S_{1},...,S_{k})$, such that for each $l=1,\ldots,k-1$, $S_{l+1}\trianglerighteq S_{l}$.

Consider $S_{k}\subseteq N$, $|S_{k}|\geq \theta$, and $i_k\in S_k$ such that $S_k\in TC_{i_{k}}(N)$ from Step~$k-1$. Since, by assumption, $S_{k}$ is not a top-coalition of $N$, there exists an agent in $S_{k}$, say agent $i_{k+1}\in S_{k}$, and a coalition in $TC_{i_{k+1}}(N)$, say coalition $S_{k+1}\in TC_{i_{k+1}}(N)$, such that $S_{k+1}\succ _{i_{k+1}}S_{k}$. Note that $|S_{k+1}|\geq \theta$ (otherwise agent $i_{k+1}$ would get a zero payoff and the strict preference would not be possible). If there is $j\in S_k\cap S_{k+1}$ such that $S_{k}\succ_{j}S_{k+1}$, then weak pairwise alignment is violated. Otherwise, $S_{k+1}\trianglerighteq S_{k}$; set $(S_1,S_2,\ldots, S_{k+1})$  as ordered set of coalitions and go to Step~$k+1$.\medskip

Note that our iterative procedure either stops with a violation of weak pairwise alignment or it produces an infinite sequence of ordered sets $(S_k)_{k\in \mathbb{N}}$ such that for each $k\in \mathbb{N}$, $S_{k+1}\trianglerighteq S_k$.
Given that $N$ is finite, the set of coalitions for $N$ is finite as well.  Hence, the infinite sequence of ordered sets $(S_k)_{k\in \mathbb{N}}$ must list some coalitions multiple times. If for some $k\in \mathbb{N}$, $k\geq 2$, $S_{k+1}=S_{k-1}$, then  [$S_{k+1}\succ _{i_{k+1}}S_{k}$ and  $S_{k}\succ _{i_{k}}S_{k+1}$] and weak pairwise alignment is violated. Otherwise, the infinite sequence of ordered sets $(S_k)_{k\in \mathbb{N}}$ contains at least one cycle, which violates acyclicity. Thus, we either obtain a contradiction with weak pairwise alignment or acyclicity.\end{proof}

We are now ready to prove Theorem~\ref{theorem:thetasize}.

\begin{proof}[\textbf{Proof of Theorem~\ref{theorem:thetasize}}]
Let $N\in \mathcal{N}$ and $(c,E)\in \mathcal{P}_{\theta}^{N}$. Consider  $F\in{\cal F}$ satisfying continuity, resource monotonicity, and consistency.
\citet[][Lemmas~2, 3, and Theorem~1]{gallo2018rationing} implies that $\succsim^{((c,E),F)}$ satisfies the top-coalition property; hence, a stable partition exists for the coalition formation problem with agent set $N$ induced by $((c,E),F)$ \citep[][Theorem~2]{gallo2018rationing}.
If $|N|= \theta$, then the only stable partition consists of the grand coalition $N$. Hence, assume that $|N|>\theta$.\medskip

\noindent We iteratively construct a stable partition $\pi\in St(\succsim^{((c,E),F)})$ with coalition sizes of at most $\theta$.\medskip

\noindent \textbf{Step~$\bm{1}$.} Let $N_1:=N$, $|N_1|> \theta$. The set of top-coalitions of $N_1$ is non-empty and, by Lemma~\ref{lemma:topcoalition}, there exists a top-coalition $S'_1\subseteq N_1$ of $N_1$ such that $|S_1'|\geq\theta$, i.e., for each $i\in S'_1 $ and each $T\subseteq N_1$ with $i\in T$, we have $S'_1\succsim^{((c,E),F)}_{i}T$. If $|S'_1|=\theta$, then set $S_1:=S'_1$. Otherwise, by Proposition~\ref{prop:toptheta}, there is a coalition $S_1\subsetneq S'_1$, such that $|S_1|=\theta$ and for each $j\in S_1$, $S_1\succsim^{((c,E),F)}_{j}S'_1$. In particular, since $S'_1$ is a top-coalition, it follows that $S_1\sim^{((c,E),F)}_{j}S'_1$. Hence, for each $i\in S_1$ and each $T\subseteq N_1$ with $i\in T$, we have $S_1\sim^{((c,E),F)}_{i}S'_1\succsim^{((c,E),F)}_{i}T$ and $S_1$ is a top-coalition of $N_1$ as well. Hence, agents in $S_1$ can never be strictly better off in any other coalition $T\subseteq N_1=N$. Thus, if $S_1$ is part of a stable partition, no agent in $S_1$ can block it.\medskip

Set $N_2:=N\setminus S_1$. If $|N_2|\leq \theta$, then set $S_2:=N_2$, define $\pi:=\{S_1,S_2\}$, and stop. Otherwise, go to Step~2.\medskip\pagebreak

\noindent \textbf{Step~$\bm{k}$ ($\bm{k>1}$).} Recall from Step~$k-1$ that $N_k:=N\setminus\left(\cup_{i=1}^{k-1}S_{i}\right)$ and $|N_k|> \theta$. The set of top-coalitions of $N_k$ is non-empty and, by Lemma~\ref{lemma:topcoalition} (with $N_k$ in the role of $N$), there exists a top-coalition $S'_k\subseteq N_k$ of $N_k$ such that $|S'_k|\geq \theta$, i.e., for each $i\in S'_k$ and each $T\subseteq N_k$ with $i\in T$, we have $S'_k\succsim^{((c,E),F)}_{i}T$. If $|S'_k|=\theta$, then set $S_k:=S'_k$. Otherwise, by Proposition~\ref{prop:toptheta}, there is a coalition $S_k\subsetneq S'_k$, such that $|S_k|=\theta$ and for each $j\in S_k$, $S_k\succsim^{((c,E),F)}_{j}S'_k.$  In particular, since $S'_k$ is a top-coalition, it follows that $S_k\sim^{((c,E),F)}_{j}S'_k$. Hence, for each $i\in S_k$ and each $T\subseteq N_k$ with $i\in T$, we have $S_k\sim^{((c,E),F)}_{i}S'_k\succsim^{((c,E),F)}_{i}T$ and $S_k$ is a top-coalition of $N_k$ as well. Hence,  agents in $S_{k}$ can never be strictly better off in any other coalition $T\subseteq N_k$. In addition, it follows from previous steps that for each $j\in\{1,\ldots,k\}$, agents in $S_{j}$ can never be strictly better off in any other coalition $T\subseteq N_j$. Thus, if $S_1,\ldots,S_k$ are part of a stable partition, no agent in $\cup_{i=1}^{k}S_i$ can block it.\medskip

Set $N_{k+1}:=N\setminus \left(\cup_{i=1}^{k}S_i\right)$. If $|N_{k+1}|\leq \theta$, then set $S_{k+1}:=N_{k+1}$, define $\pi:=\{S_1,\ldots,S_{k+1}\}$, and stop. Otherwise, go to Step~$k+1$.\medskip

After at most $|N|-\theta$ steps, we have constructed a stable partition $\pi=\{S_1,\ldots,S_l\}$ of coalitions of size at most $\theta$.\end{proof}

\section{Proof of Theorem~\ref{theorem:CEA}}\label{appendix:proofCEA}

Recall that $N=\{1,\dots, n\}$ and $c_1\leq c_2\leq \cdots\leq c_n$. For each $S\subseteq N$, we denote the CEA parameter associated with $(c_S,E_S)$ by $\lambda_{E_S}$, i.e., for each $i\in S$, $CEA_{i}(c_S,E_S)=\min\{c_i,\lambda_{E_S}\}$, where $\lambda_{E_S}$ is chosen so that $\sum_{j\in S}\min\{c_j,\lambda_{E_S}\}= E_S$.\medskip

We first introduce some lemmas that will be used to prove Theorem~\ref{theorem:CEA}. First, we show that each agent $i\in N\setminus\{n\}$ weakly prefers to form a pair with highest-claim agent~$n$, instead of with any other agent.

\begin{lemma}\label{lemma:top}
Let $N\in{\cal N}$, $(c,E)\in\mathcal{P}_{2}^{N}$, and $\succsim^{((c,E),CEA)}$ be the coalition formation problem with agent set $N$ induced by $((c,E),CEA)$. Then, for each $i\in N\setminus\{n\}$ and each $j\in N\setminus\{i,n\}$,
\begin{equation*}\{i,n\}\succsim^{((c,E),CEA)}_{i}\{i,j\}.
\end{equation*}
\end{lemma}

\begin{proof}[\textbf{Proof}]
Let $N$, $(c,E)$, and $\succsim^{((c,E),CEA)}$ be as specified in the lemma. Let $i\in N\setminus\{n\}$ and $j\in N\setminus\{i,n\}$. We prove $\{i,n\}\succsim^{((c,E),CEA)}_{i}\{i,j\}$ by showing that
\begin{equation*}CEA_i(c_{\{i,n\}}, E_{\{i,n\}})\geq CEA_i(c_{\{i,j\}}, E_{\{i,j\}}).
\end{equation*}
Given that $(c,E)\in {\cal P}_{2}^{N}$, we have $E_{\{i,n\}} = \alpha (c_i+c_n)$ and $E_{\{i,j\}} = \alpha (c_i+c_j)$.
If $CEA_i(c_{\{i,n\}}, E_{\{i,n\}})=c_i$, then the above inequality holds automatically. Hence, assume that $CEA_i(c_{\{i,n\}}, E_{\{i,n\}})=\lambda_{E_{\{i,n\}}}<c_i$. Since $c_i\leq c_n$, this implies $CEA_n(c_{\{i,n\}}, E_{\{i,n\}})=\lambda_{E_{\{i,n\}}}< c_n$.  Thus, $CEA_i(c_{\{i,n\}}, E_{\{i,n\}})=\frac{\alpha (c_i+c_n)}{2}$. We distinguish two cases:\medskip

\noindent \textit{Case~$1$.} $CEA_j(c_{\{i,j\}},E_{\{i,j\}})=c_j$.\footnote{Note that this case only happens when $j<i$.} Hence, $CEA_i(c_{\{i,j\}},E_{\{i,j\}})=\alpha (c_i+c_j)-c_j$.  Then,
$$CEA_i(c_{\{i,n\}},E_{\{i,n\}})\geq CEA_i(c_{\{i,j\}},E_{\{i,j\}})$$
$$\Leftrightarrow\quad \frac{\alpha(c_i + c_n)}{2}\geq \alpha (c_i+c_j)-c_j$$
$$\Leftrightarrow\quad \alpha c_i + \alpha c_n \geq 2 \alpha c_i+ 2(\alpha - 1) c_j$$
$$\Leftrightarrow\quad \underbrace{\alpha}_{>0} (\underbrace{c_n - c_i}_{\geq 0})\geq   2(\underbrace{\alpha - 1}_{<0}) \underbrace{c_j}_{\geq 0}.$$

\noindent \textit{Case~$2$.} $CEA_j(c_{\{i,j\}},E_{\{i,j\}})=\lambda_{E_{\{i,j\}}}=\frac{\alpha (c_i+c_j)}{2}$. Thus, since $c_j\leq c_n$,
\begin{equation*}
CEA_i(c_{\{i,n\}}, E_{\{i,n\}})=\frac{\alpha (c_i+c_n)}{2}\geq\frac{\alpha (c_i+c_j)}{2}= CEA_i(c_{\{i,j\}}, E_{\{i,j\}}).\qedhere
\end{equation*}\end{proof}

We will now focus on agent~$n$ and discover with whom she wants to form a pair.\medskip

Our next lemma captures the intuition that for low values of $\alpha$ that are associated to the situation in which all agents receive an equal split of the endowment, i.e., for each $i\in N\setminus\{n\}$, $CEA_{i}(c_{\{i,n\}},E_{\{i,n\}})=\lambda_{E_{\{i,n\}}}= CEA_{n}(c_{\{i,n\}},E_{\{i,n\}})$, agent~$n$ weakly prefers coalition $\{n-1,n\}$ to any other coalition of size two.

\begin{lemma}\label{corollary:alpha-very-small}
Let $N\in{\cal N}$, $(c,E)\in\mathcal{P}_{2}^{N}$, and $\succsim^{((c,E),CEA)}$ be the coalition formation problem with agent set $N$ induced by $((c,E),CEA)$. If $\alpha \leq \frac{2c_1}{c_1+c_n}\equiv\beta_1$, then for each $i\in N\setminus\{n-1,n\}$, \begin{equation*}\{n-1,n\}\succsim^{((c,E),CEA)}_{n}\{i,n\}.
\end{equation*}
\end{lemma}

\begin{proof}Let $N$, $(c,E)$, $\succsim^{((c,E),CEA)}$, and $\alpha$ be as specified in the lemma. Let $i\in N\setminus\{n-1,n\}$. We prove $\{n-1,n\}\succsim^{((c,E),CEA)}_{n}\{i,n\}$ by showing that
\begin{equation*}CEA_{n}(c_{\{n-1,n\}},E_{\{n-1,n\}}) \geq CEA_{n}(c_{\{i,n\}},E_{\{i,n\}}).
\end{equation*}
First note that $$\alpha \leq \frac{2 c_1}{c_1 + c_n}$$$$ \Leftrightarrow \alpha (c_1 + c_n)\leq 2 c_1$$ $$\Leftrightarrow \frac{\alpha (c_1 + c_n)}{2}\leq c_1.$$
Given that $(c,E)\in {\cal P}_{2}^{N}$, we have $E_{\{1,n\}} = \alpha (c_1+c_n)$. Therefore, $CEA_1 (c_{\{1,n\}},E_{\{1,n\}})=\frac{\alpha (c_1 +c_n)}{2}=CEA_n (c_{\{1,n\}},E_{\{1,n\}})$.\medskip

Observe also that $$\alpha \leq \frac{2 c_1 }{c_1 + c_n}$$
$$ \Leftrightarrow \alpha c_n\leq (2-\alpha) c_1.$$
Let $j\in N\setminus\{1,n\}$. Then, $\alpha c_n\leq (2-\alpha) c_1$, $c_1 \leq c_j$, and $(2-\alpha)>0$ imply
$$\alpha c_n\leq (2-\alpha) c_j$$ $$\Leftrightarrow \frac{\alpha (c_j + c_n)}{2}\leq c_j.$$ Therefore, $CEA_j (c_{\{j,n\}},E_{\{j,n\}})=\frac{\alpha (c_j +c_n)}{2}=CEA_n (c_{\{j,n\}},E_{\{j,n\}})$.\medskip

Finally, let $i\in N\setminus\{n-1,n\}$. Given than $c_i \leq c_{n-1}\leq c_{n}$,
\begin{equation*}
 CEA_{n}(c_{\{n-1,n\}}, E_{\{n-1,n\}})=\frac{\alpha(c_{n-1} +c_{n})}{2}\geq \frac{\alpha(c_{i} +c_{n})}{2}=CEA_{n}(c_{\{i,n\}}, E_{\{i,n\}}).\qedhere
 \end{equation*}\end{proof}

Similarly, the following lemma captures the intuition that for high values of $\alpha$ that are associated to the situation in which all agents, except agent $n$, receive their claim, i.e., for each $i\in N\setminus\{n\}$, $CEA_{i}(c_{\{i,n\}},E_{\{i,n\}})=c_i$, agent~$n$ weakly prefers coalition $\{1,n\}$ to any other coalition of size two.

\begin{lemma}\label{corollary:alpha-very-large}
Let $N\in{\cal N}$, $(c,E)\in\mathcal{P}_{2}^{N}$, and $\succsim^{((c,E),CEA)}$ be the coalition formation problem with agent set $N$ induced by $((c,E),CEA)$. If $\alpha\geq\frac{2c_{n-1}}{c_{n-1}+ c_n}\equiv\gamma_1$, then for each $i\in N\setminus\{1,n\}$,
\begin{equation*}\{1,n\}\succsim^{((c,E),CEA)}_{n}\{i,n\}.
\end{equation*}
\end{lemma}

\begin{proof}Let $N$, $(c,E)$, $\succsim^{((c,E),CEA)}$, and $\alpha$ be as specified in the lemma. Let $i\in N\setminus\{1,n\}$. We prove $\{1,n\}\succsim^{((c,E),CEA)}_{n}\{i,n\}$ by showing that
\begin{equation*}CEA_{n}(c_{\{1,n\}},E_{\{1,n\}}) \geq CEA_{n}(c_{\{i,n\}},E_{\{i,n\}}).
\end{equation*}
First note that $$\alpha \geq \frac{2c_{n-1}}{c_{n-1}+c_n}$$ $$\Leftrightarrow \alpha (c_{n-1} + c_n)\geq 2 c_{n-1}$$ $$\Leftrightarrow \frac{\alpha (c_{n-1} + c_n)}{2}\geq c_{n-1}.$$
Given that $(c,E)\in {\cal P}_{2}^{N}$, we have $E_{\{n-1,n\}} = \alpha (c_{n-1}+ c_{n})$. Therefore, $CEA_{n-1} (c_{\{n-1,n\}},E_{\{n-1,n\}})=c_{n-1}$ and $CEA_{n} (c_{\{n-1,n\}},E_{\{n-1,n\}})=\alpha c_n -(1-\alpha) c_{n-1}$.\medskip

Observe also that $$\alpha \geq \frac{2c_{n-1}}{c_{n-1}+c_n}$$$$ \Leftrightarrow \alpha c_n \geq (2-\alpha) c_{n-1}.$$ Let $j\in N\setminus\{n-1,n\}$. Then, $\alpha c_n \geq (2-\alpha) c_{n-1}$, $c_j \leq c_{n-1}$, and $(2-\alpha)> 0$, imply
$$\alpha c_n \geq (2-\alpha) c_j$$ $$\Leftrightarrow \frac{\alpha (c_j + c_n)}{2}\geq c_j.$$  Therefore, $CEA_j (c_{\{j,n\}},E_{\{j,n\}})=c_j$ and $CEA_n (c_{\{j,n\}},E_{\{j,n\}})= \alpha c_n -(1-\alpha) c_{j}$.\medskip

Finally, let $i\in N\setminus\{1,n\}$. Given than $c_1\leq c_{i}\leq c_{n}$ and $(1-\alpha)>0$,
\begin{equation*}
CEA_{n}(c_{\{1,n\}}, E_{\{1,n\}})=\alpha c_n -(1-\alpha) c_{1}\geq \alpha c_n -(1-\alpha) c_{i} =CEA_{n}(c_{\{i,n\}}, E_{\{i,n\}}).\qedhere
\end{equation*}\end{proof}

Lemmas~\ref{corollary:alpha-very-small} and \ref{corollary:alpha-very-large} now imply that agents $1$ and $n-1$ are potential ``stable partners'' for agent~$n$. When $\alpha$ is low ($\alpha\leq \frac{2c_1}{c_1+c_n}\equiv\beta_1$), then $\{n-1,n\}$ is the candidate for a stable pair, and when $\alpha$ is high ($\alpha\geq \frac{2c_{n-1}}{c_{n-1}+c_n}\equiv\gamma_1$), then $\{1,n\}$ is the candidate for a stable pair. Next, we show that for any other value of $\alpha$ in between $\beta_1$ and $\gamma_1$, agents $1$ and $n-1$ are also potential ``stable partners'' for agent~$n$. Thus, we next need to determine a threshold value $\delta_1$ for parameter $\alpha$ when $\beta_1 <\alpha \leq \gamma_1$ to see when agent~$n$'s partner of choice is $n-1$ (for $\alpha\leq\delta_1$) and when it is $1$ (for $\alpha\geq\delta_1$).\medskip

We next show that $\delta_1\equiv \frac{2c_1}{2c_1 -c_{n-1} +c_n}$, the value specified at Step~1 of the CEA algorithm to trigger either Case~(i) with coalition $\{n-1,n\}$ or Case~(ii) with coalition $\{1,n\}$.\pagebreak

\begin{lemma}\label{lemma:threshold}
Let $N\in{\cal N}$, $(c,E)\in\mathcal{P}_{2}^{N}$, and $\succsim^{((c,E),CEA)}$ be the coalition formation problem with agent set $N$ induced by $((c,E),CEA)$. Assume that $\beta_1 \leq\alpha \leq \gamma_1$. Then, for $\delta_1\equiv \frac{2c_1}{2c_1 -c_{n-1} +c_n}$ we have
\begin{itemize}
	\item[\emph{(i)}]If $\alpha\leq\delta_1$, then for each $i\in N\setminus\{n-1,n\}, \{n-1,n\}\succsim^{((c,E),CEA)}_{n}\{i,n\}$.
	\item[\emph{(ii)}] If $\alpha\geq\delta_1$, then for each $i\in N\setminus\{1,n\}, \{1,n\}\succsim^{((c,E),CEA)}_{n}\{i,n\}$.
\end{itemize}
\end{lemma}

\begin{proof}[\textbf{Proof}]Let $N$, $(c,E)$, and $\succsim^{((c,E),CEA)}$ be as specified in the lemma. Given that $(c,E)\in {\cal P}_{2}^{N}$, we have $E_{\{1,n\}} = \alpha (c_1+c_n)$ and $E_{\{n-1,n\}} = \alpha (c_{n-1}+c_n)$. Assume that $\beta_1 \leq\alpha \leq \gamma_1$, which implies $$(2-\alpha)c_1\leq \alpha c_n\leq (2-\alpha)c_{n-1}.$$ This, together with similar arguments to those in the proofs of Lemmas~\ref{corollary:alpha-very-small} and \ref{corollary:alpha-very-large}, implies
$$ \frac{\alpha (c_1 + c_n)}{2}\geq c_1\mbox{ and } \frac{\alpha (c_{n-1} + c_n)}{2}\leq c_{n-1}.$$  Therefore, $CEA_{1}(c_{\{1,n\}}, E_{\{1,n\}})=c_1$ and $CEA_{n-1}(c_{\{n-1,n\}}, E_{\{n-1,n\}})=\frac{\alpha(c_{n-1}+c_n)}{2}\leq c_{n-1}$. Hence, $CEA_{n}(c_{\{1,n\}}, E_{\{1,n\}})=\alpha(c_1+c_n)-c_1$ and $CEA_{n}(c_{\{n-1,n\}}, E_{\{n-1,n\}})=\frac{\alpha(c_{n-1}+c_n)}{2}$.  Next, we check when agent $n$ prefers to form a coalition with either agent $1$ or agent $n-1$.
$$CEA_{n}(c_{\{1,n\}}, E_{\{1,n\}})=\alpha(c_1+c_n)-c_1\gtrless \frac{\alpha(c_{n-1}+c_n)}{2}=CEA_{n}(c_{\{n-1,n\}}, E_{\{n-1,n\}})$$
$$\Leftrightarrow\quad \alpha(2c_1+2c_n-c_{n-1}-c_n)\gtrless 2c_1$$
$$\Leftrightarrow\quad \alpha\gtrless \frac{2c_1}{(2c_1-c_{n-1}+c_n)}=\delta_1.$$

It follows that

\begin{itemize}
	\item[(a.1)]If $\alpha\leq\delta_1$, then $CEA_{n}(c_{\{n-1,n\}},E_{\{n-1,n\}}) \geq CEA_{n}(c_{\{1,n\}},E_{\{1,n\}})$.
	\item[(a.2)] If $\alpha>\delta_1$, then $CEA_{n}(c_{\{1,n\}},E_{\{1,n\}}) > CEA_{n}(c_{\{n-1,n\}},E_{\{n-1,n\}})$.
\end{itemize}

To prove the general Cases (i) and (ii) we first establish two facts that do not depend on $\alpha$.\medskip

\noindent \textit{Fact~$1$.} When considering two agents $j,k\in N\setminus\{n\}$ who, with agent $n$, receive an equal share under CEA, agent $n$ prefers to form a coalition with the higher-claim agent.

Let $j,k\in N\setminus\{n\}$, with $c_j\leq c_k$, be such that $CEA_{j}(c_{\{j,n\}},E_{\{j,n\}})=\frac{\alpha (c_j+c_n)}{2}$ and $CEA_{k}(c_{\{k,n\}},E_{\{k,n\}})=\frac{\alpha (c_k+c_n)}{2}$. Then, $CEA_{n}(c_{\{j,n\}},E_{\{j,n\}})=\frac{\alpha (c_j+c_n)}{2}$ and  $CEA_{n}(c_{\{k,n\}},E_{\{k,n\}})=\frac{\alpha (c_k+c_n)}{2}$. Since $\alpha>0$ and $c_j\leq c_k$,
$$CEA_{n}(c_{\{k,n\}},E_{\{k,n\}})\geq CEA_{n}(c_{\{j,n\}},E_{\{j,n\}}).$$

\noindent \textit{Fact~$2$.} When considering two agents $j,k\in N\setminus\{n\}$ who, with agent $n$, receive their claim under CEA, agent $n$ prefers to form a coalition with the lower-claim agent.

Let $j,k\in N\setminus\{n\}$, with $c_j\leq c_k$, be such that $CEA_{j}(c_{\{j,n\}},E_{\{j,n\}})=c_j$ and $CEA_{k}(c_{\{k,n\}},E_{\{k,n\}})=c_k$. Then, $CEA_{n}(c_{\{j,n\}},E_{\{j,n\}})=\alpha (c_j+c_n)-c_j=\alpha c_n -(1-\alpha)c_j$ and  $CEA_{n}(c_{\{k,n\}},E_{\{k,n\}})=\alpha (c_k+c_n)-c_k=\alpha c_n -(1-\alpha)c_k$. Since $(1-\alpha)>0$ and $c_j\leq c_k$,
$$CEA_{n}(c_{\{j,n\}},E_{\{j,n\}})\geq CEA_{n}(c_{\{k,n\}},E_{\{k,n\}}).$$

We are now ready to prove (i) and (ii).\medskip

\noindent (i) Let $\alpha\leq\delta_1$ and $i\in N\setminus\{n-1,n\}$.

If $CEA_{i}(c_{\{i,n\}},E_{\{i,n\}})=\frac{\alpha (c_i+c_n)}{2}$, then $CEA_{n-1}(c_{\{n-1,n\}},E_{\{n-1,n\}})=\frac{\alpha (c_{n-1}+c_n)}{2}$ and by Fact~1, $CEA_{n}(c_{\{n-1,n\}},E_{\{n-1,n\}})\geq CEA_{n}(c_{\{i,n\}},E_{\{i,n\}}).$

If $CEA_{i}(c_{\{i,n\}},E_{\{i,n\}})=c_i$, then $CEA_{1}(c_{\{1,n\}},E_{\{1,n\}})=c_1$ and either by $i=1$ or by Fact~2, $CEA_{n}(c_{\{1,n\}},E_{\{1,n\}})\geq CEA_{n}(c_{\{i,n\}},E_{\{i,n\}}).$ This, together with (a.1), implies $CEA_{n}(c_{\{n-1,n\}},E_{\{n-1,n\}})\geq CEA_{n}(c_{\{i,n\}},E_{\{i,n\}}).$ Hence, $$\{n-1,n\}\succsim^{((c,E),CEA)}_{n}\{i,n\}.$$

\noindent (ii) Let $\alpha>\delta_1$ and $i\in N\setminus\{1,n\}$.

If $CEA_{i}(c_{\{i,n\}},E_{\{i,n\}})=c_i$, then $CEA_{1}(c_{\{1,n\}},E_{\{1,n\}})=c_1$ and by Fact~2, $CEA_{n}(c_{\{1,n\}},E_{\{1,n\}})\geq CEA_{n}(c_{\{i,n\}},E_{\{i,n\}}).$

If  $CEA_{i}(c_{\{i,n\}},E_{\{i,n\}})=\frac{\alpha (c_i+c_n)}{2}$, then $CEA_{n-1}(c_{\{n-1,n\}},E_{\{n-1,n\}})=\frac{\alpha (c_{n-1}+c_n)}{2}$ and either by $i=n-1$ or by Fact~1, $CEA_{n}(c_{\{n-1,n\}},E_{\{n-1,n\}})\geq CEA_{n}(c_{\{i,n\}},E_{\{i,n\}}).$ This, together with (a.2), implies $CEA_{n}(c_{\{1,n\}},E_{\{1,n\}})\geq CEA_{n}(c_{\{i,n\}},E_{\{i,n\}}).$ Hence, \[\{1,n\}\succsim^{((c,E),CEA)}_{n}\{i,n\}.\qedhere\]
\end{proof}

We finally show that parameters $\beta_1$, $\gamma_1$, and $\delta_1$ as defined in Lemmas~\ref{corollary:alpha-very-small}, \ref{corollary:alpha-very-large}, and \ref{lemma:threshold} satisfy
\begin{equation}\label{eq:beta-delta-gamma}\beta_1\leq \delta_1\leq \gamma_1.
\end{equation}

$$\beta_1\leq \delta_1$$
$$\Leftrightarrow\quad\frac{2c_1}{(c_1+c_n)}\leq\frac{2c_1}{(2c_1 -c_{n-1} +c_n)}$$
$$\Leftrightarrow\quad\frac{2c_1}{(c_1+c_n)}\leq\frac{2c_1}{(c_1+c_n) - (\underbrace{c_{n-1} - c_1}_{\geq 0})}$$
and
$$  \delta_1\leq \gamma_1$$
$$\frac{2c_1}{(2c_1 -c_{n-1} +c_n)}\leq\frac{2c_{n-1}}{(c_{n-1}+c_n)}$$
$$\Leftrightarrow\quad 2c_1(c_{n-1}+c_n)\leq 2c_{n-1}(2c_1 -c_{n-1} +c_n)$$
$$\Leftrightarrow\quad c_1c_{n-1}+c_1c_n\leq 2c_1c_{n-1} -c_{n-1}c_{n-1} +c_{n-1}c_n$$
$$\Leftrightarrow\quad 0\leq c_1c_{n-1} -c_1c_n-c_{n-1}c_{n-1} +c_{n-1}c_n$$
$$\Leftrightarrow\quad 0\leq c_{n-1}(c_n-c_{n-1})-c_1(c_n-c_{n-1} )$$
$$\Leftrightarrow\quad 0\leq (\underbrace{c_{n-1}-c_1}_{\geq 0})(\underbrace{c_n-c_{n-1}}_{\geq 0})$$

We are now ready to prove Theorem~\ref{theorem:CEA}.

\begin{proof}[\textbf{Proof of Theorem~\ref{theorem:CEA}}] Let $N\in{\cal N}$ such that $|N|>2$ and $(c,E)\in{\cal P}_{2}^{N}$. Consider the coalition formation problem with agent set $N$ induced by $((c,E),CEA)$. We show that the partition $\pi=\{S_1,\dots,S_l\}$ obtained by the CEA algorithm is stable.\medskip

First, we show that the coalitions obtained by the CEA algorithm cannot be blocked by coalitions of size smaller than or equal to two. \medskip

\noindent\textbf{Step~$\bm{1}$.} Recall that $N_1:=N$, $|N_1|>2$, and according to Cases~(i) or (ii) at Step~$1$ of the CEA algorithm, $S_1\in\{\{n-1,n\},\{1,n\}\}$. We show that in either case, $S_1$ is for each of its members at least as desirable as any other coalition of size smaller than or equal to two.\medskip

\noindent \textit{Case~\emph{(i)}.} $\alpha\leq\delta_1$.
Note that since $(c,E)\in \mathcal{P}_{2}^{N}$, we have that for each $j\in\{n-1,n\}$, $$\{n-1,n\}\succ^{((c,E),CEA)}_{j}\{j\}.$$
Next, let $j\in\{n-1,n\}$ and $i\in N_1\setminus\{n-1,n\}$.
We prove $\{n-1,n\}\succsim^{((c,E),CEA)}_{j}\{i,j\}$ by showing that \begin{equation}\label{eqCEAi}
CEA_j(c_{\{n-1,n\}}, E_{\{n-1,n\}})\geq CEA_j(c_{\{i,j\}}, E_{\{i,j\}}).
\end{equation}
For $j=n-1$, inequality (\ref{eqCEAi}) follows from Lemma~\ref{lemma:top}. For $j=n$, inequality (\ref{eqCEAi}) follows from Lemmas~\ref{corollary:alpha-very-small} and \ref{lemma:threshold}~(i).\medskip

\noindent \textit{Case~\emph{(ii)}.} $\alpha>\delta_1$.
Note that since $(c,E)\in \mathcal{P}_{2}^{N}$, we have that for each $j\in\{1,n\}$, $$\{1,n\}\succ^{((c,E),CEA)}_{j}\{j\}.$$
Next, let $j\in\{1,n\}$ and $i\in N_1\setminus\{1,n\}$.
We prove $\{1,n\}\succsim^{((c,E),CEA)}_{j}\{i,j\}$ by showing that \begin{equation}\label{eqCEAii}
CEA_j(c_{\{1,n\}}, E_{\{1,n\}})\geq CEA_j(c_{\{i,j\}}, E_{\{i,j\}}).
\end{equation}
For $j=1$, inequality (\ref{eqCEAii}) follows from Lemma~\ref{lemma:top}. For $j  =n$, inequality (\ref{eqCEAi}) follows from Lemmas~\ref{corollary:alpha-very-large} and \ref{lemma:threshold}~(ii).

In particular, there exists no coalition $T\subseteq N_1$ such that $|T|\leq 2$, and for each $i\in S_1\cap T$, $T\succ^{((c,E),CEA)}_{i}S_1$, i.e., an agent from set $S_1$ cannot be part of a blocking coalition of size smaller than or equal to two with an agent from set $N_1$.\medskip

\noindent\textbf{Step~$\bm{k}$ ($\bm{k>1}$).} Recall from Step~$k-1$ that $N_k:=N\setminus\left(\cup_{j=1}^{k-1}S_{j}\right)$ and $|N_k|> 2$. Furthermore, agents in $N_k$ are relabeled such that $N_k=\{1',\ldots,n'\}$, $c_{1'}\leq\ldots\leq c_{n'}$, and $\delta_k=\frac{2c_{1'}}{2c_{1'} -c_{(n-1)'} +c_{n'}}$. Then, according to Cases~(i) or (ii) at Step~$k$ of the CEA algorithm, $S_k\in\{\{(n-1)',n'\},\{1',n'\}\}$. Hence, by a similar reasoning than at Step~1 (with agents $1'$, $(n-1)'$, and $n'$ in the roles of agents $1$, $n-1$, and $n$ respectively), it follows that there exists no coalition $T\subseteq N_k$ such that $|T|\leq 2$, and for each $i\in S_k\cap T$, $T\succ^{((c,E),CEA)}_{i}S_k$, i.e., an agent from set $S_k$ cannot be part of a blocking coalition with an agent from set $N_k$. By the previous steps, an agent from set $\bigcup_{j\in\{1,\ldots,k-1\}}S_j$ cannot be part of a blocking coalition with an agent from set $S_k$. Hence, no agent from set $S_k$ can be part of a blocking coalition of size smaller than or equal to two.\medskip

After $l-1$ steps, we have shown that there is no blocking coalition $T\subsetneq N$ such that $|T|\leq 2$.\medskip

Second, assume, by contradiction, that $\pi$ is not stable. Then, there exists a blocking coalition $T\subseteq N$ such that for each agent $i\in T$, $T \succ^{((c,E),CEA)}_{i}\pi(i)$. In particular, $|T|> 2$.\medskip

Then, by Proposition~\ref{prop:toptheta} (Appendix~\ref{appendix:proofTheorem6}), there is a coalition $T'\subsetneq T$ with $|T'|=2$ such that for each $j\in T'$, $$T'\succsim^{((c,E),CEA)}_{j}T\succ^{((c,E),CEA)}_{j}\pi(j),$$ which contradicts the fact that there is no blocking coalition of size smaller than or equal to two.\medskip

We have proven that partition $\pi$ obtained by the CEA algorithm is stable.\end{proof}

\section[Cases and parameters]{Cases and parameters in the $\bm{\theta}$-CEA algorithm}\label{appendix:conditionalgorithm}

Recall that $N=\{1,\dots, n\}$ and $c_1\leq c_2\leq \cdots\leq c_n$. Throughout this appendix, we consider sets of cardinalities $2,\ldots,\theta$ and proportional generalized claims problems, i.e., problems in $\mathcal{P}^{N}$ (without imposing a minimal coalition size of $\theta$).\medskip

Note that in the $\theta$-CEA set algorithm, starting with agent $n$, we construct a set of cardinality $\theta$ by adding agents one by one. More precisely, at each Step~$k$, a set of agents $S'_{k-1}$ considers to add either the lowest-label or the highest-label agent of the remaining set of agents. The results in this appendix show that the agent who is added, according to $\theta$-CEA set algorithm Case (i) or (ii), is weakly preferred by the agents in $S'_{k-1}$ to any other agent who could have been added.\medskip

Let $S\subsetneq N$, $S\neq\emptyset$, and consider agents $i,j\in N\setminus S$ such that $i<j$, i.e., $c_i\leq c_j$. We first show that if both agents receive their claim in coalitions $S\cup\{i\}$ and $S\cup\{j\}$ respectively, then all agents in $S$  weakly prefer to form a coalition with agent~$i$.

\begin{lemma}\label{lemma:thetaalpha-very-large}
Let $N\in{\cal N}$, $(c,E)\in\mathcal{P}^{N}$, and $\succsim^{((c,E),CEA)}$ be the coalition formation problem with agent set $N$ induced by $((c,E),CEA)$. Let $S\subsetneq N$, $S\neq\emptyset$, and $j\in N\setminus S$ such that $$CEA_{j}(c_{(S\cup\{j\})}, E_{(S\cup\{j\})})=c_j.$$

\noindent Then, for each $i\in N\setminus(S\cup\{j\})$ such that $i< j$,

$$CEA_{i}(c_{(S\cup\{i\})}, E_{(S\cup\{i\})})=c_i,$$
and for each $k\in S$,
\begin{equation}\label{eq1:thetaalpha-very-large}
S\cup\{i\}\succsim^{((c,E),CEA)}_{k} S\cup\{j\}.
\end{equation}
\end{lemma}

\begin{proof} [\textbf{Proof}] Let $N$, $(c,E)$, $\succsim^{((c,E),CEA)}$, $S\subsetneq N$ and $i,j\in N\setminus S$ be as specified in the lemma. Hence, $c_i\leq c_j$. Given that $(c,E)\in {\cal P}^{N}$, we have $E_{(S\cup\{i\})} = \alpha c^{(S\cup\{i\})}$ and $E_{(S\cup\{j\})} = \alpha c^{(S\cup\{j\})}$.

Starting from problem $(c_{(S\cup\{j\})}, E_{(S\cup\{j\})})\in\mathcal{C}^{(S\cup\{j\})}$, we replace agent $j$ (with claim $c_j$) with agent $i$ (with claim $c_i$) but without changing the CEA payoffs of agents in $S$ in the resulting problem in $\mathcal{C}^{(S\cup\{i\})}$. That is, we consider the auxiliary  problem $(c_{(S\cup\{i\})}, \bar{E}_{(S\cup\{i\})})\in\mathcal{C}^{(S\cup\{i\})}$ such that for each $k\in S$, $CEA_k(c_{(S\cup\{i\})}, \bar{E}_{(S\cup\{i\})})=CEA_k(c_{(S\cup\{j\})}, E_{(S\cup\{j\})})$, and $CEA_i(c_{(S\cup\{i\})}, \bar{E}_{(S\cup\{i\})})=c_i.$ Hence, $$\bar{E}_{(S\cup\{i\})}=\sum_{k\in S} CEA_k(c_{(S\cup\{j\})}, E_{(S\cup\{j\})})+c_i=E_{(S\cup\{j\})}-c_j + c_i.$$  Note that $(c_{(S\cup\{i\})}, \bar{E}_{(S\cup\{i\})})$ is such that $\bar{E}_{(S\cup\{i\})}\neq \alpha c^{(S\cup\{i\})}$.
Furthermore, by construction, $\lambda_{\bar{E}_{(S\cup\{i\})}}=\lambda_{E_{(S\cup\{j\})}}$.\\

Finally, consider problem $(c_{(S\cup\{i\})}, E_{(S\cup\{i\})})\in\mathcal{C}^{(S\cup\{i\})}$. We have that $E_{(S\cup\{i\})} = E_{(S\cup\{j\})}-\alpha c_j +\alpha c_i$. Therefore, since $c_j-c_i\geq 0$ and $\alpha\in (0,1)$,
$$E_{(S\cup\{i\})} =E_{(S\cup\{j\})}-\alpha (c_j -c_i)\geq E_{(S\cup\{j\})}-(c_j -c_i)=\bar{E}_{(S\cup\{i\})}.$$
Then, by definition of the CEA rule, $\lambda_{E_{(S\cup\{i\})}}\geq\lambda_{\bar{E}_{(S\cup\{i\})}}=\lambda_{E_{(S\cup\{j\})}}.$ Thus, we have that
\begin{equation*}
CEA_i(c_{(S\cup\{i\})}, E_{(S\cup\{i\})})=c_i=CEA_i(c_{(S\cup\{i\})}, \bar{E}_{(S\cup\{i\})}),
\end{equation*}
and for each $k\in S$,
\begin{equation*}CEA_{k}(c_{(S\cup \{i\})}, E_{(S\cup\{i\})})\geq CEA_k(c_{(S\cup\{i\})}, \bar{E}_{(S\cup\{i\})})= CEA_{k}(c_{S\cup\{j\}}, E_{(S\cup\{j\})}),
\end{equation*}
or equivalently,
\begin{equation*}
S\cup\{i\}\succsim^{((c,E),CEA)}_{k} S\cup\{j\}.
\qedhere
\end{equation*}\end{proof}

Lemma~\ref{lemma:thetaalpha-very-large} implies that, given any coalition $S\subsetneq N$, among all agents receiving their claim when added, coalition $S$ weakly prefers to add a lowest-claim agent. In particular, if there are several such agents, we add the agent with the lowest label.\medskip

Let $S\subsetneq N$, $S\neq\emptyset$, and consider agents $i,j\in N\setminus S$ such that $i<j$, i.e., $c_i\leq c_j$. We second show that if both agents do not receive their claim in coalitions $S\cup\{i\}$ and $S\cup\{j\}$ respectively, then all agents in $S$ weakly prefer to form a coalition with agent~$j$.

\begin{lemma}\label{lemma:thetaalpha-very-small}
Let $N\in{\cal N}$, $(c,E)\in\mathcal{P}^{N}$, and $\succsim^{((c,E),CEA)}$ be the coalition formation problem with agent set $N$ induced by $((c,E),CEA)$. Let $S\subsetneq N$, $S\neq\emptyset$, and $i\in N\setminus S$ such that $$CEA_{i}(c_{(S\cup\{i\})}, E_{(S\cup\{i\})})=\lambda_{E_{(S\cup\{i\})}}<c_i.$$
\noindent Then, for each $j\in N\setminus(S\cup\{i\})$ such that $i< j$,
$$CEA_{j}(c_{(S\cup\{j\})}, E_{(S\cup\{j\})})=\lambda_{E_{(S\cup\{j\})}}<c_j,$$
and for each $k\in S$,
 \begin{equation}\label{eq1:thetaalpha-very-small}
S\cup\{j\}\succsim^{((c,E),CEA)}_{k} S\cup\{i\}.
\end{equation}
\end{lemma}

\begin{proof} [\textbf{Proof}] Let $N$, $(c,E)$, $\succsim^{((c,E),CEA)}$, $S\subsetneq N$ and $i,j\in N\setminus S$ be as specified in the lemma. Hence, $c_i\leq c_j$. Given that $(c,E)\in {\cal P}^{N}$, we have $E_{(S\cup\{i\})} = \alpha c^{(S\cup\{i\})}$ and $E_{(S\cup\{j\})} = \alpha c^{(S\cup\{j\})}$.
Suppose that $CEA_j(c_{(S\cup\{j\})}, E_{(S\cup\{j\})})=c_j$. Then, by Lemma~\ref{lemma:thetaalpha-very-large}, $CEA_i(c_{(S\cup\{i\})}, E_{(S\cup\{i\})})= c_i$, contradicting our assumption that $CEA_i(c_{(S\cup\{i\})}, E_{(S\cup\{i\})})=\lambda_{E_{(S\cup\{i\})}}< c_i$. Hence,
\begin{equation*}
CEA_j(c_{(S\cup\{j\})}, E_{(S\cup\{j\})})=\lambda_{E_{(S\cup\{j\})}}<c_j.
\end{equation*}

Starting from problem $(c_{(S\cup\{i\})}, E_{(S\cup\{i\})})\in\mathcal{C}^{(S\cup\{i\})}$, we replace agent $i$ (with claim $c_i$) with agent $j$ (with claim $c_j$) without changing the endowment. That is, we consider the auxiliary problem $(c_{(S\cup\{j\})}, E_{(S\cup\{i\})})\in\mathcal{C}^{(S\cup\{j\})}$. Note that $(c_{(S\cup\{j\})}, E_{(S\cup\{i\})})$ is such that $E_{(S\cup\{i\})}\leq \alpha c^{(S\cup\{j\})}$. Since $CEA_i(c_{(S\cup\{i\})}, E_{(S\cup\{i\})})=\lambda_{E_{(S\cup\{i\})}}< c_i$ and $c_i \leq c_j,$
$$CEA_i(c_{(S\cup\{i\})}, E_{(S\cup\{i\})})=\lambda_{E_{(S\cup\{i\})}}=CEA_j(c_{(S\cup\{j\})}, E_{(S\cup\{i\})})$$
and for each agent $k\in S$, $$CEA_{k}(c_{(S\cup\{i\})}, E_{(S\cup\{i\})})=CEA_{k}(c_{(S\cup\{j\})}, E_{(S\cup\{i\})}).$$

Finally, consider problem $(c_{(S\cup\{j\})}, E_{(S\cup\{j\})})\in\mathcal{C}^{(S\cup\{j\})}$. Since $E_{(S\cup\{j\})}=\alpha c^{(S\cup\{j\})}\geq \alpha c^{(S\cup\{i\})}=E_{(S\cup\{i\})}$, by the definition of the CEA rule, we have that $\lambda_{E_{(S\cup\{j\})}}\geq \lambda_{E_{(S\cup\{i\})}}$. Thus, for each $k\in S$,
\begin{equation*}CEA_{k}(c_{(S\cup \{j\})}, E_{(S\cup\{j\})})\geq CEA_{k}(c_{S\cup\{i\}}, E_{(S\cup\{i\})}),
\end{equation*}
or equivalently,
\begin{equation*}S\cup\{j\}\succsim^{((c,E),CEA)}_{k} S\cup\{i\}.\qedhere
\end{equation*}
\end{proof}

Lemma~\ref{lemma:thetaalpha-very-small} implies that, given any coalition $S\subsetneq N$, among all the agents who do not receive their claim when added, coalition $S$ weakly prefers to add a highest-claim agent. In particular, if there are several such agents, we add the agent with the highest label.\medskip

To summarize, by Lemmas~\ref{lemma:thetaalpha-very-large} and \ref{lemma:thetaalpha-very-small}, a coalition $S\subsetneq N$ weakly prefers to add either the lowest-label agent or the highest-label agent to their coalition. Finally, we determine which of these two agents the coalition weakly prefers to add.\medskip

Let $S\subsetneq N$, $S\neq\emptyset$, and consider agents $i,j\in N\setminus S$ such that $i<j$, i.e., $c_i\leq c_j$.  We next show with which agent, $i$ or $j$, all agents in $S$ weakly prefer to form a coalition if agent $i$ receives her claim in coalition $S\cup\{i\}$ and agent $j$ does not receive her claim in coalition $S\cup\{j\}$.

\begin{lemma}\label{lemma:thresholdtheta}
Let $N\in{\cal N}$, $(c,E)\in\mathcal{P}^{N}$, and $\succsim^{((c,E),CEA)}$ be the coalition formation problem with agent set $N$ induced by $((c,E),CEA)$. Let $S\subsetneq N$, $S\neq\emptyset$, and $i,j\in N\setminus S$ be such that $i< j$ and $$CEA_{i}(c_{(S\cup\{i\})}, E_{(S\cup\{i\})})=c_i,$$ and $$CEA_{j}(c_{(S\cup\{j\})}, E_{(S\cup\{j\})})=\lambda_{E_{(S\cup\{j\})}}< c_j.$$
Hence, $c_i< c_j$.
\begin{itemize}
	\item[\emph{(i)}]If $\lambda_{E_{(S\cup\{j\})}}\leq (1-\alpha)c_i +\alpha c_j$, then for each $k\in S$, $S\cup\{j\}\succsim^{((c,E),CEA)}_{k} S\cup\{i\}.$
	\item[\emph{(ii)}] If $\lambda_{E_{(S\cup\{j\})}}> (1-\alpha)c_i +\alpha c_j$, then for each $k\in S$, $S\cup\{i\}\succsim^{((c,E),CEA)}_{k} S\cup\{j\}.$
\end{itemize}
\end{lemma}

\begin{proof}[\textbf{Proof}]Let $N$, $(c,E)$, $\succsim^{((c,E),CEA)}$, $S\subsetneq N$ and $i,j\in N\setminus S$ be as specified in the lemma. Given that $(c,E)\in {\cal P}^{N}$, we have $E_{(S\cup\{i\})} = \alpha c^{(S\cup\{i\})}$ and $E_{(S\cup\{j\})} = \alpha c^{(S\cup\{j\})}$. Since agent $i$ in coalition $S\cup\{i\}$ receives $c_i$ but only contributes $\alpha c_i$, she receives a subsidy of $ (1-\alpha)c_i> 0$. Agent $j$ in coalition $S\cup\{j\}$ receives $\lambda_{E_{(S\cup\{j\})}}< c_j$ and contributes $\alpha c_j$. Thus, if $\lambda_{E_{(S\cup\{j\})}}\leq \alpha c_j$,\footnote{Hence, $\lambda_{E_{(S\cup\{j\})}}< (1-\alpha)c_i+\alpha c_j$ and the premise of Case~(i) is satisfied.} agent $j$ weakly transfers $\alpha c_j - \lambda_{E_{(S\cup\{j\})}}\geq 0$, which implies that each agent in $S$ weakly prefers $S\cup\{j\}$ to $S\cup\{i\}$. Otherwise, if $\lambda_{E_{(S\cup\{j\})}}> \alpha c_j$, agent $j$ receives a subsidy of $\lambda_{E_{(S\cup\{j\})}}-\alpha c_j>0$. Hence, we distinguish two cases:\medskip

\noindent (i) Agent $j$ weakly transfers $\alpha c_j - \lambda_{E_{(S\cup\{j\})}}\geq 0$ or receives a subsidy of $\lambda_{E_{(S\cup\{j\})}}-\alpha c_j>0$ that is weakly smaller than the subsidy $(1-\alpha)c_i$ of agent $i$. Thus, $\lambda_{E_{(S\cup\{j\})}}-\alpha c_j \leq (1-\alpha)c_i$, which is equivalent to $\lambda_{E_{(S\cup\{j\})}}\leq (1-\alpha)c_i +\alpha c_j$. In this case, for each $k\in S$, $S\cup\{j\}\succsim^{((c,E),CEA)}_{k} S\cup\{i\}$.\medskip

\noindent (ii) Agent $i$ receives a  smaller subsidy than agent $j$.

Then, $(1-\alpha)c_i< \lambda_{E_{(S\cup\{j\})}}-\alpha c_j$, which is equivalent to $\lambda_{E_{(S\cup\{j\})}}> (1-\alpha)c_i +\alpha c_j$. In this case, for each $k\in S$, $S\cup\{i\}\succsim^{((c,E),CEA)}_{k} S\cup\{j\}$.\medskip
\end{proof}

Let $S\subsetneq N$. Define $L\equiv\{i'\in N\setminus S : CEA_{i'}(c_{(S\cup\{i'\})}, E_{(S\cup\{i'\})})=c_{i'}\}$ and $R\equiv\{j'\in N\setminus S: CEA_{j'}(c_{(S\cup\{j'\})}, E_{(S\cup\{j'\})})=\lambda_{E_{(S\cup\{j'\})}}\neq c_{j'}\}$. Consider $i\equiv \min L$ and $j\equiv \max R$. Lemmas~\ref{lemma:thetaalpha-very-large} and \ref{lemma:thetaalpha-very-small} together imply that agents $i$ and $j$ are potential ``stable partners'' for coalition $S$. If $R=\emptyset$, then Lemma~\ref{lemma:thetaalpha-very-large} implies that agent $i$ is added to coalition $S$. Similarly, if $L=\emptyset$, then Lemma~\ref{lemma:thetaalpha-very-small} implies that agent $j$ is added to coalition $S$. Finally, if $L\neq\emptyset\neq R$, Lemma~\ref{lemma:thresholdtheta}, depending on Cases (i) or (ii), together with Lemmas~\ref{lemma:thetaalpha-very-large} and \ref{lemma:thetaalpha-very-small}, determines which of the two agents, $i$ or $j$, is added to coalition $S$.

\section{Proof of Theorem~\ref{theorem:CEAtheta}}\label{appendix:proofTheorem7}

Recall that $N=\{1,\dots, n\}$ and $c_1\leq c_2\leq \cdots\leq c_n$. For each $S\subseteq N$, we denote the CEA parameter associated with $(c_S,E_S)$ by $\lambda_{E_S}$, i.e., for each $i\in S$, $CEA_{i}(c_S,E_S)=\min\{c_i,\lambda_{E_S}\}$, where $\lambda_{E_S}$ is chosen so that $\sum_{j\in S}\min\{c_j,\lambda_{E_S}\}= E_S$.\medskip

We introduce some notation before the proof. For $S\subseteq N$ and $CEA(c_S, E_S)$, we denote the set of agents receiving an over-proportional or proportional payoff in coalition $S$, also called the \textit{set of \textbf{r}ecipients}, by $R_{S}\equiv\{i\in S: CEA_i (c_S, E_S)\geq \alpha c_i\}$. Moreover, for each agent $i\in R_{S}$, we denote the subsidy she receives in coalition $S$ as $r_i^{S}=CEA_i (c_S, E_S)-\alpha c_i\geq 0$. Similarly, we denote the set of agents receiving an under-proportional payoff in coalition $S$, also called the \textit{set of \textbf{t}ransfer agents}, by $T_{S}\equiv\{i\in S: CEA_i (c_S, E_S)< \alpha c_i\}$. Furthermore, for each agent $i\in T_S$, we denote the transfer she makes to coalition $S$ as $t_i^{S}=\alpha c_i-CEA_i (c_S, E_S)>0$.\smallskip

Note that for any coalition $S\subseteq N$, $S\neq \emptyset$, by the definition of the CEA rule, $R_{S}\neq \emptyset$ and [$T_{S}=\emptyset$ if and only if for all $i,j\in S$, $c_i=c_j$].

\begin{proof}[\textbf{Proof of Theorem~\ref{theorem:CEAtheta}}] Let $N\in{\cal N}$ such that $|N|> \theta$ and $(c,E)\in{\cal P}_{\theta}^{N}$. Consider the coalition formation problem with agent set $N$ induced by $((c,E),CEA)$. We show that the partition $\pi=\{S_1,\dots,S_l\}$ obtained by the $\theta$-CEA algorithm is stable.\medskip

First, we show that the coalitions obtained by the $\theta$-CEA algorithm cannot be blocked by coalitions of size smaller than or equal to $\theta$. \medskip

\noindent\textbf{Step~$\bm{1}$.} Recall that $N_1:=N$, $|N_1|>\theta$. We show that $S_1$
cannot be blocked by a coalition of size smaller than or equal to $\theta$. Note that since $(c,E)\in \mathcal{P}_{\theta}^{N}$, we have that for each $S\subsetneq N$ with $|S|<\theta$ and each $k\in S_1\cap S$, $$S_{1}\succsim^{((c,E),CEA)}_{k}S.$$

Next, assume, by contradiction, that $S_1$ can be blocked by a coalition of size $\theta$, i.e., there exists a set $S\subsetneq N$ with $|S|=\theta$ such that $S\cap S_1\neq \emptyset$ and for each $k\in S\cap S_1$,
 $$S\succ^{((c,E),CEA)}_{k} S_1.$$
 Note that no agent receiving her claim in $S_1$ can be in $S$. Hence, $$\lambda_{E_S}>\lambda_{E_{S_1}}.$$

Let $i^*$ be the agent with the highest label in $S_1$ receiving $c_i$, i.e., $i^{*}\equiv\max \{i\in S_1 : CEA_{i}(c_{S_1}, E_{S_1})=c_{i}\}$ and $j^*$ be the agent with the lowest label in $S_1$ receiving $\lambda_{E_{S_1}}\geq\alpha c_{j^*}$, $\lambda_{E_{S_1}}\neq c_{j^*}$, i.e., $j^*\equiv\min \{j\in S_1 : CEA_{j}(c_{S_1}, E_{S_1})=\lambda_{E_{S_1}}\geq\alpha c_{j}\mbox{ and }\lambda_{E_{S_1}}\neq c_{j}\}$. Since $R_{S_1}\neq\emptyset$, $\{i^*,j^*\}\neq\emptyset$, i.e., at least one of these agents (or both) exists.

Next, we show that if both agents $i^*$ and $j^*$ exist, then there is no agent in $S_1$ whose label is in between that of these two agents. Assume, by contradiction, that $i^*$ and $j^*$ exist and that there is an agent $l\in S_1$ such that $i^*<l< j^*$; thus, $c_{i^*}\leq c_l\leq c_{j^*}$. Since $i^{*}<l$, $CEA_{l}(c_{S_1}, E_{S_1})=\lambda_{E_{S_1}}\neq c_l$. Since $l<j^*$, $CEA_{l}(c_{S_1}, E_{S_1})=\lambda_{E_{S_1}}< \alpha c_l$. Since $c_l\leq c_{j^*}$, we have $\lambda_{E_{S_1}}< \alpha c_l\leq  \alpha c_{j^*}                                                                                                                                                                                                                               $, contradicting $\lambda_{E_{S_1}}\geq  \alpha c_{j^*}$.

\begin{figure}[ht!]
	\centering
\begin{tikzpicture}[scale=1]
	\tikzset{lw/.style = {line width=1pt}}
	\draw [-, line width=0.45mm](0,0) -- (6,0);
	\foreach \x in {0}
	\draw[-, line width=0.45mm] (\x cm,3pt) -- (\x cm,-3pt);
	\foreach \x in {0.5}
	\draw[-, line width=0.45mm] (\x cm,3pt) -- (\x cm,-3pt);
	\foreach \x in {1}
	\draw[-, line width=0.45mm] (\x cm,3pt) -- (\x cm,-3pt);
	\foreach \x in {3}
	\draw[-] (\x cm,1.5pt) -- (\x cm,-1.5pt);
	\foreach \x in {4}
	\draw[-] (\x cm,1.5pt) -- (\x cm,-1.5pt);
	\foreach \x in {5}
	\draw[-, line width=0.45mm](\x cm,3pt) -- (\x cm,-3pt);
	\foreach \x in {5.5}
	\draw[-, line width=0.45mm](\x cm,3pt) -- (\x cm,-3pt);
	\foreach \x in {6}
	\draw[-, line width=0.45mm] (\x cm,3pt) -- (\x cm,-3pt);
	\draw (1,0) node[below=3pt] {$c_{i^*}$} node[above=3pt] {$   $};
	\draw (3,0) node[below=3pt] {$\alpha c_{j^*} $} node[above=3pt] {$   $};
	\draw (4,0) node[below=3pt] {$\lambda_{E_{S_{1}}}  $} node[above=3pt] {$   $};
	\draw (5,0) node[below=3pt] {$c_{j^*}  $} node[above=3pt] {$   $};
	\end{tikzpicture}%
\caption{Both agents, $i^*$ and $j^*$, exist.}
\end{figure}
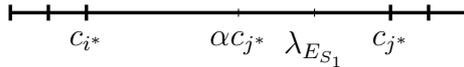

We divide the rest of the proof into four parts.\bigskip

\noindent\textbf{Part~1.} \emph{We show that all agents in $S\setminus S_1\neq\emptyset$ are recipients. That is, for each $l\in S\setminus S_1$, $l\in R_{S}$. }\medskip

Recall that $S\neq S_1$ and $|S|=|S_1|=\theta$. Hence, $S\setminus S_1\neq\emptyset$.\smallskip

We distinguish two cases:\smallskip

\noindent\textit{Case~$1$.} Agent $j^*$ exists, i.e., $j^*$ is the agent with the lowest label in $S_1$ receiving $\lambda_{E_{S_1}}\geq\alpha c_{j^*}$ ($\lambda_{E_{S_1}}\neq c_{j^*}$).\smallskip

Let $l\in S\setminus S_1$. Note first that $l<j^*$ because, by construction of coalition $S_1$ by the $\theta$-CEA set algorithm, each agent with a higher label than $j^*$ is necessarily in $S_1$.\footnote{Note that coalition $S_1$ constructed by the algorithm is such that each agent with a higher label than $j^*$ is either a transfer agent or a recipient with a smaller subsidy, which means that they are chosen before $j^*$ in the algorithm.} Therefore, $c_l\leq c_{j^*}$. Recall that $\lambda_{E_S}>\lambda_{E_{S_1}}\geq \alpha c_{j^*}$. Hence, $$\lambda_{E_S}>\lambda_{E_{S_1}}\geq \alpha c_{j^*} \geq \alpha c_l,$$ which implies $\lambda_{E_S}>\alpha c_l$. Thus, $l\in R_S$.\medskip

\noindent\textit{Case~$2$.} Agent $j^*$ does not exist and $i^*$ is the agent with the highest label in $S_1$ receiving~$c_{i^*}$.\smallskip

Thus, $T_{S_1}\neq\emptyset$. Denote by $j^{t}$ the transfer agent in $S_1$ with the lowest label, i.e., $j^{t}\equiv\min\{j\in T_{S_1}:\lambda_{E_{S_1}}<\alpha c_{j}\}$. Since agent $j^*$ does not exist, agents $i^*$ and $j^{t}$ are adjacent agents in $S_1$, i.e., there exists no $l\in S_1$ such that $i^*<l< j^{t}$.\medskip

\begin{figure}[ht!]
	\centering
\begin{tikzpicture}[scale=1]
	\tikzset{lw/.style = {line width=1pt}}
		\draw [-, line width=0.45mm](0,0) -- (6,0);
	\foreach \x in {0}
	\draw[-, line width=0.45mm] (\x cm,3pt) -- (\x cm,-3pt);
	\foreach \x in {0.5}
	\draw[-, line width=0.45mm] (\x cm,3pt) -- (\x cm,-3pt);
	\foreach \x in {1}
	\draw[-, line width=0.45mm] (\x cm,3pt) -- (\x cm,-3pt);
	\foreach \x in {3}
	\draw[-] (\x cm,1.5pt) -- (\x cm,-1.5pt);
	\foreach \x in {4}
	\draw[-] (\x cm,1.5pt) -- (\x cm,-1.5pt);
	\foreach \x in {5}
	\draw[-, line width=0.45mm](\x cm,3pt) -- (\x cm,-3pt);
	\foreach \x in {5.5}
	\draw[-, line width=0.45mm](\x cm,3pt) -- (\x cm,-3pt);
	\foreach \x in {6}
	\draw[-, line width=0.45mm] (\x cm,3pt) -- (\x cm,-3pt);
		\draw (1,0) node[below=3pt] {$c_{i^*}$} node[above=3pt] {$   $};
	\draw (3,0) node[below=3pt] {$\lambda_{E_{S_{1}}} $} node[above=3pt] {$   $};
	\draw (4,0) node[below=3pt] {$\alpha c_{j^{t}}  $} node[above=3pt] {$   $};
	\draw (5,0) node[below=3pt] {$c_{j^{t}}  $} node[above=3pt] {$   $};
	\end{tikzpicture}%
\caption{Agent $j^*$ does not exist and agent $j^{t}$ is a transfer agent with the lowest label.}
\end{figure}
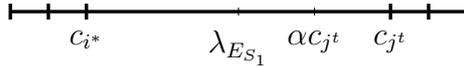

First, we show that in the construction of coalition $S_1$ by the $\theta$-CEA set algorithm, agent $j^t$ is not the last agent added. Assume, by contradiction, that at Step~$\theta - 1$ (i) of the $\theta$-CEA set algorithm, agent $j^t$ has been chosen over agent $i^*+1$. Recall that $\alpha c_{j^t}>\lambda_{E_{S_1}}$ and that therefore, agent $j^t$ is a transfer agent in coalition $S_1$; hence, $\lambda_{E_{S_1}}>\lambda_{E_{S_1\setminus\{j^t\}}}$.\smallskip

\noindent \textit{Case $2.1$.} Assume that at Step~$\theta - 2$ of the $\theta$-CEA set algorithm, agent $i^*$ has been chosen over agent $j^t$. Since at Step~$\theta - 2$, the set of agents $S_1\setminus\{i^*,j^t\}$ decides to add agent $i^*$ instead of agent $j^t$, we have that the set $S_1\setminus\{j^t\}$ obtained when agent $i^*$ is added is better than the set $S_1\setminus\{i^*\}$ obtained when agent $j^t$ is added. Thus, $\lambda_{E_{S_1\setminus\{j^t\}}}>\lambda_{E_{S_1\setminus\{i^*\}}}$. Therefore, $$\alpha c_{j^t}>\lambda_{E_{S_1}}>\lambda_{E_{S_1\setminus\{j^t\}}}> \lambda_{E_{S_1\setminus\{i^*\}}}.$$
Hence, since $\lambda_{E_{S_1\setminus\{i^*\}}}< \alpha c_{j^t}$, $\lambda_{E_{S_1\setminus\{i^*\}}}\leq (1-\alpha)c_{i^*}+\alpha c_{j^t}$, i.e., the inequality for Case~(i) that indicates that agent $j^t$ should have been chosen over agent $i^*$ at Step~$\theta -2$ of the $\theta$-CEA set algorithm is satisfied; a contradiction.\smallskip

\noindent \textit{Case $2.2$.} Assume that at Step~$\theta-2$ of the $\theta$-CEA set algorithm, agent $j^{t}+1$ has been chosen over $i^*+1$. Then, $c_{j^t+1}\geq c_{j^t}$, together with previous inequalities, $\alpha c_{j^t}>\lambda_{E_{S_1}}$ and $\lambda_{E_{S_1}}>\lambda_{E_{S_1\setminus\{j^t\}}}$, implies $\alpha c_{j^t+1}\geq \alpha c_{j^t}> \lambda_{E_{S_1}}>\lambda_{E_{S_1\setminus\{j^t\}}}$. Thus, $\alpha c_{j^t +1}>\lambda_{E_{S_1\setminus\{j^t\}}}$ and therefore agent $j^t +1$ is a transfer agent in $S_1\setminus\{j^t\}$. Hence, \begin{equation}\label{smallk}
\lambda_{E_{S_1\setminus\{j^t\}}}>\lambda_{E_{S_1\setminus\{j^t,j^t+1\}}}.
\end{equation}

More generally, assume that at Step $\theta-k$ ($k> 2$), agent $j^t +(k-1)$ has been chosen over agent $i^*+1$. Then, $c_{j^t + (k-1)}\geq c_{j^t+(k-2)}$, together with previous inequalities, implies $\alpha c_{j^t+(k-1)}\geq \alpha c_{j^t+(k-2)}> \lambda_{E_{S_1\setminus\{j^t,\ldots,j^t+(k-3)\}}}> \lambda_{E_{S_1\setminus\{j^t,\ldots,j^t+(k-2)\}}}$. Thus,
\begin{equation}\label{eq:transfer}
\alpha c_{j^t+(k-1)}>\lambda_{E_{S_1\setminus\{j^t,\ldots,j^t+(k-2)\}}}
\end{equation}
and therefore agent $j^t+(k-1)$ is a transfer agent in ${S_1\setminus\{j^t,\ldots,j^t+(k-2)\}}$. Hence,
\begin{equation}\label{eq:lambda}
\lambda_{E_{S_1\setminus\{j^t,\ldots,j^t+(k-2)\}}}> \lambda_{E_{S_1\setminus\{j^t,\ldots,j^t+(k-1)\}}}.
\end{equation}

Since agent $i^*$ has been added to $S_1$ before agent $j^t$, there is a step of the $\theta$-CEA set algorithm, Step~$\theta - \ell$ ($\ell\geq 3$), at which agent $i^*$ has been chosen over agent $j^t +(\ell-2)$. Let $\bar{S}=S_1\setminus\{j^t,\ldots,j^t+(\ell-2),i^*\}$ be the coalition obtained at Step~$\theta - \ell-1$ of the $\theta$-CEA set algorithm. Since at Step~$\theta - \ell$, the set of agents $\bar{S}$ decides to add agent $i^*$ instead of agent $j^t +(\ell-2)$, we have that the set $\bar{S}\cup\{i^*\}$ obtained when agent $i^*$ is added is better than the set $\bar{S}\cup\{j^t +(\ell-2)\}$ obtained when agent $j^t +(\ell-2)$ is added. Thus, $\lambda_{E_{\bar{S}\cup\{i^*\}}}>\lambda_{E_{\bar{S}\cup\{j^t +(\ell-2)\}}}$, i.e.,
\begin{equation}\label{eq:ell}
\lambda_{E_{S_1\setminus\{j^t,\ldots,j^t+(\ell-2)\}}}> \lambda_{E_{S_1\setminus\{j^t,\ldots,j^t+(\ell-3),i^*\}}}.
\end{equation}

Note that  at Step~$\theta-\ell+1$, agent $j^t +(\ell-2)$ has been chosen over agent $i^*+1$. Thus, by inequalities (\ref{eq:transfer}), (\ref{eq:lambda}), and (\ref{eq:ell}),
$$\alpha c_{j^t +(\ell-2)}>\lambda_{E_{S_1\setminus\{j^t,\ldots,j^t+(\ell-3)\}}}> \lambda_{E_{S_1\setminus\{j^t,\ldots,j^t+(\ell-2)\}}}> \lambda_{E_{S_1\setminus\{j^t,\ldots,j^t+(\ell-3),i^*\}}}.$$

Hence, since $\lambda_{E_{S_1\setminus\{j^t,\ldots,j^t+(\ell-3),i^*\}}}<\alpha c_{j^t +(\ell-2)}$, $\lambda_{E_{S_1\setminus\{j^t,\ldots,j^t+(\ell-3),i^*\}}}\leq (1-\alpha)c_{i^*}+\alpha c_{j^t +(\ell-2)}$, i.e., the inequality for Case~(i) that indicates that agent $j^t +(\ell-2)$ should have been chosen over agent $i^*$ at Step~$\theta -\ell$ of the $\theta$-CEA set algorithm is satisfied; a contradiction. We conclude that in the construction of coalition $S_1$ by the $\theta$-CEA set algorithm, agent $j^t$ is not the last agent added.

Therefore, agent $i^*$ is the last agent added to coalition $S_1$, and by Case~(ii), $$(1-\alpha)c_{i^*}< \lambda_{E_{((S_{1}\setminus\{i^*\})\cup \{j^{t}-1\})}}-\alpha c_{j^{t}-1},$$ which means that if agent $i^*$ is added to coalition $(S_{1}\setminus\{i^*\})$, even if receiving her claim, she would receive a smaller subsidy than if agent $j^{t}-1$ is added to coalition $(S_{1}\setminus\{i^*\})$. This, together with the definition of the CEA rule, implies $$\lambda_{E_{S_{1}}}\geq\lambda_{E_{((S_{1}\setminus\{i^*\})\cup \{j^{t}-1\})}}.$$ Furthermore, since $(1-\alpha)c_{i^*}>0$, we have $\lambda_{E_{((S_{1}\setminus\{i^*\})\cup \{j^{t}-1\})}}-\alpha c_{j^{t}-1}>0$. Thus, $$\lambda_{E_{((S_{1}\setminus\{i^*\})\cup \{j^{t}-1\})}}>\alpha c_{j^{t}-1}.$$

Let $l\in S\setminus S_1$. Note first that $l \leq j^{t}-1$ because, by construction of coalition $S_1$ by the $\theta$-CEA set algorithm, each agent with a higher label than $j^{t}$ is necessarily in $S_1$. Therefore, $\alpha c_{l}\leq \alpha c_{j^{t}-1}$. Recall that $\lambda_{E_S}>\lambda_{E_{S_1}}\geq\lambda_{E_{(S_1\setminus\{i^*\})\cup \{j^{t}-1\}}}>\alpha c_{j^{t}-1}$. Hence,
$$\lambda_{E_S}>\lambda_{E_{S_1}}\geq\lambda_{E_{(S_1\setminus\{i^*\})\cup \{j^{t}-1\}}}>\alpha c_{j^{t}-1}\geq \alpha c_l,$$ which implies $\lambda_{E_S}>\alpha c_l$. Thus, $l\in R_S$.
\medskip

As a result of Part~1, all new agents in coalition $S$ are recipients, i.e., $S\setminus S_1\subseteq R_S$.\bigskip

\noindent\textbf{Part~2.} \emph{We show that each agent in $S\cap S_1$ who is a recipient in $S_1$ is also a recipient in $S$. Furthermore, the subsidy a recipient in $S\cap S_1$ receives in $S$ is larger than the subsidy she receives in $S_1$. That is, for each $l\in S \cap R_{S_1}$, we have $l\in R_{S}$ and $r_{l}^{S}> r_{l}^{S_1}$.}\smallskip

Let $l\in S \cap R_{S_1}$. Since $l\in R_{S_1}$, we have $\lambda_{E_{S_1}}\geq\alpha c_l$. Hence, $\lambda_{E_S}>\lambda_{E_{S_1}}$ implies $\lambda_{E_S} > \alpha c_l\geq 0$ and $l\in R_{S}$. Recall that since $l\in S\cap S_1$, $CEA_l(c_{S_1},E_{S_1})=\lambda_{E_{S_1}}\neq c_l$. Now, we distinguish two cases:\medskip

\noindent (i) If $CEA_l(c_{S},E_{S})= c_l$, then $r_{l}^{S}=CEA_l(c_S, E_S)-\alpha c_l=(1-\alpha) c_l>\lambda_{E_{S_1}}-\alpha c_l=CEA_l(c_{S_1}, E_{S_1})-\alpha c_l=r_{l}^{S_1}.$\medskip

\noindent (ii) If $CEA_l(c_{S},E_{S})= \lambda_{E_S}$, then $r_{l}^{S}=CEA_l(c_S, E_S)-\alpha c_l=\lambda_{E_S}-\alpha c_l>\lambda_{E_{S_1}}-\alpha c_l=CEA_l(c_{S_1}, E_{S_1})-\alpha c_l=r_{l}^{S_1}$.\medskip

Hence, for each $l\in S \cap R_{S_1}$,
\begin{equation}\label{eq4theorem}
r_{l}^{S}>r_{l}^{S_1}.
\end{equation}

As a result of Part~2, all recipients of $S_1$ who are in $S$ are recipients in $S$, i.e., $S\cap R_{S_1}\subseteq R_S$. Moreover, by Part 1, we know that $S\setminus S_1\subseteq R_S$. Thus, since $|S|=|S_1|$, we conclude that
\begin{equation}\label{eq2theorem}
|R_S|\geq |R_{S_1}|.
\end{equation}

\noindent\textbf{Part~3.} \emph{We show that each agent in $S\cap S_1$ who is a transfer agent in $S_1$ is either a transfer agent in $S$ transferring less in $S$ than in $S_1$ or a recipient in $S$. That is, for each $l\in S\cap T_{S_1}$, we have  either $[l\in T_S$ and $t_{l}^{S}< t_{l}^{S_1}]$ or $l\in R_{S}$. Furthermore, $T_S\subseteq S\cap T_{S_1}$.}\smallskip

Let $l\in S \cap T_{S_1}$. Since $l\in T_{S_1}$, we have $\lambda_{E_{S_1}}<\alpha c_l$. Hence, given that $\lambda_{E_S}>\lambda_{E_{S_1}}$, we distinguish two cases:\medskip

\noindent (i) If $\alpha c_l>\lambda_{E_S}$, then agent $l$ continues being a transfer agent, i.e., $l\in T_S$. Recall that since $l\in T_{S_1}$, $CEA_l(c_{S_1},E_{S_1})=\lambda_{E_{S_1}}\neq c_l$. Therefore,
\begin{equation}\label{eq3theorem}
t_{l}^{S}=\alpha c_l-CEA_l (c_S, E_S)=\alpha c_l-\lambda_{E_S}<\alpha c_l-\lambda_{E_{S_1}}=\alpha c_l-CEA_l (c_{S_1}, E_{S_1})= t_{l}^{S_1}.
\end{equation}
(ii) If $\alpha c_l\leq\lambda_{E_S}$, then agent $l$ becomes a recipient, i.e., $l\in R_{S}$.\medskip

By Part~1, $S\setminus S_1\subseteq R_S$. By Part~2, $S\cap R_{S_1}\subseteq R_S$. Hence, $T_S\subseteq S\cap T_{S_1}$. This, together with inequality~(\ref{eq3theorem}), implies
\begin{equation}\label{eq6theorem}
\sum_{l\in T_S}t_{l}^{S}\leq\sum_{l'\in T_{S_1}}t_{l'}^{S_1}.
\end{equation}
If $T_S\neq\emptyset$, then inequality~(\ref{eq6theorem}) is strict.\bigskip

\noindent\textbf{Part~4.} \emph{We show that for each agent in $S\setminus S_1$, the subsidy she receives in $S$ is larger than or equal to the maximal subsidy any agent in $S_1$ receives. That is, for each $l\in S\setminus S_1$, $r_l ^{S}\geq \max \{r_{l'}^{S_1}\}_{l'\in R_{S_{1}}}$.}\smallskip

Note first that, by construction of coalition $S_1$ by the $\theta$-CEA set algorithm, the maximal subsidy any recipient in $S_1$ receives is either $(1-\alpha)c_{i^*}$ or $\lambda_{E_{S_{1}}}-\alpha c_{j^*}$.\smallskip

Consider now the agent in $S\setminus S_{1}$ with the lowest subsidy, i.e., let $l^*\in  S\setminus S_{1}$ be such that for each $l\in  S\setminus S_{1}$, $r_{l^*}^{S}\leq r_{l}^{S}$. We show that $r_{l^*}^{S}$ is larger than or equal to $\max \{r_{l'}^{S_1}\}_{l'\in R_{S_{1}}}$. Recall that no agent receiving her claim in $S_1$ can be in $S$. In particular, neither agent $i^*$ nor any agent with a smaller label than $i^*$ can be in $S$. Furthermore, by construction of coalition $S_1$ by the $\theta$-CEA set algorithm, any agent with a higher label than $j^*$ is in $S_1$. Thus, $i^*<l^{*}<j^*$ and $c_{i^*}\leq c_{l^*}\leq c_{j^*}$.\smallskip

We distinguish three cases:\medskip

\noindent\textit{Case~$1$.} Agent $j^*$ does not exist.\smallskip

Then, agent $i^*$ has been chosen in the last step of the algorithm over agent $j^t-1$ (recall that $j^t$ is the transfer agent in $S_1$ with the lowest label and  $c_{i^*}\leq c_{j^t}$.) That is, $\max \{r_{l'}^{S_1}\}_{l'\in R_{S_{1}}}=(1-\alpha)c_{i^*}$ and by Step~$\theta-1$ (ii) of the $\theta$-CEA set algorithm, \begin{equation}
\label{eq:agenti}
(1-\alpha)c_{i^*}< \lambda_{E_{((S_1\setminus\{i^*\})\cup\{j^t-1\})}}-\alpha c_{j^t-1}.
\end{equation}
By construction of coalition $S_1$ by the $\theta$-CEA set algorithm, we have $\lambda_{E_{S_1}}\geq \lambda_{E_{((S_1\setminus\{i^*\})\cup\{j^t-1\})}}$ (this follows from Lemma~\ref{lemma:thresholdtheta} in  Appendix~\ref{appendix:conditionalgorithm}). We distinguish two cases:\smallskip

\noindent (i) If $CEA_{l^*} (c_S, E_S)=c_{l^*}$, then, since $c_{l^*}\geq c_{i^*}$ and $(1-\alpha)>0$,

$$r_{l^*}^{S}=(1-\alpha)c_{l^*}\geq (1-\alpha)c_{i^*}=\max \{r_{l'}^{S_1}\}_{l'\in R_{S_{1}}}.$$

\noindent (ii) If $CEA_{l^*} (c_S, E_S)=\lambda_{E_S}< c_{l^*}$, then, since $\lambda_{E_S}>\lambda_{E_{S_1}}\geq \lambda_{E_{((S_1\setminus\{i^*\})\cup\{j^t -1\})}}$ and $c_{l^*}\leq c_{{j^t} -1}$,
$$\lambda_{E_S} - \alpha c_{l^*}\geq \lambda_{E_S} -\alpha c_{{j^t} -1}> \lambda_{E_{S_1}} -\alpha c_{{j^t} -1}\geq \lambda_{E_{((S_1\setminus\{i^*\})\cup\{j^t -1\})}}-\alpha c_{j^t -1}.$$
This, together with inequality~(\ref{eq:agenti}), implies
$$r_{l^*}^{S}=\lambda_{E_S} - \alpha c_{l^*}>\lambda_{E_{((S_1\setminus\{i^*\})\cup\{j^t -1\})}}-\alpha c_{j^t -1}> (1-\alpha)c_{i^*}=\max \{r_{l'}^{S_1}\}_{l'\in R_{S_{1}}}.$$

\noindent\textit{Case~$2$.} Agent $i^*$ does not exist.\smallskip

Then, agent $j^*$ has been chosen in the last step of the algorithm over agent~$1$. That is, $\max \{r_{l'}^{S_1}\}_{l'\in R_{S_{1}}}=\lambda_{E_{S_1}}-\alpha c_{j^*}$ and by Step~$\theta-1$ (i) of the $\theta$-CEA set algorithm,
$$\lambda_{E_{S_1}}-\alpha c_{j^*}\leq (1-\alpha) c_{1}.$$
We distinguish two cases:\smallskip

\noindent (i) If $CEA_{l^*} (c_S, E_S)=c_{l^*}$, then, since $\lambda_{E_{S_1}}-\alpha c_{j^*}\leq (1-\alpha) c_{1},$ and $c_{1}\leq c_{l^*}$, $$r_{l^*}^{S}=(1-\alpha) c_{l^*}\geq (1-\alpha) c_{1}\geq \lambda_{E_{S_1}}-\alpha c_{j^*}=\max \{r_{l'}^{S_1}\}_{l'\in R_{S_{1}}}.$$

\noindent (ii) If $CEA_{l^*} (c_S, E_S)=\lambda_{E_S}< c_{l^*}$, then, since $c_{l^*}\leq c_{j^*}$ and $\lambda_{E_S}>\lambda_{E_{S_1}}$,
$$r_{l^*}^{S}=\lambda_{E_S}-\alpha c_{l^*}\geq\lambda_{E_S}-\alpha c_{j^*} > \lambda_{E_{S_1}}-\alpha c_{j^*}=\max \{r_{l'}^{S_1}\}_{l'\in R_{S_{1}}}.$$

\noindent\textit{Case~$3$.} Both agents $i^*$ and $j^*$ exist.\smallskip

Then, $\max \{r_{l'}^{S_1}\}_{l'\in R_{S_{1}}}=\max\{(1-\alpha)c_{i^*},\lambda_{E_{S_{1}}}-\alpha c_{j^*}\}.$\smallskip

Note that in Cases~$1$ and $2$ it is clear who the last agent added to coalition $S_1$ is and therefore, what the highest subsidy in set $S_1$ is. However, when both agents $i^*$ and $j^*$ exist, we need to distinguish cases depending on the last agent added to coalition $S_1$, who after all is the agent with the highest subsidy in $S_1$. We distinguish the following cases:\smallskip

\noindent (i) $CEA_{l^*} (c_S, E_S)=c_{l^*}$.\smallskip

\noindent (i.1) If agent $i^*$ is the last agent added to coalition $S_1$, then $\max \{r_{l'}^{S_1}\}_{l'\in R_{S_{1}}}=(1-\alpha)c_{i^*}$ and we apply Case~1~(i).\smallskip

\noindent (i.2) If agent $j^*$ is the last agent added to coalition $S_1$, then $\max \{r_{l'}^{S_1}\}_{l'\in R_{S_{1}}}=\lambda_{E_{S_{1}}}-\alpha c_{j^*}$ and we apply Case~2~(i) by replacing agent $1$ with agent $i^* +1$.\smallskip

\noindent (ii) $CEA_{l^*} (c_S, E_S)=\lambda_{E_S}< c_{l^*}$.
\smallskip

\noindent (ii.1) If agent $i^*$ is the last agent added to coalition $S_1$, then $\max \{r_{l'}^{S_1}\}_{l'\in R_{S_{1}}}=(1-\alpha)c_{i^*}$ and we apply Case~1~(ii) by replacing agent $j^t -1$ with agent $j^*-1$.\smallskip

\noindent (ii.2) If agent $j^*$ is the last agent added to coalition $S_1$, then $\max \{r_{l'}^{S_1}\}_{l'\in R_{S_{1}}}=\lambda_{E_{S_{1}}}-\alpha c_{j^*}$ and we apply Case~2~(ii).\medskip\pagebreak

Hence, we can conclude that $r_{l^*}^{S}\geq\max \{r_{l'}^{S_1}\}_{l'\in R_{S_{1}}}$. Furthermore, since for each $l\in S\setminus S_{1}$, $r_{l^*}^{S}\leq r_{l}^{S}$, we have that for each $l\in S\setminus S_{1}$,
 \begin{equation*}
r_{l}^{S}\geq\max \{r_{l'}^{S_1}\}_{l'\in R_{S_{1}}}.
\end{equation*}

Part 4, together with inequalities~(\ref{eq4theorem}) and (\ref{eq2theorem}), implies that the total subsidy in coalition $S$ is larger than or equal to the total subsidy in coalition $S_1$. That is,

\begin{equation}\label{eq5theorem}
\sum_{k\in R_S}r_{k}^{S}\geq\sum_{k'\in R_{S_1}}r_{k'}^{S_1}.
\end{equation}
If $S\cap R_{S_1}\neq \emptyset$, then by inequality~(\ref{eq4theorem}), inequality~(\ref{eq5theorem}) is strict. \medskip

Finally, inequalities~(\ref{eq6theorem}) and (\ref{eq5theorem}), together with the fact that in coalition $S_1$ the total transfer has to be equal to the total subsidy, imply that the total transfer in coalition $S$ is smaller than or equal to the total subsidy in coalition $S$. That is,
 $$\sum_{k\in R_S}r_{k}^{S}\stackrel{(\ref{eq5theorem})}{\geq}\sum_{k'\in R_{S_1}}r_{k'}^{S_1}=\sum_{l'\in T_{S_1}}t_{l'}^{S_1}\stackrel{(\ref{eq6theorem})}{\geq}\sum_{l\in T_S}t_{l}^{S}.$$

Note that if [$S\cap R_{S_1}\neq \emptyset$ or $T_S\neq \emptyset$], then [inequality~(\ref{eq5theorem}) or inequality~(\ref{eq6theorem}) is strict, respectively]. Thus, since the total transfer has to be equal to the total subsidy, a contradiction is reached.

Next, we show that it is impossible that $S\cap R_{S_1}= \emptyset$ and $T_S= \emptyset$. Assume, by contradiction, that  $S\cap R_{S_1}= \emptyset$ and $T_S= \emptyset$. On the one hand, since $S\cap R_{S_1}= \emptyset$, and given that $S\cap S_1\neq \emptyset$ ($S_1=R_{S_1}\cup T_{S_1}$), we have $S\cap T_{S_1}\neq\emptyset$. Thus, $T_{S_1}\neq \emptyset$ and therefore, $\sum_{k'\in R_{S_1}}r_{k'}^{S_1}>0$.
On the other hand, since $T_S=\emptyset$, we have $S=R_S$, and thus, $\sum_{k\in R_S}r_{k}^{S}=0$.
Hence, $$\sum_{k\in R_S}r_{k}^{S}=0<\sum_{k'\in R_{S_1}}r_{k'}^{S_1},$$ which contradicts inequality~(\ref{eq5theorem}).\medskip

In particular, there exists no coalition $T\subseteq N_1$ such that $|T|\leq \theta$, and for each $i\in S_1\cap T$, $T\succ^{((c,E),CEA)}_{i}S_1$, i.e., an agent from set $S_1$ cannot be part of a blocking coalition of size smaller than or equal to $\theta$ with agents from set $N_1$.\medskip\pagebreak

\noindent\textbf{Step~$\bm{k}$ ($\bm{k>1}$).} Recall from Step~$k-1$ that $N_k:=N\setminus\left(\cup_{j=1}^{k-1}S_{j}\right)$ and $|N_k|> \theta$. Furthermore, agents in $N_k$ are relabeled such that $N_k=\{1',\ldots,n'\}$, $c_{1'}\leq\ldots\leq c_{n'}$. Hence, by a similar reasoning than at Step~1 (with the new set of agents $N_k$), it follows that there exists no coalition $T\subseteq N_k$ such that $|T|\leq \theta$, and for each $i\in S_k\cap T$, $T\succ^{((c,E),CEA)}_{i}S_k$, i.e., an agent from set $S_k$ cannot be part of blocking coalition of size smaller than or equal to $\theta$ with agents from set $N_k$. By the previous steps, an agent from set $\bigcup_{j\in\{1,\ldots,k-1\}}S_j$ cannot be part of a blocking coalition with agents from set $S_k$. Hence, no agent from set $S_k$ can be part of a blocking coalition of size smaller than or equal to $\theta$.\medskip

After $l-1$ steps, we have shown that there is no blocking coalition $T\subsetneq N$ such that $|T|\leq \theta$.\medskip

Finally, assume, by contradiction, that $\pi$ is not stable. Then, there exists a blocking coalition $T\subseteq N$ such that for each agent $i\in T$, $T \succ^{((c,E),CEA)}_{i}\pi(i)$. In particular, $|T|> \theta$.\medskip

Then, by Proposition~\ref{prop:toptheta} (Appendix~\ref{appendix:proofTheorem6}), there is a coalition $T'\subsetneq T$ with $|T'|=\theta$ such that for each $j\in T'$, $$T'\succsim^{((c,E),CEA)}_{j}T\succ^{((c,E),CEA)}_{j}\pi(j),$$ which contradicts the fact that there is no blocking coalition of size smaller than or equal to $\theta$.\medskip

We have proven that partition $\pi$ obtained by the $\theta$-CEA algorithm is stable.\end{proof}

\section{Proof of Theorem~\ref{theorem:CELtheta}}\label{appendix:proofTheorem8}

Recall that $N=\{1,\dots, n\}$ and $c_1\leq c_2\leq \ldots\leq c_n$. For each $S\subseteq N$, we denote the CEL parameter associated with $(c_S,E_S)$ by $\lambda_{E_S}$, i.e., for each $i\in S$, $CEL_{i}(c_S,E_S)=\max\{0,c_i-\lambda_{E_S}\}$, where $\lambda_{E_S}$ is chosen so that $\sum_{j\in S}\max\{0,c_j-\lambda_{E_S}\}= E_S$.

\begin{proof}[\textbf{Proof of Theorem~\ref{theorem:CELtheta}}]                                                                                                                                                                                                                                                                                                                                                                                                                                                                                                                                                                                                                                                                                                                                                                                                                                                                                                                                                                                                                                                                                                                                                                                                                                                                                                                                                                                                                                                                                                                                                                                                                                                                                                                     Let $N\in{\cal N}$ such that $|N|>\theta$ and $(c,E)\in{\cal P}_{\theta}^{N}$. Consider the coalition formation problem with agent set $N$ induced by $((c,E),CEL)$. We show that the partition $\pi=\{S_1,\dots,S_l\}$ obtained by the $\theta$-CEL algorithm is stable.\medskip

First, we show that the coalitions obtained by the $\theta$-CEL algorithm cannot be blocked by coalitions of size smaller than or equal to $\theta$. \medskip

\noindent\textbf{Step~$\bm{1}$.} Recall that $N_1:=N$, $|N_1|>\theta$, and $S_1:=\{1,\ldots,\theta\}$. We show that $S_1$ is for each of its members at least as desirable as any other coalition of size smaller than or equal to $\theta$.\medskip

Note that since $(c,E)\in \mathcal{P}_{\theta}^{N}$, we have that for each $S\subsetneq N$ with $|S|<\theta$ and each $i\in S\cap S_1$, $$S_1 \succsim^{((c,E),CEL)}_{i} S.$$
Next, let $j\in S_1$ and $k\in N\setminus S_{1}$. Then, $k>\theta$ and $c_k\geq c_{\theta}$. We consider coalition $T=(S_{1}\setminus\{j\})\cup \{k\}$ and prove that for each $i\in S_{1}\setminus\{j\}$, $S_1\succsim^{((c,E),CEL)}_{i} T$ by showing that $$CEL_{i}(c_{S_{1}},E_{S_{1}})\geq CEL_{i}(c_{T},E_{T}).$$
If $c_k=c_j$, then $CEL_{i}(c_{S_{1}},E_{S_{1}})= CEL_{i}(c_{T},E_{T})$ follows immediately. Hence, assume that $c_k>c_j$. We consider two cases:\medskip

\noindent\textit{Case~$1$.} $j\in S_{1}$ is such that $CEL_j(c_{S_{1}}, E_{S_{1}})\neq 0$.\smallskip

\noindent Then, at problem $(c_{S_{1}}, E_{S_{1}})\in\mathcal{C}^{S_1}$, since agent $j$ contributes $\alpha c_j$ to and requests $c_j$ from coalition $S_1$, she contributes a loss of $(1-\alpha)c_j$ to coalition $S_1$.

Next, consider auxiliary problem $(c_{T}, \bar{E}_{T})\in\mathcal{C}^{T}$ at which agent $k$ also contributes a loss of $(1-\alpha)c_j$ to coalition $T$. Then, agent $k$ contributes $c_k- (1-\alpha)c_j$ to coalition $T$ and $\bar{E}_{T}=E_{S_{1}}-\alpha c_j+(c_k-c_j+\alpha c_j)=E_{S_{1}}+ (c_k-c_j)$. Note that auxiliary problem $(c_{T}, \bar{E}_{T})$ is such that $\bar{E}_{T}\neq \alpha c^{T}$. Since $c_k>c_j$, we have, by the definition of the CEL rule, that for each agent $i\in S_{1}\setminus\{j\}=T\setminus\{k\}$, $$CEL_{i}(c_{S_{1}}, E_{S_{1}})=CEL_{i}(c_{T}, \bar{E}_{T}).$$

Now, consider problem $(c_{T}, E_{T})\in \mathcal{C}^{T}$ and note that agent $k$ contributes a loss of $(1-\alpha)c_k$ to coalition $T$. Then, agent $k$ contributes $\alpha c_k$ to coalition $T$ and $E_{T}=E_{S_{1}}+\alpha (c_k- c_j)$.\medskip

Since $\alpha\in (0,1)$ and $c_k> c_j$,  $$\bar{E}_{T}=E_{S_{1}}+ (c_k-c_j)>E_{S_{1}}+\alpha (c_k- c_j)=E_{T}.$$
Then, by resource monotonicity of the CEL rule, for each agent $i\in S_{1}\setminus\{j\}= T\setminus\{k\}$,
$$CEL_{i}(c_{T},\bar{E}_{T})\geq CEL_{i}(c_{T},E_{T})$$
and hence,
$$CEL_{i}(c_{S_{1}},E_{S_{1}})\geq CEL_{i}(c_{T},E_{T}).$$

\noindent\textit{Case~$2$.} $j\in S_{1}$ is such that $CEL_j(c_{S_{1}}, E_{S_{1}})=0$.\smallskip

\noindent Starting from problem $(c_{S_{1}},E_{S_1})\in\mathcal{C}^{S_1}$, we consider the auxiliary problem that is obtained when agent $j$ leaves with her payoff $CEL_j(c_{S_{1}}, E_{S_{1}})=0$. Thus, the total endowment is now divided between agents in coalition $S_{1}\setminus\{j\}$ and the auxiliary problem $(c_{S_{1}\setminus\{j\}},E_{S_1})\in\mathcal{C}^{S_{1}\setminus\{j\}}$ is such that $E_{S_1}\neq\alpha c^{S_{1}\setminus\{j\}}$.
By consistency of the CEL rule, for each $i\in S_{1}\setminus\{j\}$,
\begin{equation}\label{CELtheta:eq1*}
CEL_{i}(c_{S_{1}\setminus\{j\}},E_{S_{1}})=CEL_{i}(c_{S_{1}},E_{S_{1}}).\end{equation}

Next, let $T=(S_1\setminus \{j\})\cup k$ and consider problem $(c_{T},E_{T})\in\mathcal{C}^{T}$. Since $E_{T}>0$ and $c_k\geq c_{\theta}$, by the definition of the CEL rule, $CEL_k (c_{T},E_{T})=c_k-\lambda_{E_T}>0$. Then, the amount of the endowment that is allocated to agents in $T\setminus \{k\}=S_{1}\setminus \{j\}$ is \begin{equation}\label{CELtheta:eq1}
\bar{E}_{T\setminus \{k\}}=E_T - (c_k-\lambda_{E_T})= \alpha c^{T\setminus \{k\}}+\alpha c_k -(c_k-\lambda_{E_T})=\alpha c^{T\setminus \{k\}}+(\alpha-1) c_k +\lambda_{E_T}.
\end{equation}
By the definition of the CEL rule, if an agent receives a zero payoff, then agents with the highest claims receive an over-proportional payoff. Hence, $c_k-\lambda_{E_T}>\alpha c_k$, or equivalently, \begin{equation}\label{CELtheta:eq2}
\lambda_{E_T}< (1-\alpha)c_k.
\end{equation}
Starting with equality (\ref{CELtheta:eq1}) and using inequality (\ref{CELtheta:eq2}), we obtain
\begin{equation}\label{CELtheta:eq3*}\bar{E}_{T\setminus \{k\}}=\alpha c^{T\setminus \{k\}}+(\alpha-1) c_k +\lambda_{E_T}< \alpha c^{T\setminus \{k\}}+(\alpha-1) c_k +(1-\alpha) c_k=\alpha c^{T\setminus \{k\}}.
\end{equation}

Now, starting from problem $(c_{T},E_{T})\in \mathcal{C}^{T}$, we consider the auxiliary problem that is obtained when agent $k$ leaves with her payoff $CEL_k(c_{T}, E_{T})$. The auxiliary problem $(c_{T\setminus\{k\}},\bar{E}_{T\setminus \{k\}})\in \mathcal{C}^{(T\setminus \{k\})}$ is such that $\bar{E}_{T\setminus \{k\}}\neq\alpha c^{T\setminus \{k\}}$. By consistency of the CEL rule, for each $i\in T\setminus\{k\}=S_{1}\setminus\{j\}$,
\begin{equation}\label{CELtheta:eq2*}
CEL_{i}(c_{T\setminus\{k\}},\bar{E}_{T\setminus \{k\}})=CEL_{i}(c_{T},E_{T}).
\end{equation}

We will now put the above arguments together. Since $S_1\setminus\{j\}=T\setminus\{k\}$ and $\alpha c_j\geq 0$, we have $E_{S_{1}} = \alpha c^{S_{1}}\geq\alpha c^{T\setminus \{k\}}$, and by inequality~(\ref{CELtheta:eq3*}), we have that $\alpha c^{T\setminus \{k\}}>\bar{E}_{T\setminus \{k\}}=\bar{E}_{S_1\setminus \{j\}}$. Hence, $E_{S_{1}}>\bar{E}_{T\setminus \{k\}}$. This, together with inequalities (\ref{CELtheta:eq1*}), (\ref{CELtheta:eq2*}), and resource monotonicity of the CEL rule, implies that for each $i\in S_{1}\setminus\{j\}=T\setminus\{k\}$,
\begin{equation*}
CEL_{i}(c_{S_{1}},E_{S_{1}})\overset{(\ref{CELtheta:eq1*})}{=}CEL_{i}(c_{S_{1}\setminus\{j\}},E_{S_{1}})\geq CEL_{i}(c_{T\setminus\{k\}},\bar{E}_{T\setminus \{k\}})\overset{(\ref{CELtheta:eq2*})}{=}CEL_{i}(c_{T},E_{T}).
\end{equation*}

Note that so far we have only shown that any agent $i\in S_1$ weakly prefers coalition $S_1$ to a coalition $T$ where another agent was replaced by an agent outside of $S_1$. However, this proof step can be repeated by replacing other agents one by one and hence, for each agent $i\in S_1$ and each coalition  $S\subsetneq N_1$ such that [$i\in S$ and $|S|=\theta$], $$S_1 \succsim^{((c,E),CEL)}_{i} S.$$

In particular, there exists no coalition $T\subseteq N_1$ such that $|T|\leq \theta$, and for each $i\in S_1\cap T$, $T\succ^{((c,E),CEL)}_{i}S_1$, i.e., an agent from set $S_1$ cannot be part of a blocking coalition of size smaller than or equal to $\theta$ with agents from set $N_1$.
\medskip

\noindent\textbf{Step~$\bm{k}$ ($\bm{k>1}$).} Recall from Step~$k-1$ that $N_k:=N\setminus\left(\cup_{j=1}^{k-1}S_{j}\right)$ and $|N_k|> \theta$.
Set $S_k:=\{\theta k- (\theta-1), \theta k- (\theta-2), \ldots, \theta k\}$. Note that agents $\theta k- (\theta-1), \ldots, \theta k$ are lowest-claims agents in $N_{k}$. Hence, by a similar reasoning than at Step~1 (with agents $\theta k- (\theta-1), \ldots, \theta k$ in the roles of $1,\ldots,\theta$ respectively), it follows that there exists no coalition $T\subseteq N_k$ such that $|T|\leq \theta$, and for each $i\in S_k\cap T$, $T\succ^{((c,E),CEL)}_{i}S_k$, i.e., an agent from set $S_k$ cannot be part of a blocking coalition of size smaller than or equal to $\theta$ with agents from set $N_k$. By the previous steps, an agent from set $\bigcup_{j\in\{1,\ldots,k-1\}}S_j$ cannot be part of a blocking coalition with an agent from set $S_k$. Hence, no agent from set $S_k$ can be part of a blocking coalition of size smaller than or equal to $\theta$.\medskip

After $l-1$ steps, we have shown that there is no blocking coalition $T\subsetneq N$ such that $|T|\leq \theta$.\medskip

Finally, assume, by contradiction, that $\pi$ is not stable. Then, there exists a blocking coalition $T\subseteq N$ such that for each agent $i\in T$, $T \succ^{((c,E),CEL)}_{i}\pi(i)$. In particular, $|T|> \theta$.\medskip

Then, by Proposition~\ref{prop:toptheta} (Appendix~\ref{appendix:proofTheorem6}), there is a coalition $T'\subsetneq T$ with $|T'|=\theta$ such that for each $j\in T'$, $$T'\succsim^{((c,E),CEL)}_{j}T\succ^{((c,E),CEL)}_{j}\pi(j),$$ which contradicts the fact that there is no blocking coalition of size smaller than or equal to $\theta$.\medskip

We have proven that the partition $\pi$ obtained by the $\theta$-CEL algorithm is stable.\end{proof}

\section{Results when agents have single-peaked preferences}\label{appendix:surplus}

Let $N\in\mathcal{N}$ be a finite set of agents. Each agent $i\in N$ has \textbf{single-peaked preferences} $R_i$ on $\mathbb{R}_{+}$, i.e., there is a point $p(R_i)\in \mathbb{R}_{+}$ such that for each $x,y\in \mathbb{R}_{+}$, if either $y<x\leq p(R_i)$ or $p(R_i)\leq x<y$, then $x\mathbin{P_i} y$, where $P_i$ denote the strict preference relation associated with $R_i.$ The point $p(R_i)$ is called the \textbf{peak} of $R_i.$ Let $\mathcal{R}$ be the set of single-peaked preference relations. A \textbf{single-peaked problem with coalition} $\bm{N}$  is a pair $(R,E)\in \mathcal{R}^{N}\times \mathbb{R}_{+}$ where $R\in \mathcal{R}^N$ is the \textbf{preference profile} and $E\in \mathbb{R}_{+}$ is the \textbf{endowment} of coalition $N$. Let $\bm{\mathcal{SP}^{N}}$ denote the class of such problems and  $\bm{\mathcal{SP}}\equiv\bigcup_{N\in\mathcal{N}} \mathcal{SP}^{N}$. \medskip

An \textbf{allocation for} $\bm{(R,E)}\in \mathcal{SP}^{N}$ is a vector $x=(x_{i})_{i\in N}\in \mathbb{R}_{+}^{N}$ that satisfies $\sum_{i\in N}x_{i}=E$. A \textbf{rule} is a function $F$ defined on $\mathcal{SP}$ that associates with each $N\in\mathcal{N}$ and each $(R,E)\in \mathcal{SP}^{N}$ an allocation for $(R,E)$.
The \textbf{uniform rule} allocates the endowment as equally as possible by taking agents' peaks into account as follows: for excess demand problems, agents' peaks  are upper bounds and, for excess supply problems, agents' peaks are lower bounds. The uniform rule extends the CEA rule from excess demand problems (where $\sum_{i\in N}p(R_i) \leq E$) to excess supply problems (where $\sum_{i\in N}p(R_i) \geq E$). The uniform rule $U$ is defined as follows: for each $(R,E)\in \mathcal{SP}^{N}$,
$$U_i(R,E)= \left\{\
\begin{aligned}
  \min\,\{p(R_i),\lambda\} & \quad\text{ if }\sum_{i\in N}p(R_i) \geq E\quad \text{(excess demand)} \text{ and}\\
  \max\,\{p(R_i),\lambda\} & \quad\text{ if }\sum_{i\in N}p(R_i) \leq E \quad \text{(excess supply)},
\end{aligned}\right.
$$
where $\lambda$ is chosen so that $\sum_{i\in N}U_i(R,E)=E$. Note that for excess demand problems, $U=CEA$.\medskip

For excess supply problems, the \textbf{equal surplus rule} allocates the total surplus induced by the endowment  equally among agents.  The equal surplus rule extends the CEL rule from excess demand problems (where $\sum_{i\in N}p(R_i) \leq E$) to excess supply problems (where $\sum_{i\in N}p(R_i) \geq E$). The equal surplus rule $ES$ is defined as follows: for each $(R,E)\in \mathcal{SP}^{N}$ such that $\sum_{i\in N}p(R_i) \leq E$ (excess supply),
$$ES_i(R,E)= p(R_i)+\frac{E-\sum_{j\in N}p(R_j)}{|N|}.$$

Next, we generalize the notion of a single-peaked problem. Consider  $N\in\mathcal{N}$ and  $(R,E)\in\mathcal{SP}^{N}$. Then, each coalition of agents $S \subseteq N$ has the \textbf{reduced preference profile} $\bm{R_{S}}=(R_i)_{i\in S}$. Next, assume that for each coalition $S \subseteq N$  there is a \textbf{coalitional endowment} $\bm{E_{S}}$ such that $(R_S,E_{S})\in\mathcal{SP}^{S}$ and $E_{N}=E$. Formally, given $N\in\mathcal{N}$, a \textbf{generalized single-peaked problem with coalition} $\bm{N}$ is a pair $(R,(E_{S})_{S\subseteq N})\in \mathcal{R}^{N}\times \mathbb{R}_{+}^{2^{|N|}-1}$, such that for each coalition $S\subseteq N$, $(R_{S},E_{S})\in\mathcal{SP}^{S}$. Let $\bm{\mathcal{GSP}^{N}}$ denote the class of such problems and $\bm{\mathcal{GSP}}\equiv\bigcup_{N\in\mathcal{N}} \mathcal{GSP}^{N}$. \medskip

An \textbf{allocation configuration for} $\bm{(R,(E_{S})_{S\subseteq N})}\in\mathcal{GSP}^{N}$ is a list $(x_{S})_{S\subseteq N}$ such that for each $S\subseteq N$, $x_{S}$ is an allocation for the single-peaked problem derived from $(R,(E_{S})_{S\subseteq N})$ for coalition $S$, $(R_S,E_{S})$.\medskip

Now, starting from a generalized single-peaked problem (instead of a generalized claims problem) and a rule, a coalition formation problem is induced. The following example illustrates this for the uniform rule and shows that a stable partition need not exist.

\begin{example}[\textbf{CEA / uniform rule: excess demand / supply}]\normalfont
Consider $N=\{1,2,3\}$ and a preference profile $R\in\mathcal{R}$ such that $p(R)=(2,4,5)$ and agents' preferences are symmetric with respect to their peaks, i.e., for each agent $i\in N$, the distance to the peak determines her preferences. Let $(R,(E_{S})_{S\subseteq N})\in \mathcal{GSP}^{N}$ such that the coalitional endowments are given by Table~\ref{tableEx4endowments}.

\begin{table}[th]
\centering
\renewcommand{\arraystretch}{1.5}
\begin{tabular}{c||c|c|c|c|c}

Coalition & $\{1\},\ \{2\},\ \{3\}$ & $\{1,2\}$ & $\{1,3\}$ & $\{2,3\}$ & $\{1,2,3\}$
\\ \hline\hline
Endowment & $0$  & $7$ & $6$ & $11$ & $21$ \\
\end{tabular}%
\caption{Coalitions and their endowments.}
\label{tableEx4endowments}
\end{table}

For instance, coalition $\{1,2\}$, which claims $6$, gets an endowment of $7$ and has excess supply. However, coalition $\{1,3\}$, which claims $7$, gets only $6$ and has excess demand. The individual payoffs of agents in each possible coalition are listed in Table~\ref{Table:uniformpayoffs}.

\begin{table}[th]
\centering
\renewcommand{\arraystretch}{1.5}
\begin{tabular}{c||c|c|c|c|c}

Coalition & $\{1\},\ \{2\},\ \{3\}$ & $\{1,2\}$ & $\{1,3\}$ & $\{2,3\}$ & $\{1,2,3\}$
\\ \hline\hline
 CEA / uniform rule payoffs & $(0)$ & $(3,4)$ & $(2,4)$ & $(5.5,5.5)$ & $(7,7,7)$  \\
\end{tabular}%
\caption{Agents' coalitional payoffs induced by the CEA / uniform rule.}\label{Table:uniformpayoffs}
\end{table}

Given that each agent's preferences are symmetric and only depend on the distance of the payoffs from the peak, agents have the following preferences over coalitions.
\begin{align*}
\succsim_1^{CEA/U}:&\quad\{1,3\}\succ \{1,2\}\succ\{1\}\succ\{1,2,3\},\\
\succsim_2^{CEA/U}:&\quad\{1,2\}\succ \{2,3\}\succ\{1,2,3\}\succ\{2\},\\
\succsim_3^{CEA/U}:&\quad\{2,3\}\succ\{1,3\}\succ\{1,2,3\}\succ\{3\}.
\end{align*}
Observe that no stable partition exists for these preferences over coalitions.\label{example:uniformrule}\hfill $\diamond$
\end{example}

Similarly, the following example shows that a stable partition need not exist for the CEL / equal surplus rule.

\begin{example}[\textbf{CEL / equal surplus rule: excess demand / supply}]\normalfont
Consider $N=\{1,2,3\}$ and a preference profile $R\in\mathcal{R}$ such that $p(R)=(2,7,18)$ and agents' preferences are symmetric with respect to their peaks, i.e., for each agent $i\in N$, the distance to the peak determines her preferences. Let $(R,(E_{S})_{S\subseteq N})\in \mathcal{GSP}^{N}$ such that the coalitional endowments are given by Table~\ref{tableEx5endowments}.

\begin{table}[th]
\centering
\renewcommand{\arraystretch}{1.5}
\begin{tabular}{c||c|c|c|c|c|c|c}

Coalition & $\{1\}$ & $\{2\}$& $\{3\}$ & $\{1,2\}$ & $\{1,3\}$ & $\{2,3\}$ & $\{1,2,3\}$
\\ \hline\hline
Endowment & $7$ & $0$ & $0$ & $15$ & $10$ & $13$ & $54$   \\
\end{tabular}%
\caption{Coalitions and their endowments.}
\label{tableEx5endowments}
\end{table}

For instance, coalition $\{1,2\}$, which claims $9$, gets an endowment of $15$ and has excess supply. However, coalition $\{2,3\}$, which claims $25$, gets only $13$ and has excess demand. The individual payoffs of agents in each possible coalition are listed in Table~\ref{Table:equalsurpluspayoffs}.

\begin{table}[th]
\centering
\renewcommand{\arraystretch}{1.5}
\begin{tabular}{c||c|c|c|c|c|c|c}

Coalition & $\{1\}$& $\{2\}$& $\{3\}$ & $\{1,2\}$ & $\{1,3\}$ & $\{2,3\}$ & $\{1,2,3\}$
\\ \hline\hline
 CEL / equal surplus rule payoffs & $(7)$ & $(0)$ & $(0)$& $(5,10)$ & $(0,10)$ & $(1,12)$ & $(11,16,27)$ \\
\end{tabular}%
\caption{Agents' coalitional payoffs induced by the CEL / equal surplus rule.}\label{Table:equalsurpluspayoffs}
\end{table}

Given that each agent's preferences are symmetric and only depend on the distance of the payoffs from the peak, agents have the following preferences over coalitions.
\begin{align*}
\succsim_1^{CEL/ES}:&\quad\{1,3\}\succ \{1,2\}\succ\{1\}\succ\{1,2,3\},\\
\succsim_2^{CEL/ES}:&\quad\{1,2\}\succ \{2,3\}\succ\{2\}\succ\{1,2,3\},\\
\succsim_3^{CEL/ES}:&\quad\{2,3\}\succ\{1,3\}\succ\{1,2,3\}\succ\{3\}.
\end{align*}
Observe that no stable partition exists for these preferences over coalitions.\label{example:equalsurplusrule}\hfill $\diamond$
\end{example}

Similarly as for generalized claims problems, we consider a subclass of generalized single-peaked problems. First, for each coalition $S\subseteq N$, we define a \textbf{coalitional peak}, $\bm{p(R^{S})}$, that is equal to the sum of the peaks of the members of the coalition, i.e., $p(R^{S}):=\sum_{j\in S}p(R_j)$. Then, given $(R,E)\in\mathcal{SP}^{N}$ and $\alpha:=\frac{E}{p(R)}\in \mathbb{R}_{++}$, a $\bm{\theta}$\textbf{-minimal proportional generalized single-peaked problem} is a tuple $(R,(E_{S})_{S\subseteq N})$ such that, for each coalition $S\subseteq N$, $S\neq\emptyset$, with $|S|<\theta$, $E_{S}=0$, and for each coalition $S\subseteq N$ with $|S|\geq \theta$, $E_{S}=\alpha p(R^{S})$. Let $\bm{\mathcal{PSP}_{\theta}^{N}}$ denote the class of such problems and $\bm{\mathcal{PSP}_{\theta}}\equiv\bigcup_{N\in\mathcal{N}} \mathcal{PSP}_{\theta}^{N}$. Since  coalitional endowments $E_{S}$, $S\subseteq N$, are completely determined by $\theta$, $R$, and $E$, we will simplify notation and denote a $\theta$-minimal proportional generalized single-peaked problem $(R,(E_{S})_{S\subseteq N})\in \mathcal{PSP}_{\theta}^N$ by $(R,E)\in \mathcal{PSP}_{\theta}^N$.\medskip

For the class of $\theta$-minimal proportional generalized single-peaked problems with excess demand, all results for the CEA and the CEL rule translate without adjustments.
We next show how algorithmic results for the CEA and the CEL rule extend from excess demand to excess supply.

\subsection*{Excess supply: Stability under the uniform rule}

\noindent\textbf{Uniform algorithm for $\bm{\theta=2}$}\medskip

\noindent\textbf{Input:} $N\in\mathcal{N}$ such that $|N|>2$ and $(R,E)\in \mathcal{PSP}_{2}^N$ such that $p(R_1)\leq p(R_2)\leq \cdots\leq p(R_{|N|})$.\medskip

\noindent\textbf{Step~$\bm{1}$.} Let $N_1:=N$. Set $S_1:=\{1,2\}$ and $N_2:=N\setminus S_1$. If $|N_2|\leq 2$, then set $S_2:=N_2$, define $\pi:=\{S_1,S_2\}$, and stop. Otherwise, go to Step~2.\medskip

\noindent\textbf{Step~$\bm{k}$ ($\bm{k>1}$).} Recall from Step~$k-1$ that $N_k:=N\setminus\left(\cup_{j=1}^{k-1}S_{j}\right)$ and $|N_k|> 2$.
Set $S_k:=\{2k-1,2k\}$ and $N_{k+1}:=N\setminus \cup_{j=1}^{k}S_j$. If $|N_{k+1}|\leq 2$, then set $S_{k+1}:=N_{k+1}$, define $\pi:=\{S_1,\ldots,S_{k+1}\}$, and stop. Otherwise, go to Step~$k+1$.\medskip

\noindent\textbf{Output:} A partition $\pi=\{S_1,\ldots,S_l\}$ for the coalition formation problem with agent set $N$ induced by $((R,E),U)$ such that for each $k\in\{1,\dots,l-1\}$, $|S_{k}|= 2$ and $|S_{l}|\leq 2$. If $|N|$ is even, then partition $\pi$ is constructed in $l-1=\frac{n-2}{2}$ steps. If $|N|$ is odd, then partition $\pi$ is constructed in $l-1=\frac{n-1}{2}$ steps.\medskip

Note that the uniform algorithm forms coalitions by sequentially pairing up agents with lowest peaks, which results in sequentially lowest per capita endowments. The uniform rule allocates the coalitional endowments as equally as possible, taking the agents' peaks as lower bounds. If an agents receives her peak at the partition obtained by the uniform algorithm, then she cannot block. Hence, only agents who don't receive their peak could potentially block. However, note that in such a blocking pair (since the uniform algorithm sequentially pairs agents with lowest per capita endowments), the lower (or possibly equal) peak agent in the blocking pair would have a higher per capita endowment and receive a larger uniform payoff as before, making her worse off. Thus, the partition obtained by the uniform algorithm is stable. We state this result without a formal proof.

\begin{theorem}\label{UG}
Let  $\theta=2$, $N\in{\cal N}$ such that $|N|>2$, and $(R,E)\in \mathcal{PSP}_{2}^N$ with $\alpha>1$. Consider the coalition formation problem with agent set $N$ induced by $((R,E), U)$. Then, the partition obtained by the uniform algorithm is stable.
\end{theorem}

Note that the uniform algorithm can easily be extended to $\theta >2$ by sequentially matching agents with lowest peaks into $\theta$-size coalitions (with possibly one coalition of largest peak agents of size smaller than $\theta$). It follows similarly as before that the resulting partition is stable.

\subsection*{Excess supply: Stability under the equal surplus rule}

\noindent\textbf{Equal surplus algorithm for $\bm{\theta=2}$}\medskip

\noindent\textbf{Input:} $N\in\mathcal{N}$ such that $|N|>2$ and $(R,E)\in \mathcal{PSP}_{2}^N$ such that $p(R_1)\leq p(R_2)\leq \cdots\leq p(R_{|N|})$.\medskip

\noindent\textbf{Step~$\bm{1}$.} Let $N_1:=N$. Set $S_1:=\{1,2\}$ and $N_2:=N\setminus S_1$. If $|N_2|\leq 2$, then set $S_2:=N_2$, define $\pi:=\{S_1,S_2\}$, and stop. Otherwise, go to Step~2.\medskip

\noindent\textbf{Step~$\bm{k}$ ($\bm{k>1}$).} Recall from Step~$k-1$ that $N_k:=N\setminus\left(\cup_{j=1}^{k-1}S_{j}\right)$ and $|N_k|> 2$.
Set $S_k:=\{2k-1,2k\}$ and $N_{k+1}:=N\setminus \cup_{j=1}^{k}S_j$. If $|N_{k+1}|\leq 2$, then set $S_{k+1}:=N_{k+1}$, define $\pi:=\{S_1,\ldots,S_{k+1}\}$, and stop. Otherwise, go to Step~$k+1$.\medskip

\noindent\textbf{Output:} A partition $\pi=\{S_1,\ldots,S_l\}$ for the coalition formation problem with agent set $N$ induced by $((R,E), ES)$ such that for each $k\in\{1,\dots,l-1\}$, $|S_{k}|= 2$ and $|S_{l}|\leq 2$. If $|N|$ is even, then partition $\pi$ is constructed in $l-1=\frac{n-2}{2}$ steps. If $|N|$ is odd, then partition $\pi$ is constructed in $l-1=\frac{n-1}{2}$ steps.\medskip

Note that the equal surplus algorithm forms coalitions by sequentially pairing up agents with lowest peaks, which results in sequentially lowest per capita surpluses. The equal surplus rule adds these per capita surpluses to agents' peaks. This immediately implies that the partitions obtained by the equal surplus algorithm cannot be strictly blocked by another partition (which would need to have smaller per capita surpluses to strictly block). Thus, the partition obtained by the equal surplus algorithm is stable. We state this result without a formal proof.

\begin{theorem}\label{ES}
Let $\theta=2$, $N\in{\cal N}$ such that $|N|>2$, and $(R,E)\in \mathcal{PSP}_{2}^N$ with $\alpha>1$. Consider the coalition formation problem with agent set $N$ induced by $((R,E),ES)$. Then, the partition obtained by the equal surplus algorithm is stable.
\end{theorem}

Note that the equal surplus algorithm can easily be extended to $\theta >2$ by sequentially matching agents with lowest peaks into $\theta$-size coalitions (with possibly one coalition of largest peak agents of size smaller than $\theta$). It follows similarly as before that the resulting partition is stable.

\begin{remark} \textbf{\emph{(Excess supply under monotonic preferences)}} \normalfont Let us consider a $\theta$-minimal proportional generalized claims problem with $\alpha>1$, i.e., the excess supply case of our model. Furthermore, let us assume that agents have monotonic preferences, i.e., the larger the payoff, the better. Then, similarly to the model with single-peaked preferences, the CEA and CEL rules are extended to the uniform rule and the equal surplus rule, respectively (in the definition of both rules, peaks are replaced by claims). By a similar reasoning to the model with single-peaked preferences, both the uniform and the equal surplus algorithms give rise to a stable partition formed by
sequentially matching agents with highest claims into $\theta$-size coalitions (with possibly one coalition of lowest-claims agents of size smaller than $\theta$). The difference to our results with single-peaked preferences (where smaller payoffs towards the peaks are better) is that now positively assortative coalitions are formed starting with high-claim agents instead of with low-peak agents.\end{remark}

\end{appendix}



\end{document}